\documentclass[11pt]{book} 

\usepackage{graphicx}
\usepackage{amssymb}
\usepackage{amsmath}
\usepackage{latexsym}
\usepackage{amsthm}
\usepackage{euscript}
\usepackage{mathrsfs}
\usepackage{bbm}
\usepackage{stmaryrd}
\usepackage{graphicx}
\newcommand{\testfunc}[1]{C^{\infty}_{0}(#1)}	%test functions supported in #1 
\newcommand{\testfuncv}[2]{C^{\infty}_{0}(#1,#2)}%test functions supported in #1 valued in #2 
\newcommand{\smoothfunc}[1]{C^{\infty}(#1)} 	%smooth functions defined in #1 
\newcommand{\smoothfuncv}[2]{C^{\infty}(#1,#2)} %smooth functions defined in #1 valued in #2 
\newcommand{\supp}[1]{\operatorname{supp}(#1)}		%support
\newcommand{\Dom}[1]{\operatorname{Dom}(#1)}		%Domain
\newcommand{\Ran}[1]{\operatorname{Ran}(#1)}		%Range

\newcommand{\bop}[1]{\mathbb B(#1)}		%bounded operators on the Hilbert space #1
\newcommand{\hs}[1]{\mathcal #1}		%Hilbert space 
\newcommand{\varhs}[1]{\mathfrak #1}		%Hilbert space (typographical variant)
\newcommand{\distrib}[1]{\mathcal{D}^{\prime}(#1)}		%Hilbert space (typographical variant)

\newcommand{\norm}[1]{\left\|#1\right\|}	%norm 
\newcommand{\sqint}[2]{L^{2}(#1,d#2)}		%L2-space
\newcommand{\csqint}[2]{L^{2}(#1,d#2)_{\bC}}	%Complex L2-space
\newcommand{\ip}[2]{\langle #1, #2\rangle}	%inner product
\newcommand{\fock}[1]{\mathfrak{F}^{+}(#1)}   	%(bosonic) Fock space

\newcommand{\Ch}{\operatorname{Ch}}		%Characters

\newcommand{\cl}[1]{\mbox{cl}(#1)}		%topological closure
\newcommand{\open}[1]{ #1}			%open set
\newcommand{\Int}[1]{\operatorname{Int}(#1)}		%Interior

\newcommand{\spacetime}[1]{\mathbb #1} 		%spacetime 
\newcommand{\sdc}[1]{\mathscr #1}		%Cauchy surface 
\newcommand{\ccompl}[1]{{#1}^{\perp}}		%causal complement
\newcommand{\cfuture}[1]{J^{+}(#1)} 		%causal future
\newcommand{\cpast}[1]{J^{-}(#1)} 		%causal past
\newcommand{\cdevel}[1]{J(#1)} 			%causal development
\newcommand{\tfuture}[1]{I^{+}(#1)} 		%cronological future
\newcommand{\tpast}[1]{I^{-}(#1)} 		%cronological past
\newcommand{\Cdata}[1]{D_{#1}}		%Cauchy data space
\newcommand{\Cdatav}[1]{\testfuncv{#1}{\bR}\oplus\testfuncv{#1}{\bR}}		%Cauchy data space (explicit)

\newcommand{\DIAMOND}[1]{\Diamond(#1)}		%diamond based on #1

\newcommand{\alg}[1]{\mathcal #1} 			%generic C*-algebra
\newcommand{\sympls}[1]{\mathsf #1}		%symplectic space 
\newcommand{\sympl}[1]{\mathsf K(#1)}		%symplectic space ``supported'' in #1
\newcommand{\sweyl}[1]{\mathcal W(#1)}		%spatial (Weyl) net of local algebras (evaluated in #2)
\newcommand{\stweyl}[1]{\mathcal W[#1]}		%spacetime (Weyl) net of local algebras (evaluated in #2)
\newcommand{\weylalg}[2]{\alg{A}[#1, #2]}
\newcommand{\net}[2]{{\mathcal #1}_{\poset{#2}}}	%net of local algebras on poset #2
\newcommand{\snet}[2]{\mathcal #1(#2)}		%spatial net of local (von Neumann) algebras (evaluated in #2)
\newcommand{\stnet}[2]{\mathcal #1[#2]}		%spacetime net of local (von Neumann) algebras (evaluated in #2)

\newcommand{\poset}[1]{\mathscr #1}    		%poset
\newcommand{\coc}[1]{#1}			%1-coycle 
\newcommand{\cocyclecat}[1]{{\mathcal Z}^{1}(#1)} 	%(localized) $1$-cocycle category over the local net %# 1
\newcommand{\ucocyclecat}[2]{{\mathcal Z}^{1}(\poset{#1}, \bop{#2})}	%unlocalized $1$-cocycle category over the local net %# 1 (TOCHANGE!!)
\newcommand{\tcocyclecat}[1]{{\mathcal Z}^{1}_{t}(#1)}	%trivial $1$-cocycle category over the local net %# 1
\newcommand{\itwiner}[1]{#1}			%intertwiner
\newcommand{\simplex}[2]{\Sigma_{#1}(#2)} 	%set of #n-simplices 
\newcommand{\pathset}[2]{\mathsf P(#1, #2)}	%set of path (poset)
\newcommand{\cpathset}[1]{\mathsf P(#1)}	%set of path (poset)
\newcommand{\path}[1]{\mathsf #1}		%poset path
\newcommand{\homgroup}[2]{\pi_{1}\!\left(\poset{#1}, #2\right) }		%first homotopy group of a poset
\newcommand{\shomgroup}[1]{\pi_{1}\!\left(\poset{#1}\right) }		%first homotopy group of a poset (base independent)
\newcommand{\stdhomgroup}[1]{\pi_{1}\!\left(#1\right) }		%first homotopy group of a topological space 
\newcommand{\app}[1]{\mbox{App}(#1)}				%poset  approximations of a curve
\newcommand{\indexset}[1]{\mathscr{I}(#1,\perp)}				%poset  approximations of a curve
\newcommand{\category}[1]{\mathfrak{#1}}			%category

\newtheorem{proposition}{Proposition}
\newtheorem{theorem}{Theorem}
\newtheorem{lemma}{Lemma}
\newtheorem{definition}{Definition}
\newtheorem{corollary}{Corollary}

\def\bC{{\mathbb C}}           %%%  complex numbers and so on

\def\bI{{\mathbb I}}

\def\bN{{\mathbb N}}

\def\bR{{\mathbb R}}

\def\bS{{\mathbb S}}

\def\bZ{{\mathbb Z}}

\def\cpropagator{E}
\def\KGsol{\mathcal S}
\def\KGop{K}
\def\rest{\upharpoonright}
\def\Ch{\operatorname{Ch}}

\begin{document}
%Frontespizio
\begin{titlepage}
\begin{figure} [h]
\begin{center}
\includegraphics{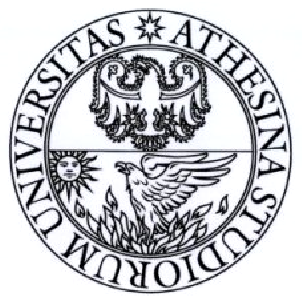}
\bigskip

{\bf{\Large UNIVERSIT\`A DEGLI STUDI DI TRENTO}}\\
\large Facolt\`a di Scienze Matematiche, Fisiche e Naturali
\end{center}
\end{figure}

\
\vspace{0.5 cm}
\begin{center}
\large Ph.D. Thesis
\end{center}
\
\vspace{1 cm}
\begin{center}
{\bf{\huge Topological sectors for Weyl-algebra net in the Einstein cylindrical universe}}
\end{center}
\
\vspace{1.5 cm}
\begin{flushleft}
\large Advisor: \hspace{\stretch{2}} \large Ph.D. Student:\\
\large Prof. Valter Moretti  \hspace{\stretch{1}} \large Dott. Lorenzo Franceschini\\
\end{flushleft}
\
\
\vspace{1,5 cm}
\begin{center}
\large January, 09 - 2009
\end{center} 
\end{titlepage}

\pagebreak

%\maketitle
\tableofcontents
\newpage
\section*{Introduction}

A lot of research has be done in physics about the influence of topology on physical theories; to make just a few examples, we can cite the Aharonov-Bohm effect and the Berry phase in quantum mechanics \cite{Aharonov_Bohm_59}, and its classical counterpart called Hannay's angle \cite{Hannay_85}.
This work deals with a kind of topological effect arising in quantum field theory (QFT) on curved spacetime.

The underlying conceptual framework is the algebraic approach to QFT, also know as \emph{local quantum theory} \cite{Haag_96} where, as is well know, the physical content of the theory (observables) is encoded by a net of $C^{*}$-algebras $\stnet{A}{\open{O}$} indexed by a family (poset) of bounded spacetime regions $\open{O}\subset\spacetime{M}$; states are implemented, via the GNS theorem, by representations of these local abstract algebras as (concrete) operatorial algebras on a fixed Hilbert space. All the relevant physical information is contained in the inclusion relations of these concrete algebras for different spacetime regions.

So, given the importance of representations of local algebras, it's no surprise that considerable efforts have beeen spent in selecting physical meaningful representations of local observable algebras, that is to say physical interesting states. The way this choice is done goes by the name of \emph{selection criterion}, and a unitary equivalence class of representations satisfying the criterion is denoted as a \emph{superselection sector}.

A well known and important example of selection criterion is that introduced by Doplicher, Haag and Roberts \cite{DHR_71}, and the corresponding equivalence class is known as \emph{DHR superselection sector}; DHR superselection sectors describe charge localized in bounded regions. DHR theory has been developed on Minkowski spacetime and extended to simply connected globally hyperbolic spacetimes; however, it turns out that the extension to multiple-connected spacetimes is not straightforward, and it requires substantial changes to be of use. These changes involve modifying the selection criterion and generalizing the very notion of representation of a local algebra net; roughly speaking these generalized representations, named \emph{(unitary) net-representations}, are local version of standard representations, in a sense that will be made more precise later. Heuristically, net-representations are analogous to local coordinate charts on a manifold, while conventional representations are like a global coordinate chart; this time, although, the coordinate space is infinite-dimensional, being a subset of some $\bop{\hs{H}}$.

It can be shown \cite{Brunetti_Ruzzi_08} that these new notions of representation and selection criterion genuinely extend the old ones, reducing to them when specialized to simple-connected globally hyperbolic spacetimes; however, when the background spacetime isn't topologically trivial (i.e. its Cauchy surfaces are multiple-connected) they give rise to a truly new kind of superselection sectors of topological nature, leading as an aside to non trivial unitary representations of the first homotopy group of the spacetime. A crucial tool in such a construction is a class of objects named 1-cocycles \cite{Ruzzi_05}, relating net-representations with the fundamental group of the spacetime and encoding both the charge and topological content of representations.

 A recent article by Brunetti and Ruzzi \cite{Brunetti_Ruzzi_08} discusses the general theory of topological superselection sectors for spacetime dimension $\geq 3$; the purpose of this research work, instead, is to examine the situation in a lower-dimension context. Our approach is to consider a simple QFT model in two dimensions in order to explicitly work out the details and to gain some insight about what is going on. Actually, we choose the simpler (non trivial) model available out there, namely the free scalar massive (Klein-Gordon) field on the 2-dimensional Einstein cylinder; given the ultrastatic nature of the background spacetime we are able to pursue the calculations at a good depth. On the other end, since the Cauchy surfaces are diffeomorphic to $\bS^{1}$  we expect the appearance of some interesting phenomenon of topological nature. In fact we have proved the existence of a non countable family of topological net-representations, hence a corresponding class of non-trivial unitary representations of the fundamental group  of a Cauchy surface (it being the same as the spacetime fundamental group). 

However, we have not delved further about the charge or topological structure of the corresponding topological superselection sectors; in fact what we have done is to exhibit a family of (unitary inequivalent) non-trivial 1-cocycles, giving rise from the one hand to non-trivial net-representations of the Weyl local algebra net of observables, on the other hand to non-trivial unitary representations of the fundamental group. Actually, it turns out that some crucial results valid in dimension $\geq 3$ fail to be true in our model, and as a consequence even the proofs of various propositions that carry over to our case need essential modifications.

What can be said about the significance of the existence of topological net-representations ? As we'll see later, a topological net-representation cannot be trivialized, namely it's not unitary equivalent to a trivial net-representation; in turn, a trivial net-representation is, in fact, a conventional (global) representation of the local observable net. In other terms a topological net-representation cannot be turned into an ordinary representation via local unitary operators. Pursuing the geometric parallel issued before, the situation here is analogous to the impossibility of converting a set of local charts into a global chart by means of local coordinate trasformations.

We can also adopt another perspective. The GNS construction associates an operatorial representation to every state on a given $C^{*}$-algebra; so a state defined on the whole quasilocal observable algebra induces a global (ordinary) representation of the algebra itself. We could think a net-representation $\{\pi,\psi\}$ as originating from a family of local states $\omega_{O}$ acting on local algebras $\snet{A}{O}$; then, if the $\omega_{O}$'s were induced by restriction by a global state $\omega$, our net-representation $\{\pi,\psi\}$ would be topologically trivial. Summing up, we conclude that a topological net-representation cannot be generated by a global state. 
 
\section*{Overview}

In the first chapter we take a quick look to net-cohomology, briefly discussing concepts and facts needed for subsequent constructions. We start by introducing purely set-theoretic notions as those of poset, simplicial sets, (poset) paths and the first homotopy group of a generic poset; then we specialize the general theory to the case of interest to us, namely the topological one, and in particular the poset of spacetime diamonds. We conclude the section speaking about the category of 1-cocycles, cocycle equivalence and localization, net-representations, representations of the fundamental group, and the reciprocal relations between these notions.

In chapter 2 we set the stage for the central part of this work, contained in the last chapter. Here we give a rather detailed exposition of our reference QFT model, i.e. the massive free scalar field on the 2-dimensional Einstein universe. Following the logical path towards quantization, first of all we introduce the classical field equation (Klein-Gordon equation) with a special focus on well posedness of the Cauchy problem on globally hyperbolic spacetimes; we also define ultrastatic spacetimes.  Then we talk about the symplectic structure of the solution space of the field equation, leading to the construction of the net of local observable algebras via the Weyl correspondence. After a digression on spacetime symmetries, we define quasifree states on a general ultrastatic spacetime; then we specialize to our case, introducing the reference vacuum state and its associated vacuum representation. The central part of the chapter deals about standard properties of the von Neumann local net induced by the vacuum representation: additivity, local definiteness, factoriality, Haag duality, punctured Haag duality, Reeh-Schlieder property. We also give proofs, some of them are original; in particular our proof of Haag duality carries over to more general cases than our specific model. We conclude with an analysis of DHR sectors on the Einstein cylinder (showing that in fact there is only the vacuum one) and finally extending as far as possible the local net's properties from the spatial poset on the circle to the spacetime poset of diamonds.

Chapter 3 is the heart of this work; after some preparation we state the topological selection criterion \cite{Brunetti_Ruzzi_08} and we explicitly exhibit a family of topologically non-trivial 1-cocycles giving rise, as we said, from the one hand to a family of topological net-representations (that is to say topological superselection sectors), from the other hand to a family of non-trivial unitary representations of the fundamental group of the spacetime.

In the appendices we collect some rather technical proofs and statements from previous chapters, in order to avoid distracting the reader from the main line of reasoning. It's worth mentioning that in appendix \ref{app:universal_algebras} we proved the uniqueness of the \emph{universal algebra} in the general case, and determined its structure in the particular case of the Einstein cylinder.

\chapter{Posets and net-cohomology}

To formulate the main results of this work we need some background on homotopy of posets and net-cohomology. 
We will follow closely the exposition given in \cite{Ruzzi_05}, omitting proofs and details and focusing on the main definitions and results. While summarizing this topic we also fix notations and terminology used in the rest of the work.

The basic idea behind homotopy of posets is to reformulate standard topological concepts lying on the notion of ``paths'', i.e. continuos curves on a topological space, in terms of open sets belonging to the same topological space. Roughly speaking, this can be accomplished by approximating a curve with a chain of open sets lying on it (see fig. \ref{fig:path-approx}).

\begin{figure} [htbp]
\begin{center}
\includegraphics[width = 8cm]{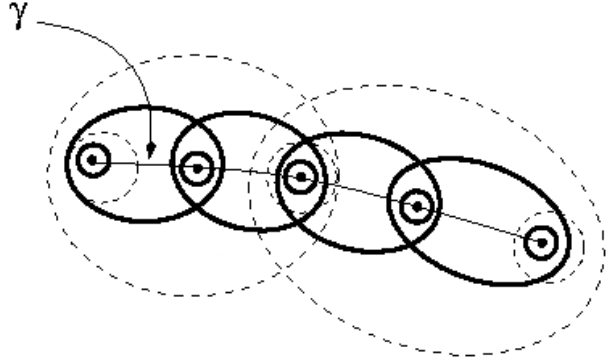}
\end{center}
\caption{Approximation of the curve $\gamma$ by a chain of open sets}
\label{fig:path-approx}
\end{figure}

The motivation for this kind of construction comes from the fact that we would like to study the relation between our net of local algebras encoding the quantum field and the topology of the background spacetime, in particular with respect to homotopy theory; unfortunately there is an obstruction here, due to the fact that local algebras are indexed by open sets while homotopy theory relies on the notion of paths, i.e. continuos curves. 

The goal of poset homotopy theory is indeed to smoothly join these different mathematical models, in order to overcome the aforementioned obstruction.
To start with we recall that a \textbf{poset} $(\poset{P},\leq)$ is a \textbf{p}artially \textbf{o}rdered \textbf{set}, namely a set $\poset{P}$ endowed with a reflexive, antisimmetric, transitive relation $\leq$; a standard example of poset is the family $2^{X}$ of all the subsets of a given set $X$, ordered with respect to set inclusion.

Then we  introduce the notion of \emph{(singular) n-simplex}. The point here is to express in terms of open sets notions like that of a line segment, a triangle, and so on. 
We recall that the \emph{standard n-simplex} $\Delta_{n}$ is defined as:

\begin{equation}
  \Delta_{n}:=
  \left\lbrace(\lambda_{0},\ldots, \lambda_{n})\in\bR^{n+1}|\lambda_{0}+\ldots\lambda_{n}=1,\:\lambda_{i}\in [0,1] \right\rbrace
\end{equation}
Then we see that $\Delta_{0}$ is  a point, $\Delta_{1}$ is a closed interval and so on. We also note that we can embed a $n$-simplex into a $n+1$-simplex in a obvious manner via inclusion maps $d_{i}^{n}:\Delta_{n-1}\mapsto\Delta_{n}$:
\begin{equation}
d_{i}^{n}(\lambda_{0},\ldots, \lambda_{n-1})=(\lambda_{0},\lambda_{1},\ldots,\lambda_{i-1},0,\lambda_{i}, \lambda_{n-1}) 
\end{equation}
for $n\geq$ 1 and $0\leq i\leq n$. 
Observing that a standard $n$-simplex can be viewed as a partially ordered set with respect to the inclusion of its subsimplices, we could define a $n$-simplex built upon open sets as an order-preserving map from a standard $n$-simplex to the poset made of open sets with respect to set inclusion. 

Having in mind this interpretation, we initially adopt a more general approach, constructing our homotopy theory for a generic poset and forgetting for a while every other topology-related details. 

So, fixing a poset $(\poset{P},\leq)$ we define a \emph{singular n-simplex} on $\poset{P}$ as an order-preserving map $f:\Delta_{n}\mapsto\poset{P}$. We denote by $\simplex{n}{\poset{P}}$ the collection of singular $n$-simplices on $\poset{P}$ and by  $\simplex{*}{\poset{P}}$ the collection of all singular simplices on $\poset{P}$, named the \emph{simplicial set} of $\poset{P}$.

The inclusion maps $d_{i}^{n}$ between standard simplices induce by duality the maps $\partial_{i}^{n}:\Sigma_{n}\mapsto\Sigma_{n-1}$, called \emph{boundaries} from their geometric meaning, by setting  $\partial_{i}^{n} f:=f\circ d_{i}^{n}$. For the sake of notational simplicity from now on we will omit the superscript from the symbol $\partial_{i}^{n}$, and denote $0$-simplices by the letter $a$, $1$-simplices by $b$, $2$-simplices by $c$ and so on.
Note that a $0$-simplex $a$ is just an element of $\poset{P}$, i.e. a point in our simplicial set; a $1$-simplex $b$ is formed by an element $|b|$ of $\poset{P}$, the \emph{support} of $b$, and two $0$-simplices $\partial_{0}b$, $\partial_{1}b$ such that $\partial_{0}$, $\partial_{1}\leq |b|$; so we can view $|b|$ as a segment and $\partial_{0}b$, $\partial_{1}b$ as its endpoints. Similarly, a $2$-simplex $c$ is made out of its support $|c|$, namely a ``triangle'', and three $1$-simplices $\partial_{0}c$, $\partial_{1}c$, $\partial_{2}c$, the sides of the triangle. 

\begin{figure} [htbp]
\begin{center}
\includegraphics[width = 10cm]{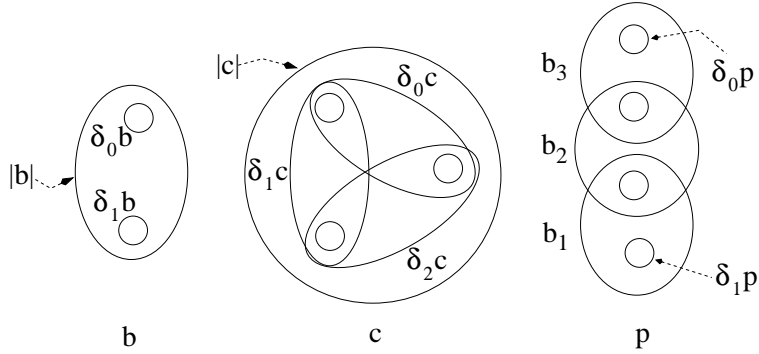}
\end{center}
\caption{$b$ is an 1-simplex, $c$ is a 2-simplex and $\path{p} =\{b_{3},b_{2},b_{1}\}$ is a path. The symbol $\delta$ stands for $\partial$.}
\label{fig:path-simplices}
\end{figure}

Chaining $1$-simplices we obtain a path; formally, given $a_{0}$, $a_{1}\in\simplex{0}{\poset{P}}$, a \emph{path} from $a_{0}$ to $a_{1}$ is a finite ordered sequence $\path{p}=\left\lbrace b_{n},\ldots,b_{1}\right\rbrace$  of $1$-simplices satisfying the relations
\begin{equation}\nonumber
 \partial_{1}b_{1} = a_{0},\:  \partial_{0}b_{i} = \partial_{1}b_{i+1},\: \forall\,\, i\in\left\lbrace 1,\ldots,n-1\right\rbrace,\: \partial_{0}b_{n} = a_{1}.
\end{equation} 
These conditions express the fact that the $1$-simplices making the path must be chained with each other starting at $a_{0}$, the \emph{starting point} $\partial_{1}\path{p}$ of $\path{p}$, and ending at $a_{1}$, the \emph{ending point} $\partial_{0}\path{p}$ of $\path{p}$. We will denote by $\pathset{a_{0}}{a_{1}}$ the set of paths from $a_{0}$ to $a_{1}$, and by $\cpathset{a_{0}}$ the set of loops (i.e. closed paths) based on $a_{0}$. The poset $\poset{P}$ is said to be \emph{pathwise connected} if for every pair $a_{0}$, $a_{1}$ of $0$-simplices the set $\pathset{a_{0}}{a_{1}}$ is nonempty. The support of the path $\path{p}$ is the collection $|\path{p}|\equiv\left\lbrace |b_{i}|, i=1,\ldots,n\right\rbrace$, and given a subset $P$ of $\poset{P}$ we will write $|\path{p}|\subseteq P$ if $|b_{i}|\in P$ for all $i$. 

\section{Causal disjointness}\label{sec:causal_disjointness} 
% Causal disjointness

Until now we developed our poset machinery having in mind only the topological structure of the spacetime, but we know that every spacetime is endowed also with a \emph{causal structure} induced by its Lorentzian metric, and this is an essential ingredient in quantum field theory. So, we need to implement the causal structure in the poset framework; this can be done for a general poset $\poset{P}$, not only for the collection of open sets in a spacetime, introducing a \emph{causal disjointness relation} on the poset $\poset{P}$, i.e. a symmetric binary relation $\perp$ on $\poset{P}$ satisfying the following properties:
\begin{description}
 \item[i)] $\forall\,\,\open{O}_{1}\in\poset{P},\quad\exists\;\open{O}_{2}\in\poset{P}$ such that $\open{O}_{1}\perp\open{O}_{2}$,
 \item[ii)] if $\open{O}_{1}\leq\open{O}_{2}$ and $\open{O}_{2}\perp\open{O}_{3}$, then $\open{O}_{1}\perp\open{O}_{3}$.
\end{description}
These two properties encode the causal structure on the poset $\poset{P}$, and let us define such thing as the causal complement of an element of the poset. Actually, it's better to take a slightly more general approach, defining the causal complement of a whole family of elements in the poset, so that to be able to give a sense to the causal complement of sets \emph{not} contained in the poset. 

Given a subset $P\subseteq\poset{P}$, the \emph{causal complement} of $P$ is the subset $\ccompl{P}$ of $\poset{P}$ defined as

\[\ccompl{P}:=\left\lbrace\open{O}\in\poset{P}| \open{O}\perp\open{O}_{1},\;\forall\,\open{O}_{1}\in P\right\rbrace. 
\]
From the definition it follows immediately that $P_{1}\subseteq P$ implies $\ccompl{P}\subseteq\ccompl{P}_{1}$.

\section{The first homotopy group of a poset}\label{sec:the_first_homotopy_group_of_a_poset}

Our goal now is to introduce the notion of the first homotopy group of a poset; as in classical algebraic topology the  motivation is to have an algebraic object encoding some information about the topological structure of our space. The route to follow is already traced by the classical construction: to start with we have to define a notion of path (as we have already done); then we must know how to compose paths and say what is a reverted path; finally we need a rule to identify homotopic paths. 

First of all we define composition of paths and the reverse of a path. Given $\path{p}=\left\lbrace b_{n},\ldots,b_{1}\right\rbrace \in\pathset{a_{0}}{a_{1}}$ and $\path{q}=\left\lbrace b_{k}^{\prime},\ldots,b_{1}^{\prime}\right\rbrace \in\pathset{a_{1}}{a_{2}}$ the \emph{composition} of $\path{p}$ and $\path{q}$ is the path $\path{p}\ast\path{q}\in\pathset{a_{0}}{a_{2}}$ obtained linking the two paths:
\[
\path{p}\ast\path{q}:=\left\lbrace b_{k}^{\prime},\ldots,b_{1}^{\prime},b_{n},\ldots,b_{1}\right\rbrace
\]
The \emph{reverse} of a $1$-simplex $b$ is the $1$-simplex $\bar{b}$ such that 
\[
\partial_{0}\bar{b} = \partial_{1}b,\: \partial_{1}\bar{b} = \partial_{0}b,\: |\bar{b}| = |b|,
\]
that is the same ``segment'' with inverted endpoints. The reverse of a path is obtained simply reverting each of its components; if $\path{p}=\left\lbrace b_{n},\ldots,b_{1}\right\rbrace \in\pathset{a_{0}}{a_{1}}$, then $\bar{\path{p}}=\left\lbrace \bar{b}_{1},\ldots,\bar{b}_{n}\right\rbrace \in\pathset{a_{1}}{a_{0}}$.

Then we come to homotopy of paths. We recall that from the topological point of view two (topological) paths are homotopic if and only if one of them can be \emph{continuosly} deformed into the other; also the poset notion of homotopy is based on deformation of paths, but it's a deformation of a discrete kind, made up of a finite sequence of elementary deformations. This situation is ultimately due to the fact that (poset) paths are discrete in nature, being made up of a finite number of $1$-simplices. We don't insist here on the notion of elementary deformation, for which we refer to \cite{Ruzzi_05}; figure \ref{fig:deformation} should give an idea of the situation.

\begin{figure} [htbp]
\begin{center}
\includegraphics[width = 10cm]{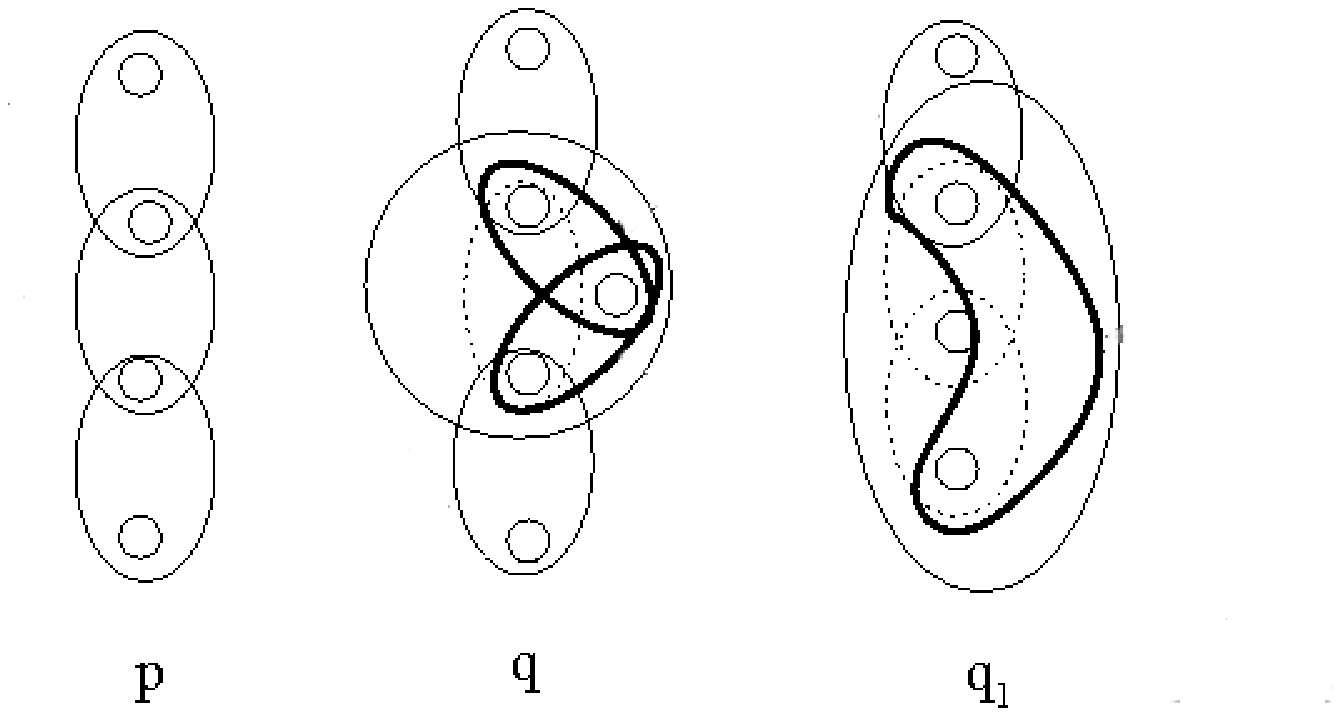}
\end{center}
\caption{$\path{q}$ is an elementary ampliation of the path $\path{p}$ and $\path{q}_{1}$ is an elementary contraction of $\path{p}$.}
\label{fig:deformation}
\end{figure}

Now, given $a_{0}, a_{1}\in\simplex{0}{\poset{P}}$, a \emph{homotopy of paths} in $\pathset{a_{0}}{a_{1}}$ is a map $h:\left\lbrace 1,\ldots,n\right\rbrace\mapsto\pathset{a_{0}}{a_{1}}$ such that $h(i)$ is an elementary deformation of $h(i-1)$ for $1<i\leq n$. Two paths $\path{p}, \path{q}\in\pathset{a_{0}}{a_{1}}$ are said to be \emph{homotopic}, $\path{p}\sim\path{q}$, if there exists a homotopy of paths $h$ in $\pathset{a_{0}}{a_{1}}$ such that $h(1)=\path{p}$ and $h(n)=\path{q}$. 

\noindent A $1$-simplex $b$ is said to be \emph{degenerate} to a $0$-simplex $a_{0}$ if
\[\partial_{0}b = a_{0} = \partial_{1}b,\quad a_{0} = |b|.\]
We  will denote by $b(a_{0})$ the $1$-simplex degenerate to $a_{0}$.
Below we summarize some basic properties of the notions just discussed.

\begin{lemma}
If $\path{p}$, $\path{q}$, $\path{p}_{i}$, $\path{q}_{i}$ are paths, then the following relations hold, whenever they make sense:
\begin{itemize}
\item Path composition is associative, i.e. $\path{p}_{1}\ast(\path{p}_{2}\ast\path{p}_{3}) = (\path{p}_{1}\ast\path{p}_{2})\ast\path{p}_{3}$,
\item $\bar{\bar{\path{p}}}=\path{p}$,
\item Homotopy of paths is an equivalence relation on paths with the same endpoints,
\item If $\path{p}_{1}\sim\path{q}_{1}$ and $\path{p}_{2}\sim\path{q}_{2}$, then $\path{p}_{2}\ast\path{p}_{1}\sim\path{q}_{2}\ast\path{q}_{1}$,
\item $\path{p}\sim\path{q}\Rightarrow \bar{\path{p}}\sim\bar{\path{q}}$, 
\item $\path{p}\ast b(\partial_{1}\path{p})\sim\path{p}\sim b(\partial_{0}\path{p})\ast\path{p}$,
\item $\path{p}\ast\bar{\path{p}}\sim b(\partial_{0}\path{p})$ and $\bar{\path{p}}\ast\path{p}\sim b(\partial_{1}\path{p})$.
\end{itemize}
\end{lemma}
\noindent Finally we are ready to define the first homotopy group of a poset. Fix a base $0$-simplex $a_{0}$ and consider the set $\cpathset{a_{0}}$ of loops, i.e. closed paths, based at $a_{0}$; then the composition and the reverse are internal operations  on $\cpathset{a_{0}}$ and $b(a_{0})\in\cpathset{a_{0}}$. 
The first homotopy group $\homgroup{P}{a_{0}}$ of $\poset{P}$ is $\cpathset{a_{0}}$ quotiented with respect to the homotopy equivalence relation $\sim$:
\[
\homgroup{P}{a_{0}}\equiv\cpathset{a_{0}}/\sim.
\]
Let $[\path{p}]$ denote the homotopy class of an element $\path{p}$ of $\cpathset{a_{0}}$; the product 
\[
[\path{p}]\ast[\path{q}]:=[\path{p}\ast\path{q}],\quad [\path{p}],[\path{q}]\in\homgroup{P}{a_{0}},
\] 
is associative and makes $\homgroup{P}{a_{0}}$ into a group, the group identity being $[b(a_{0})]$ and $[\path{p}]^{-1}\equiv[\bar{\path{p}}]$. In general the first homotopy group depends on the base point where it is calculated; however, at least in the ``topological'' case, if the space is arcwise connected all homotopy groups are isomorphic. This situation carries over to the poset case, but arcwise connectedness must be replaced by pathwise connectedness. In fact, given another base point $a_{1}$, let $\path{q}$ be a path from $a_{0}$ to $a_{1}$; then the map

\[\homgroup{P}{a_{0}}\ni[\path{p}]\mapsto[\path{q}\ast\path{p}\ast\bar{\path{q}}]\in\homgroup{P}{a_{1}}\]
is a group isomorphism. We can summarize the above discussion giving the following

\begin{definition}
 With the above notations, we call $\homgroup{P}{a_{0}}$ the \emph{first homotopy group} of $\poset{P}$ based on $a_{0}\in\simplex{0}{\poset{P}}$. If $\poset{P}$ is pathwise connected, we denote this group by $\shomgroup{P}$ and call it the \emph{fundamental group} of $\poset{P}$. If the fundamental group is trivial, i.e. $\shomgroup{P}=1$, we'll say that $\poset{P}$ is \emph{simply connected}.
\end{definition}

We conclude this section citing an useful result.
\begin{proposition}
 If $\poset{P}$ is directed\footnote{A poset is (upward) directed if for any $a_{0}, a_{1}\in\poset{P}$ exists $a_{2}\in\poset{P}$ such that $a_{0}, a_{1}\leq a_{2}$.}, then it is pathwise and simply connected.
\end{proposition}

\section{Topological posets}
So far we considered a generic poset $\poset{P}$; now we restrict our attention to the case of interest to us, that is when $\poset{P}$ is a basis for a topological space $X$. This way we can establish a relation between the poset notions discussed so far and the topological ones, our goal being to understand how topology affects net-cohomology. 
So let us consider a Hausdorff topological space $X$, and take as the poset $\poset{P}$ a topological basis of $X$ ordered under set inclusion $\subseteq$. As we said, our goal is to link classical topological concepts, particularly those of continuos curves and the (topological) first group of homotopy built on them, to poset ones, like paths and the (poset) first group of homotopy; the technical tool needed here is the notion of \emph{approximation} of a curve by a poset path.
By a curve in $X$ we mean as usual a continuos map $\gamma:[0,1]\mapsto X$.
\begin{definition}
 Given a (continuos) curve $\gamma$ in $X$, a path $\path{p}=\left\lbrace b_{n},\ldots, b_{1} \right\rbrace$ is said to be a \emph{poset approximation} of $\gamma$ (or simply an \emph{approximation}) if there is a partition $0 = s_{0} < s_{1} < \cdots < s_{n} = 1$ of the interval $[0,1]$ such that:
\[
\gamma([s_{i-1},s_{i}])\subseteq |b_{i}|,\: \gamma(s_{i-1})\in\partial_{1}b_{i},\:\gamma(s_{i})\in\partial_{0}b_{i},\quad i=1,\ldots,n 
\]
We denote the set of poset approximations of $\gamma$ by $\app{\gamma}$.
 \end{definition}

Since $\poset{P}$ is a topological basis for $X$, we have that $\app{\gamma}\neq\emptyset$ for any curve $\gamma$; it can also be shown \cite{Ruzzi_05} that approximations of curves behave well with respect to operations defined on curves (composition, reversal).

The next logical step is to introduce an order relation in the set of approximations of a given curve, so that we can talk about ``better'' or ``worse'' approximations.
\begin{definition}
Given $\path{p}, \path{q}\in \app{\gamma}$, we say that $\path{q}$ is \emph{finer} than $\path{p}$ whenever $\path{p}=\left\lbrace b_{n},\ldots, b_{1} \right\rbrace$ and $\path{q}=\path{q}_{n}\ast\cdots\ast\path{q}_{1}$, where $\path{q}_{i}$ are paths satisfying
\[
|\path{q}_{i}|\subseteq |b_{i}|,\:\partial_{0}\path{q}_{i}\subseteq\partial_{0}b_{i},\:\partial_{1}\path{q}_{i}\subseteq\partial_{1}b_{i},  \quad i=1,\ldots,n.
\]
We will write $\path{p}\prec\path{q}$ to denote that $\path{q}$ is a finer approximation than $\path{p}$. 
\end{definition}
\noindent It's easy to see from the definition that ($\app{\gamma}$, $\prec$) is a directed poset (i.e. for any $\path{p}$, $\path{q}\in\app{\gamma}$ there exists $\path{p}_{1}\in\app{\gamma}$ such that $\path{p}$, $\path{q}\prec\path{p}_{1}$; this follows from the fact that $\poset{P}$ is a basis for $X$. We already saw that we can find an approximation for any curve $\gamma$; the converse, namely that for a given path $\path{p}$ there is a curve $\gamma$ such that $\path{p}\in\app{\gamma}$ is also true, provided that elements of $\poset{P}$ are arcwise connected (with respect to the topological space $X$). Furthermore there is a relation between connectedness for posets and connectedness for topological spaces: if the elements of $\poset{P}$ are arcwise connected, then an open set $Y\subseteq X$ is arcwise connected in $X$ iff the poset $\poset{P}_{Y}$ defined as
\[
\poset{P}_{Y}:=\left\lbrace\open{O}\in\poset{P}|\open{O}\subseteq Y\right\rbrace 
\]   
is pathwise connected.

\noindent The following lemma relates poset homotopy of paths and topological homotopy of curves.
\begin{lemma}
Assume that elements of $\poset{P}$ are arcwise and simply connected subsets of $X$. Take two curves $\beta$ and $\gamma$ with the same endpoints, and let $\path{p}$, $\path{q}\in\pathset{a_{0}}{a_{1}}$ be, respectively, approximations of $\beta$ anf $\gamma$; then $\path{p}$ and $\path{q}$ are homotopic if and only if $\beta$ and $\gamma$ are homotopic.
\end{lemma}
Since first homotopy groups (both the poset and topological ones) are built from curves and paths using only the notion of homotopy, it's clear that they coincide, under suitable assumptions, as stated in the main theorem of this section:
\begin{theorem}\label{th:homotopy-groups}
Let $X$ be an Hausdorff, arcwise connected topological space, and let $\poset{P}$ be a topological basis for $X$ whose elements are arcwise and simply connected subsets of $X$; then $\stdhomgroup{X}\simeq\shomgroup{P}$.
\end{theorem}
A consequence of this theorem will be useful later.
\begin{corollary}\label{cor:not_directed_poset} 
Let $X$ and $\poset{P}$ be as in the previous theorem; if $X$ is non-simply connected, then $\poset{P}$ is not directed under inclusion. 
\end{corollary}

\subsection{Index sets}
%Index sets
Let $X$ be a topological space. As we'll see later, if we want to develop net-cohomology on it we need  objects like nets of local algebras, cocycles and so on, whose in turn rely on the choice of a poset endowed with a causal disjointness relation, namely a triple ($\poset{P}$, $\leq$, $\perp$). Of course we want our net-cohomology to be related to the topological structure of $X$, so we assume that $\poset{P}$ is a subset of the family $O(X)$ of open sets in $X$, and that the order relation $\leq$ coicides with set inclusion $\subseteq$; for the moment we make no hypotheses on the causal disjointness relation $\perp$. 

When we come to net-cohomology, we know that the poset $\poset{P}$ is used to index local algebras, so we may ask what requirements $\poset{P}$ must fulfill in order to be considered a ``good'' index set. Of course it would be desiderable to avoid the introduction of ``artificial'' topological obstructions, namely we would like to have  $\stdhomgroup{X}\simeq\shomgroup{P}$, so in view of theorem \ref{th:homotopy-groups} we require  $X$ to be Hausdorff and arcwise connected and the following definition is motivated:
\begin{definition}
Let $X$ be a Hausdorff, arcwise connected topological space and $\perp$ a causal disjointness relation defined on $O(X)$.
 We say that $\poset{P}\subseteq O(X)$ is a \emph{good index set} associated with ($O(X)$, $\subseteq$, $\perp$) if $\poset{P}$ is a topological basis for $X$ whose elements are arcwise and simply connected subsets of $X$ with non-empty causal complements. We denote by $\indexset{X}$ the collection of good index sets associated with ($O(X)$, $\subseteq$, $\perp$).
\end{definition}
Note that in general $\indexset{X}$ may be empty, although this doesn't happen for the application we have in mind.

\section{Spacetime posets}

Now we come to quantum field theory. For reasons that will become clearer in chapter \ref{ch:QFT on the cylindric flat spacetime}, we restrict ourselves to consider only globally hyperbolic background spacetimes. So let $\spacetime{M}$ be a globally hyperbolic spacetime of dimension $d=n+1\ge 2$; as we said, in order to develop  net-cohomology we need to choose a poset as an index set for nets of local algebras. Following the considerations exposed above, we will take a subposet of $O(\spacetime{M})$; the causal disjointness relation is the ordinary spacetime causal disjointness\footnote{We recall that $\cdevel{S}$ denotes the \emph{causal development} of a spacetime subset $S$, i.e. the set of points in the spacetime that can be joined to $S$ by a causal curve.}, namely 
\[S_{1}\perp S_{2}\Leftrightarrow S_{1}\subseteq\spacetime{M}\setminus \cdevel{S_{2}},
\]
if $S_{1}, S_{2}\subseteq\spacetime{M}$.
For our purposes a good choice  for the index set is the family of diamonds on $\spacetime{M}$.
\begin{definition}\label{def:diamonds}
Given a foliation of the spacetime by means of spacelike Cauchy surfaces, $\spacetime{M}\simeq\Sigma\times\bR$, take a surface in the foliation, say $\sdc{C}\equiv\Sigma_{t}$ and denote by $\mathfrak{G}(\sdc{C})$ the collection of open subsets $G$ of $\sdc{C}$ of the form $\phi(B)$, where ($U$,$\phi$) is a coordinate chart of $\sdc{C}$ and $B$ is an open ball of $\bR^{n}$ with $\cl{B}\subset\phi^{-1}(U)$. We call a \emph{diamond}  of $\spacetime{M}$ a subset $\open{O}$ of the form\footnote{$D(A)$ denotes the Cauchy development of the set $A$; see section \ref{sec:spacetime formulation}.} $D(G)$ where $G\in\mathfrak{G}(\sdc{C})$ for some surface in the foliation. $G$ is called the \emph{base} of $\open{O}$ while $\open{O}$ is said to be based on $G$: for short $\open{O}\equiv\DIAMOND{G}$. We denote by $\poset{K}$ the collection of all diamonds of $\spacetime{M}$. 
\end{definition}
It's easy to see that $\poset{K}$ is indeed a good choice for an index set.
\begin{proposition}
$\poset{K}$ is a topological basis for $\spacetime{M}$. Any diamond is a relatively compact, arcwise and simply connected open subset of $\spacetime{M}$ and $\poset{K}\in\indexset{\spacetime{M}}$.
\end{proposition}

We also note that for a globally hyperbolic spacetime the (topological) first group of homotopy $\stdhomgroup{\spacetime{M}}$ coincide with $\stdhomgroup{\sdc{C}}$, where $\sdc{C}$ is an arbitrary Cauchy surface of $\spacetime{M}$; this follows from the fact that $\stdhomgroup{X\times Y}\sim\stdhomgroup{X}\times\stdhomgroup{Y}$ for each pair of topological spaces $X$, $Y$, and $\spacetime{M}\sim\sdc{C}\times\bR$.

\section{Net-cohomology}

Algebraic quantum field theory is based on nets of local algebras, and local algebras are indexed by open sets in the spacetime. Actually, this is the very motivation for developing the poset machinery above: to relate the physical content of the theory, encoded by local algebras, with spacetime topology. As we said, this goal requires to reformulate topological concepts, like curves and the first homotopy group, in terms of paths made up of open sets. We saw that this can be accomplished in a very general framework, where the only structure required is an order relation on a set. Interestingly enough, this poset approach is sufficient to develop the local algebras part of the machinery, at least for the first stages.

As it's well known, in algebraic quantum field theory local algebras come in two flavours: abstract and concrete. Nets of local abstract algebras define a correspondence between spacetime sets and abstract $C^*$-algebras (hence the name); in contrast,  nets of local concrete algebras associate to each spacetime set a von Neumann algebra, namely a closed\footnote{In the weak operatorial topology of $\bop{\hs{H}}$.} subalgebra of the  $C^*$-algebra $\bop{\hs{H}}$ of bounded operators acting on some (fixed) Hilbert space $\hs{H}$.
To be more precise, let ($\poset{P}$, $\leq$, $\perp$) be a poset endowed with a causal disjointness relation $\perp$; we also assume $\poset{P}$ to be pathwise connected. A \emph{net of local algebras} indexed by $\poset{P}$ is a correspondence
\[
\net{A}{P} :\poset{P}\ni\open{O}\mapsto\snet{A}{\open{O}}
\]
associating to any element $\open{O}$ in the poset an algebra $\snet{A}{\open{O}}$, and satisfying
\begin{eqnarray}
 \open{O}_{1}\leq \open{O}_{2} & \Rightarrow & \snet{A}{\open{O}_1}\subseteq\snet{A}{\open{O}_2}\quad \mbox{(isotony)}   \nonumber \\
 \open{O}_{1}\perp\open{O}_{2} & \Rightarrow & \snet{A}{\open{O}_{1}}\subseteq\snet{A}{\open{O}_{2}}^{\prime}\quad \mbox{(causality)} \nonumber
\end{eqnarray}
where with the symbol $\snet{A}{\open{O}}^{\prime}$  we denote the algebraic commutant of the algebra $\snet{A}{\open{O}}$.
Depending on the nature of the algebras $\snet{A}{\open{O}}$, we'll talk about abstract/concrete nets of local  algebras. 
It's worth noting that properties like isotony and causality, for abstract algebras, only make sense if we can embed ``smaller'' algebras into ``larger'' ones; in other terms for every $\open{O}_1\leq \open{O}_2$ an isometric $C^{*}$-algebras embedding $j_{\open{O}_2 \open{O}_1}:\snet{A}{\open{O}_1}\mapsto\snet{A}{\open{O}_2}$ must be defined; then, as in differential geometry, we can identify $A$ with $j_{\open{O}_2 \open{O}_1}(A)$, if $A\in\snet{A}{\open{O}_1}$\footnote{In addition, in order to define causality for abstract algebras, the poset $\poset{P}$ must be directed.}.

We previously defined the causal complement of an element $\open{O}\in\poset{P}$ as a family of elements in the poset, so what do we mean by $\snet{A}{\ccompl{\open{O}}}$ ?
The answer is: the algebra generated by all members of the family\footnote{As before, for abstract algebras we need to assume the existence of an ``environment'' algebra containing (via embeddings) all relevant local algebras in order to make sense of this definition.}. In symbols:  
 \[\snet{A}{\ccompl{\open{O}}}:=\bigvee_{\open{O}_{1}\perp\open{O}}\snet{A}{\open{O}_{1}}.\]
The net $\net{A}{P}$ is said to be \emph{irreducible} if\footnote{See previous footnote.} 
\[\bigcap_{\open{O}\in\poset{P}}\snet{A}{\open{O}}^{\prime}=\bC\cdot 1.\] 

\section{The category of $1$-cocycles}\label{sec:the_category_of_1-cocycles}

At this point we introduce a technical device which turn out to be useful in constructing unitary representations of the first homotopy group: the category of $1$-cocycles. First of all, though, we need to talk about net representations.

\subsection{Net representations}

So let $\poset{P}$ be a poset with a causal disjointness relation $\perp$, and let $\net{A}{P}:\poset{P} \ni \open{O} \to \snet{A}{\open{O}}$  be an irreducible net of local (abstract) algebras; from now on we refer to it as the \emph{reference net of  observables}.

\noindent If $\widetilde{\open{O}} \subseteq \open{O}$, the natural isometric  $*$-homomorphisms given by inclusion maps of $\snet{A}{\widetilde{\open{O}}}$ into 
$\snet{A}{\open{O}}$ will be denoted by $j_{\open{O}\widetilde{\open{O}}}$ and named the \emph{inclusion morphisms} as in \cite{Brunetti_Ruzzi_08}.
The coherence requirement $j_{\open{O}'\open{O}} = j_{\open{O}'\widetilde{\open{O}}} j_{\widetilde{\open{O}}\open{O}}$ for ${\open{O}}\subseteq \widetilde{\open{O}} \subseteq \open{O}'$ is trivially fulfilled. 

Now we come to the central concept, namely net representations. The idea here is to set up a local representation structure for the chosen reference net of  observables; the word ``local'' here means that the representation map changes as we consider different local (abstract) $C^{*}$-algebras; as a result we have a family of representations labeled by elements of the poset $\poset{P}$. As a matter of principle, Hilbert spaces carrying the representations can vary too, but for our purposes we fix one common Hilbert space and stick to it (we'll see in section \ref{sec:vacuum_representations} that this will be the one carrying the vacuum representation).

Of course, to be of any use in modeling quantum field theory, the various local representations need to be related to each other; a minimal request is the validity of an obvious compatibility condition with respect to the ordering of poset elements.
\noindent Fix an infinite-dimensional, separable Hilbert space $\hs{H}$:  a \emph{net representation} on $\hs{H}$ (for the observable net $\net{A}{P}$) is a pair $\{\pi,\psi\}$, where $\pi$ denotes a function that associates to any $\open{O}\in\poset{P}$
a representation $\pi_\open{O}$ of  $\snet{A}{\open{O}}$ on the common Hilbert space $\hs{H}$; $\psi$ is a function associating a linear operator $\psi_{\open{O}\widetilde{\open{O}}}: \hs{H}\to \hs{H}$ with any pair $\open{O},\widetilde{\open{O}}\in \poset{P}$, with
$\widetilde{\open{O}}\subseteq \open{O}$.  The functions $\pi$ and $\psi$ are required to satisfy the following compatibility relations:
\begin{alignat}{2}\label{eq:UNR}
\psi_{\open{O}\widetilde{\open{O}}} \pi_{\widetilde{\open{O}}}(A) &= \pi_{\open{O}} j_{\open{O}\widetilde{\open{O}}}(A) \psi_{\open{O}\widetilde{\open{O}}}\:, &\quad\quad
A \in \snet{A}{\widetilde{\open{O}}}\:, \: \widetilde{\open{O}} \subseteq \open{O},\\
\psi_{\open{O}'\open{O}} \psi_{\open{O}\widetilde{\open{O}}} &= \psi_{\open{O}'\widetilde{\open{O}}} \:, &\quad\quad
\widetilde{\open{O}} \subseteq \open{O} \subseteq \open{O}'.
\end{alignat}
where $j_{\open{O}\widetilde{\open{O}}}:\snet{A}{\widetilde{\open{O}}}\mapsto\snet{A}{\open{O}}$ are the isometric embeddings introduced before.

For our purposes we'll consider a particular kind of net representations, requiring the operators $\psi_{\open{O}\widetilde{\open{O}}}$ to be unitary; in this case we'll speak about \emph{unitary net representations}.
The next step is to introduce an equivalence relation between (unitary) net representations, passing through the notion of an intertwiner.

\noindent An \emph{intertwiner} from $\{\pi,\psi\}$ to $\{\rho, \phi\}$ is a function $\itwiner{T}$ associating a bounded operator
$\itwiner{T}_\open{O} : \hs{H} \to \hs{H} $ with any $\open{O} \in \poset{P}$, and satisfying the relations
\begin{equation}
\itwiner{T}_\open{O} \pi_{\open{O}} = \rho_\open{O} \itwiner{T}_\open{O}, \quad \quad \mbox{and}
\quad \quad \quad \itwiner{T}_\open{O} \psi_{\open{O}\widetilde{\open{O}}}  = \phi_{\open{O}\widetilde{\open{O}}}  \itwiner{T}_\open{O} \:, \quad
\widetilde{\open{O}} \subseteq \open{O}. \label{eq:intertREP}
\end{equation}
We denote the set of intertwiners from $\{\pi,\psi\}$ to $\{\rho, \phi\}$ by the symbol  $(\{\pi,\psi\},\{\rho, \phi\})$,
and say that two net representations are \emph{unitarily equivalent} if they admit a unitary
intertwiner $\itwiner{T}$, that is $\itwiner{T}_\open{O}$ is a unitary operator for any $\open{O}\in \poset{P}$.
$\{\pi,\psi\}$ is \emph{irreducible} when unitary elements of $(\{\pi,\psi\},\{\pi, \psi\})$  are of the form $c\!\cdot\! 1$ with 
$c\in \bC$ and $|c|=1$.

\noindent\textbf{Remark}. The construction outlined above, concerning local representations for a net of local abstract algebras, follows closely a recurrent pattern in differential geometry, where a global object is often constructed up by pasting together locally defined objects satisfying suitable compatibility conditions. From this perspective it can be said that we are dealing with a kind of non-commutative geometry, where objects are modeled after those living in a (commutative) manifold. In this scheme we can regard the local $C^{*}$-algebras as open sets on the manifold, net representations as coordinate charts on $\bop{\hs{H}}$, the linear operators  $\psi_{\open{O}\widetilde{\open{O}}}$ as ``transition functions'' between local coordinate charts, and equivalent net representations as compatible atlas belonging to a given differentiable structure. Obviously these considerations are valid only at an heuristic level, but neverthless can be useful to understand what's going on. 

\subsection{Introducing cocycles}
As we said, $1$-cocycles are a technical tool useful for constructing representations. The basic idea here is to associate a non-commutative object (an operator on some Hilbert space) to each poset path in a homotopy-invariant manner, so that the operator depends only on the equivalence class the path belongs to; this behaviour can be achieved by requesting the map to be invariant with respect to elementary deformations of paths (see section \ref{sec:the_first_homotopy_group_of_a_poset}), and is espressed by a condition called \emph{$1$-cocycle identity}. 

The formal definition is as follows: given a complex, infinite dimensional, separable  Hilbert space $\hs{H}$, a \emph{(unlocalized) $1$-cocycle} $\coc{z}$ on $\poset{P}$ valued in $\bop{\hs{H}}$ is a field $\coc{z}:\simplex{1}{\poset{P}}\ni b\mapsto \coc{z}(b)\in\bop{\hs{H}}$ of unitary operators satisfying the \emph{$1$-cocycle identity}:
\begin{equation}\label{eq:cocycle_identity} 
\coc{z}(\partial_{0}c)\cdot\coc{z}(\partial_{2}c)=\coc{z}(\partial_{1}c),\quad c\in\simplex{2}{\poset{P}}. 
\end{equation}
Relations between $1$-cocycles are described by intertwiners.
An \emph{intertwiner} $\itwiner{t}$ between a pair of $1$-cocycles $\coc{z},\coc{z}_{1}$ is a field of operators $\itwiner{t}:\simplex{0}{\poset{P}}\ni a\mapsto \itwiner{t}_{a}\in\bop{\hs{H}}$ satisfying the relation 
\[
\itwiner{t}_{\partial_{0}b}\cdot\coc{z}(b) = \coc{z}_{1}(b)\cdot\itwiner{t}_{\partial_{1}b},\quad \forall\,b\in\simplex{1}{\poset{P}}.
\]
We denote by $(\coc{z},\coc{z}_{1})$ the set of intertwiners between $\coc{z}$ and $\coc{z}_1$. En passant we note the analogy in the definitions of cocycle intertwiners and net representations intertwiners. 

It turns out that the set of $1$-cocycles has a quite rich structure. The \emph{category of (unlocalized) $1$-cocycles} is the category $\ucocyclecat{P}{\hs{H}}$ whose objects are $1$-cocycles and whose arrows are the corresponding set of intertwiners. The composition between $\itwiner{s}\in(\coc{z},\coc{z}_{1})$ and $\itwiner{t}\in(\coc{z}_{1},\coc{z}_{2})$ is the arrow $\itwiner{t}\cdot\itwiner{s}\in(\coc{ z},\coc{z}_{2})$ defined  as 
\[
(\itwiner{t}\cdot\itwiner{s})_{a}:=\itwiner{t}_{a}\cdot\itwiner{s}_{a},\quad\forall\, a\in\simplex{0}{\poset{P}}.
\]
Note that the arrow $\itwiner{1}_{z}$ of $(\coc{z},\coc{z})$ defined as $(\itwiner{1}_{z})_{a} = 1\quad \forall\, a\in\simplex{0}{\poset{P}}$ is the identity in $(\coc{z},\coc{z})$. 
In addition, the (unlocalized) $1$-cocycle category is also a \emph{$C^*$-category}\footnote{See appendix \ref{app:misc} for the definition of a $C^*$-category.}.
 In fact the set $(\coc{z},\coc{z}_{1})$ has a complex vector space structure:
\[
(\alpha\cdot\itwiner{t}+\beta\cdot\itwiner{s})_{a}:=\alpha\cdot\itwiner{t}_{a}+\beta\cdot\itwiner{s}_{a},\quad a\in\simplex{0}{\poset{P}},\:\alpha,\beta\in\bC,\:\itwiner{t},\itwiner{s}\in(\coc{z},\coc{z}_{1}).
\] 
With these operations and the composition ``$\cdot$'', the set $(\coc{z},\coc{z})$ is an algebra with identity $1_{\coc{z}}$. The category $\ucocyclecat{P}{\hs{H}}$ has an adjoint $\ast$, defined as the identity, $\coc{z}^{*} = \coc{z}$, on the objects, while the adjoint  $\itwiner{t}^{*}\in(\coc{z}_{1},\coc{z})$ of an arrow $\itwiner{t}\in(\coc{z},\coc{z}_{1})$ is defined as 
\[
(\itwiner{t}^{*})_{a}:=(\itwiner{t_{a}})^{*},\quad\forall\, a\in\simplex{0}{\poset{P}},
\]
where $(\itwiner{t}_{a})^{*}$ stands for the operator adjoint of $\itwiner{t}_{a}$ in $ \bop{\hs{H}}$. 

\noindent Now let $\norm{\cdot}$ be the operator norm in $\bop{\hs{H}}$. Given $\itwiner{t}\in (\coc{z},\coc{z}_{1})$, it's easy to see \cite{Brunetti_Ruzzi_08}  that $\norm{\itwiner{t}_{a}} = \norm{\itwiner{t}_{a_1}}$ for any pair $a$, $a_{1}$ of $0$-simplices, since $\poset{P}$ is pathwise connected. Therefore by putting
\[\norm{\itwiner{t}}:=\norm{\itwiner{t}_{a}},\quad a\in\simplex{0}{\poset{P}},\]
we see that $(\coc{z},\coc{z}_{1})$ is a complex Banach space for any $\coc{z},\coc{z}_{1}\in\ucocyclecat{P}{\hs{H}}$, and $(\coc{z},\coc{z})$ is a $C^*$-algebra for any $\coc{z}\in\ucocyclecat{P}{\hs{H}}$. This entails that $\ucocyclecat{P}{\hs{H}}$ is a $C^*$-category.

Needless to say, the expression ``unlocalized $1$-cocycle'' suggests the existence of \emph{localized} $1$-cocycles, as it's the case. 
The only additional property of this new cocycle's flavour is localization: the operator obtained by evaluating a localized cocycle on a $1$-simplex belongs to the local (concrete algebra) living on the simplex' support. In other terms, let $\net{R}{P}$ be a net of local concrete algebras valued in $\bop{\hs{H}}$ (for example that arising from a net representation of the reference net  $\net{A}{P}$); localization amounts to add the following \emph{locality condition} to the definition given above for unlocalized $1$-cocycles:
\[
\coc{z}(b)\in\snet{R}{|b|},\quad \forall\, b\in\simplex{1}{\poset{P}}.
\]
A parallel notion of localized cocycle intertwiners can be made: the intertwiner $\itwiner{t}$ is said to be localized if
\[
\itwiner{t}_{a}\in\snet{R}{a},\quad \forall\, a\in\simplex{0}{\poset{P}}.
\]
The category of localized $1$-cocycles on the local algebras net $\net{R}{P}$ will be denoted by the symbol $\cocyclecat{\net{R}{P}}$.
Finally, we say that a $1$-cocycle $\coc{z}$ is a \emph{coboundary} if it can be written as $\coc{z}(b)= W^*_{\partial_0b } W_{\partial_1b }$, $b \in \simplex{1}{\poset{P}}$, for some field of unitaries $\simplex{0}{\poset{P}} \ni a \mapsto W(a)\in\bop{\hs{H}}$.

\subsection{Cocycle equivalence}
%Cocycle equivalence
Now we introduce an equivalence relation between $1$-cocycles. 
Two (localized or not) $1$-cocycles $\coc{z}$, $\coc{z}_{1}$ are said to be \emph{equivalent} (or \emph{cohomologous}) if there exists a unitary arrow $\itwiner{t}\in (\coc{z},\coc{z}_{1})$ between them; a $1$-cocycle is \emph{trivial} if it's equivalent to the identity cocycle $\coc{i}$, defined as $\coc{i}(b)=1$ for any $b\in\simplex{1}{\poset{P}}$; this is equivalent to say that $\coc{z}$ is a coboundary.
Note that if the cocycles $\coc{z}$, $\coc{z}_{1}$ are localized, the above definition requires the intertwiner $\itwiner{t}\in (\coc{z},\coc{z}_{1})$ to be localized, too. This imply that, for localized cocycles, two different equivalence relations are available, depending on the fact we require or not the operator-valued field $\itwiner{t}$ to be localized. In the former case we simply speak about equivalent cocycles; in the latter we say that $\coc{z}$, $\coc{z}_{1}$ are equivalent in $\bop{\hs{H}}$. It's worth observing that for localized cocycles unitary equivalence is stronger than equivalence in $\bop{\hs{H}}$, since localization is required (for unlocalized cocycles only the latter one makes sense).

We denote by $\tcocyclecat{\net{R}{P}}$ the set of (localized) $1$-cocycles trivial in $\bop{\hs{H}}$. From a topological point of view, triviality in $\bop{\hs{H}}$  means path independence. In details, if the evaluation of a $1$-cocycle $\coc{z}$ on the path $\path{p} = \left\lbrace b_{n}, \ldots, b_{1} \right\rbrace$ is defined as
\[
\coc{z}(\path{p}):=\coc{z}(b_{n})\cdot\ldots\cdot\coc{z}(b_{2})\cdot\coc{z}(b_{1}),
\]
$\coc{z}$ is said to be \emph{path-independent} on a subset $P\subseteq\poset{P}$ whenever 
\[\coc{z}(\path{p}) = \coc{z}(\path{q})\quad\forall\,\path{p}, \path{q}\in\pathset{a_{0}}{a_{1}}\;\mbox{such that}\: |\path{p}|, |\path{q}|\subseteq P,\]
for any $a_{0}, a_{1}\in\simplex{0}{\poset{P}}$.
Then from pathwise connectedness of $\poset{P}$ it descends that a $1$-cocycle is trivial in $\bop{\hs{H}}$ if and only if it is path-independent on all $\poset{P}$.

We close this section with a result relating unitary net representations and $1$-cocycles.
We start by putting out a definition mimicking that given above for net representations: a $1$-cocycle $\coc{z} \in\ucocyclecat{P}{\hs{H}} $ is said to be  \emph{irreducible} if there are no unitary intertwiners  in $(\coc{z},\coc{z})$  barring those of the form $c\!\cdot\! 1$ with $c\in \bC$ and $|c|=1$.
Given a unitary net representation $\{\pi,\psi\}$ of $\net{A}{P}$ over $\hs{H}$ define
\begin{equation}\label{eq:zetapi}
\zeta^\pi(b) := \psi^*_{|b|,\partial_0b} \psi_{|b|,\partial_1b}\:, \quad b\in \simplex{1}{\poset{P}}.
\end{equation}
As usual $|b|\in \simplex{0}{\poset{P}}$ denotes the support of the symplex $b$.
One can check that $\zeta^\pi$ is a $1$-cocycle in $\ucocyclecat{P}{\hs{H}}$. 
$\{\pi,\psi\}$ is said to be \emph{topologically trivial} if $\zeta^\pi$ is trivial.

It can be proven \cite{Brunetti_Ruzzi_08} that
 if the unitary net representations $\{\pi,\psi\}$ and $\{\rho,\phi\}$ are unitarily equivalent, then the corresponding
$1$-cocycles $\zeta^\pi$ and $\zeta^\phi$ are equivalent in $\bop{\hs{H}}$; moreover, if
the unitary net representation $\{\pi,\psi\}$ is topologically trivial, then it is equivalent to a  
unitary net representation of the form $\{\rho,\bI\}$, where  all
 $\bI_{\widetilde{\open{O}} \open{O}}$ are  the identity operators.

\subsection{Representations of the first homotopy group}
%Connection between homotopy and net-cohomology
At this point we are in the right position to understand the usefulness of the cocycle machinery: it can be used to construct unitary representation of the first homotopy group. In fact, the main result is that there is a one-to-one correspondence between $1$-cocycles and unitary representations, modulo suitable equivalence relations.
Let us consider a pathwise connected poset $\poset{P}$ equipped with a causal disjointness relation $\perp$, and let $\net{R}{P}$ be an irreducible net of local (concrete) algebras. To start with, we list some preliminary results about cocycles and paths, in particular invariance of $1$-cocycles for homotopic paths.
\begin{lemma}
 Let $\coc{z}\in\cocyclecat{\net{R}{P}}$. Then:
\begin{description}
\item[i)] If $\path{p},\path{q}\in\pathset{a_{0}}{a_{1}}$ are homotopic, $\path{p}\sim\path{q}$, then $\coc{z}(\path{p}) =  \coc{z}(\path{q})$,
\item[ii)]$\coc{z}(b(a)) = 1$ for any $0$-simplex $a$,
\item[iii)]$\coc{z}(\bar{\path{p}}) = \coc{z}(\path{p})^{*}$ for any path $\path{p}$.
 \end{description}
\end{lemma}
\noindent Now we come to the main theorem. Fix a base $0$-simplex $a_{0}$; given $\coc{z}\in\cocyclecat{\net{R}{P}}$ define the following map sending an equivalence class of (poset) paths into an operator in $\bop{\hs{H}}$:
\[
\pi_{\coc{z}}([\path{p}]):=\coc{z}(\path{p}),\quad[\path{p}]\in\shomgroup{P}.
\]
The definition is well posed due to cocycle invariance with respect to path homotopy.
\begin{theorem}\label{th:cocycle-representations}
 The correspondence $\cocyclecat{\net{R}{P}}\ni\coc{z}\mapsto\pi_{\coc{z}}$ maps $1$-cocycles $\coc{z}$, equivalent in $\bop{\hs{H}}$, into equivalent unitary representations $\pi_{\coc{z}}$ of $\shomgroup{P}$ in $\hs{H}$; up to equivalence this map is injective. If $\shomgroup{P} = 1$, then $\cocyclecat{\net{R}{P}} = \tcocyclecat{\net{R}{P}}$.
\end{theorem}
\noindent The previous theorem is important because it sheds some light on this question: can we find path-dependent $1$-cocycles ? 
Suppose we are in the topological case, namely $\poset{P} \in\indexset{X}$ where $X$ is a topological (Hausdorff, arcwise connected) space; then if $X$ is simply connected it follows from theorem \ref{th:homotopy-groups} that $\shomgroup{P} = 1$, so as a trivial consequence of theorem \ref{th:cocycle-representations} we have 
\begin{corollary}
If $X$ is simply connected, any $1$-cocycle is trivial in $\bop{\hs{H}}$, namely path-independent. In other words $\cocyclecat{\net{R}{P}} = \tcocyclecat{\net{R}{P}}$.
\end{corollary}
\noindent So we conclude that the only topological obstruction to path-independence of $1$-cocycles can be non-simple connectedness of the topological space $X$; actually this is a real obstruction since in chapter \ref{chap:topological_cocycles} we will exhibit a topological (i.e. path-dependent) $1$-cocycle living in the cylindric flat universe.

\chapter{QFT on the cylindric Einstein universe}\label{ch:QFT on the cylindric flat spacetime}
In this section we discuss the QFT model on the top of which we'll construct, in the last chapter, our topological cocycles.
After developing the model we introduce the vacuum state and prove various properties of the reference vacuum representation, including DHR sectors anf Haag duality.

\section{The Cauchy problem for the field equation}
On the road to quantization the first step is to choose a classical wave equation on the background spacetime, namely that satisfied by the classical field; applying to it a suitable quantization scheme gives rise, at least in the algebraic approach \cite{Haag_96}, to the net of local algebras encoding the fundamental observables of the theory.
For our purposes we choose to start from the scalar massive Klein-Gordon equation, whose intrinsic form is:
\begin{equation}\label{eq:KG_general_form}
(\square + m^{2})u = 0. 
\end{equation}
Here the $\square$ symbol denotes the d'Alembertian operator which is given in local coordinates by the expression:
\[\square = \nabla^{\nu}\nabla_{\nu} = |g|^{- \frac{1}{2}}\partial_{\mu}g^{\mu\nu}|g|^{\frac{1}{2}}\partial_{\nu},\]
where $|g|:= |\operatorname{det}\left\lbrace g_{\mu\nu}\right\rbrace|$ is the metric determinant and m$^{2} > 0$ is the mass parameter.

An essential condition for later developments is the well posedness of the Cauchy problem for the equation (\ref{eq:KG_general_form}); but what does it mean well posedness for a wave equation in a general geometric context ? Loosely speaking, we require that, given suitable initial conditions on a suitable class of spacetime hypersurfaces, there exists only one solution of the equation satisfying these initial conditions, for every hypersurface in that class.

This well-posedness requirement forces us to restrict our scope to a particular family of background spacetime manifolds, namely the globally hyperbolic ones.
We recall that a connected time-oriented Lorentzian manifold $\spacetime{M}$ is said to be \emph{globally hyperbolic} if $\spacetime{M}$ admits a Cauchy hypersurface $\sdc{C}$, i.e. a subset such that every inextendible timelike curve in $\spacetime{M}$ meets it at exactly one point. 
It turns out \cite{O'Neill_83} that a globally hyperbolic spacetime $\spacetime{M}$ admits a (not unique) smooth foliation by Cauchy hypersurfaces; that is to say $\spacetime{M}$ is isometric to $\bR\times\sdc{C}$ with metric $-\beta dt^{2}\oplus  g_{t}$, where $\beta$ is a smooth positive function, $g_{t}$ is a Riemannian metric on $\sdc{C}$ depending smoothly on $t\in\bR$ and each set $\{t\}\times\sdc{C}$ is a smooth spacelike Cauchy hypersurface in $\spacetime{M}$. 

As we said, we restrict ourselves to consider only globally hyperbolic spacetimes so the Cauchy problem for the Klein-Gordon equation is well posed. Actually, we focus on a very special class of globally hyperbolic spacetimes, called \emph{ultrastastic spacetimes} (further details can be found e.g in \cite{Verch_93}). Roughly speaking, ultrastastic spacetimes are a product space $\times$ time, so are relatively easy to deal with. 
\begin{definition}
Let $(\Sigma, \gamma)$ be a smooth, $d$-dimensional Riemannian manifold, and consider the product manifold $\spacetime{M}\equiv\bR\times\Sigma$ endowed with the Lorentzian metric $g\equiv -dt^{2}\oplus\gamma$. We call $(\spacetime{M}, g)$ the ($d+1$-dimensional) \emph{ultrastastic spacetime} foliated by $(\Sigma, \gamma)$ and $\Sigma_{t}\equiv \{t\}\times\Sigma$, $t\in\bR$ the \emph{natural foliation} of $\spacetime{M}$. 
Moreover, if $(\Sigma, \gamma)$ is complete (as a Riemannian manifold) $\spacetime{M}$ is globally hyperbolic and its natural foliation is made up of spacelike Cauchy surfaces.
\end{definition}

\noindent The well-posedness of the Cauchy problem for the Klein-Gordon equation implies \cite{Bar_Ginoux_Pfaffle_06} the existence of two continuos\footnote{With respect to the standard locally convex topologies on $\testfuncv{\spacetime{M}}{\bR}$ and $\smoothfuncv{\spacetime{M}}{\bR}$.} linear maps 
$E^{\pm}: \testfunc{\spacetime{M}}\mapsto\smoothfunc{\spacetime{M}}$ uniquely determined by the following properties:
\[
(\square + m^{2})E^{\pm} f = f = E^{\pm}(\square + m^{2})f,\quad \forall\, f\in\testfunc{\spacetime{M}},
\]
and 
\[
\supp{E^{\pm} f}\subset J^{\pm}(\supp{f}).
\]
They are called the \emph{advanced (+) and retarded (-) fundamental solutions} of the Klein-Gordon equation. $E:=E^{+} - E^{-}$ is called the \emph{causal propagator} of the Klein-Gordon equation and it follows from the definition that 
\[
(\square + m^{2})E f = 0 = E(\square + m^{2})f,\quad \forall\, f\in\testfunc{\spacetime{M}}.
\]
So we see that $E$ maps $\testfunc{\spacetime{M}}$ to the set $\KGsol$ of smooth  solutions of the Klein-Gordon equation compactly supported on each Cauchy surface (and one can show \cite{Dimock_80} that this map is surjective).

\subsection{Well posedness}
As we said, a crucial fact for the theory we are going to develop is well-posedness\footnote{Well-posedness is a consequence of tha classical energy-estimate for solutions of second order hyperbolic partial differential equations; see e.g. \cite{Bar_Ginoux_Pfaffle_06}.} of the Klein-Gordon equation in a globally hyperbolic spacetime; that is to say that Cauchy data on a Cauchy surface uniquely determine a solution of (\ref{eq:KG_general_form}). 
In details, let $\spacetime{M}$ be a globally hyperbolic spacetime, and select a Cauchy surface $\sdc{C}$ belonging to a foliation: if $\spacetime{M} =\Sigma\times\bR$, then $\sdc{C}=\Sigma_{t}$ for some $t\in\bR$. Then the vector field $\partial_{t}$ is a timelike, future-pointing vector field normal to every $\Sigma_{t}$, so applying it to a spacetime function gives its normal derivative with respect to $\sdc{C}$. Now, if $\varphi$ is smooth solution of the Klein-Gordon equation, its \emph{Cauchy data} on $\sdc{C}$ are defined as the initial ``position'' $\varphi_{|\sdc{C}}$ and ``velocity'' $\partial_{t}\varphi_{|\sdc{C}}$ of the solution itself. So we can define a ``projection'' operator $P_{\sdc{C}}:\KGsol\mapsto\Cdatav{\sdc{C}}$  mapping a solution of equation (\ref{eq:KG_general_form}) to its Cauchy data:
\[
P_{\sdc{C}}(\varphi):=\varphi_{|\sdc{C}}\oplus\partial_{t}\varphi_{|\sdc{C}}.
\] 
Well posedness of equation (\ref{eq:KG_general_form}) means exactly that this map is bijective: given any pair $u_{0}\oplus u_{1}$ of Cauchy data in $\Cdatav{\sdc{C}}$ there exists exactly one smooth solution $\varphi$ of the Klein-Gordon equation such that
\[
P_{\sdc{C}}(\varphi) = u_{0}\oplus u_{1}.
\]
Furthermore, solutions of equation (\ref{eq:KG_general_form}) propagate with finite speed, i.e. if Cauchy data are supported on a set $G\subset\sdc{C}$, then the corresponding solution is supported on the causal development of this set, $\cdevel{G}$. From this fact it follows that if a solution is compactly supported on a Cauchy surface $\sdc{C}$ then it is compactly supported on any other Cauchy surface $\sdc{C}^{\prime}$. 

\section{Classical field quantization}

\subsection{Symplectic spaces}\label{sec:symplectic_spaces} 

Now we finally come to quantization. The key fact here is that the space of smooth solutions of the wave equation has a natural symplectic structure, which fact makes it possible to define local algebras, as we'll see shortly.

Actually, we can introduce several equivalent versions of this space. To start with, we consider the space $\KGsol$ of all real-valued smooth solutions of the Klein-Gordon equation such that their Cauchy data have compact support on one (hence every) Cauchy surface, endowed with the symplectic form:

\[\sigma(\varphi, \psi):= \int_{\sdc{C}}(\varphi\,\partial_{t}\psi - \psi\,\partial_{t}\varphi) d\eta_{\sdc{C}},\]
where $d\eta_{\sdc{C}}$ denotes the volume element induced by the (Riemannian) metric on $\sdc{C}$.
It's easy to see that the right hand side is independent on the choice of $\sdc{C}$, and that $\sigma$ is nondegenerate; so the pair $(\KGsol, \sigma)$ is a real symplectic space indeed.

Another option is to consider the space $\Cdata{\sdc{C}}:=\Cdatav{\sdc{C}}$ of Cauchy data living on an arbitrary but fixed Cauchy surface $\sdc{C}$ and to introduce essentially the same symplectic form as above:
\[\delta_{\sdc{C}}(u_{0}\oplus u_{1}, v_{0}\oplus v_{1}):= \int_{\sdc{C}}(u_{0}\,v_{1} - v_{0}\,u_{1}) d\eta_{\sdc{C}}.\]
We obtain another symplectic space $(\Cdata{\sdc{C}}, \delta_{\sdc{C}})$, clearly isomorphic to the former; they are related by the symplectic map: 
$P_{\sdc{C}}: \KGsol\mapsto\Cdata{\sdc{C}}$.
It can also be noted that the Cauchy data spaces associated to different Cauchy surfaces, say $\sdc{C}$ and $\sdc{C}^{\prime}$, are isomorphic. In fact the map: $P_{\sdc{C}}^{\prime}\circ P_{\sdc{C}}^{-1}:\Cdata{\sdc{C}}\mapsto\Cdata{\sdc{C}}^{\prime}$ is a symplectomorphism, being the composition of two symplectomorphism. What we do here is: ``take a pair of Cauchy data on $\sdc{C}$, make them evolve to the solution $\varphi$ and finally determine its corresponding Cauchy data on $\sdc{C}^{\prime}$''.

The definition of the third space requires a little more work. Start by taking the space of test functions on the entire spacetime and view two of them as equal if the causal propagator assumes the same value on them; in other words take the quotient $\sympls{K}:=\testfuncv{\spacetime{M}}{\bR}\setminus\ker(E)$. Then consider the following bilinear form:
\[\kappa([f],[h]):= \int_{\spacetime{M}}f\cdot(E h) d\eta,\]
where $[\cdot]:\testfuncv{\spacetime{M}}{\bR}\mapsto\sympls{K}$ is the quotient map  and $d\eta$ is the metric-induced volume measure on $\spacetime{M}$. It's easy to see that $\kappa$ is well defined and a nondegenerate symplectic form on $\sympls{K}$. We already said that the map $f\mapsto E f$ is surjective, so the map $[f]\mapsto E f$, mapping $(\sympls{K}, \kappa)$ to $(\KGsol, \sigma)$, is a bijection; moreover it's a symplectomorphism, so we conclude that $(\sympls{K}, \kappa)$ and ($\KGsol, \sigma)$ are isomorphic.

So far we have seen that $(\KGsol, \sigma)$, $(\Cdata{\sdc{C}}, \delta_{\sdc{C}})$ and $(\sympls{K}, \kappa)$ are just different implementations of the same algebraic object, the solution space of the wave equation; it's a global object, living on the entire spacetime. For the sake of algebraic quantization, though, we need local objects, in primis local observable algebras, so we would like to introduce local versions of those symplectic spaces. 

\noindent It's better to start with Cauchy data, because in this context it's clear what localization means: given an open set $\open{G}\subseteq\sdc{C}$ with compact closure, a pair of Cauchy data $u_{0}\oplus u_{1}\in\Cdata{\sdc{C}}$ is localized in $\open{G}$ if $\supp{u_{0}}\cup\supp{u_{1}}\subseteq\open{G}$. So we can consider the symplectic subspace of $(\Cdata{\sdc{C}}, \delta_{\sdc{C}})$ made up  of Cauchy data supported in $\open{G}$, $(\Cdata{\open{G}}, \delta_{\open{G}})$. This family of symplectic spaces indexed by (precompact) open subsets of a fixed Cauchy surface induces, via the symplectic map 
$P_{\sdc{C}}$, a corresponding family of symplectic subspaces of the solution space $\KGsol$, say $\KGsol_{\open{G}}$. 
What about $(\sympls{K}, \kappa)$ ? We could repeat the previous reasoning, this time using the map $f\mapsto E f$, thus obtaining a family of symplectic subspaces of the global space $(\sympls{K}, \kappa)$. How we can characterize these subspaces ? We could be tempted to identify them with the sets $\sympl{\open{O}}\equiv[\testfuncv{\open{O}}{\bR}]$, namely test functions supported in the spacetime open set $\open{O}$, modulus the equivalence relation defined by the causal propagator.
%TODO:CHIARIRE LA LOCALIZZAZIONE DELLE FUNZIONI TEST
This family has a nice property: the map $\open{O}\mapsto\sympl{\open{O}}$ has the structure of an isotonous local net, where locality means that the symplectic form $\kappa([f], [h])$ vanishes for $[f]\in\sympl{\open{O}}$ and  $[h]\in\sympl{\open{O}_{1}}$ whenever $\open{O}_{1}\subset\ccompl{\open{O}}$. 

Actually this identification fails to be true, and we need to restrict ourselves to a suitable subfamily of open sets, the collection of \emph{diamonds} based on Cauchy surfaces (see definition \ref{def:diamonds}). To be more specific, fix a Cauchy surface $\sdc{C}$, and let $\open{G}$ be a (relatively compact) open subset of $\sdc{C}$; then the symplectic spaces $(\Cdata{\open{G}}, \delta_{\open{G}})$ and $(\sympl{\DIAMOND{\open{G}}}, \kappa_{|\sympl{\DIAMOND{\open{G}}}})$ are isomorphic. This descends from the fact, proved by Dimock in \cite{Dimock_80}, that if $\open{N}$ is an open neighborhood (in $\spacetime{M}$) of $\open{G}$ it holds:
\[\sympl{\DIAMOND{\open{G}}}\subseteq\sympl{\open{N}},\]
and from the remark that $\DIAMOND{\open{G}}$ is a globally hyperbolic spacetime on its own right, equipped with the restriction of the spacetime metric $g$.

\subsection{The cylindric flat universe}\label{sec:the_cylindric_flat_universe} 
After these general preliminaries, now we discuss free quantum field theory for our wave equation (\ref{eq:KG_general_form}) on the cylindric Einstein universe, namely the ultrastatic spacetime $\spacetime{M}$ foliated by the 1-dimensional torus $\bS^{1}$ equipped with the standard euclidean metric; in other words $\spacetime{M}=\bS^1\times \bR$, and denoting by $\theta \in [-\pi,\pi]$ (with identified endpoints) the standard coordinate over $\bS^1$ and $t\in\bR$, the metric reads:  
$$g = - dt\otimes dt +  d\theta \otimes d \theta\:.$$
\noindent In explicit form the equation of motion is thus:
\begin{equation}\label{eq:KG-cylindric}
(-\partial_t^2 + \partial^2_\theta - m^2) \varphi(t,\theta) = 0. 
\end{equation}
Let $\sdc{C}\cong\bS^1$ be a Cauchy surface in the chosen foliation; then $\partial_{t}$ is a timelike smooth field, normal to $\sdc{C}$. For future convenience we also fix a {\em positive rotation} orientation for $\bS^1$.

Our choice of the background spacetime is dictated by a straightforward principle: we need a manifold simple enough to explicitly working out the details. In this respect, the cylindric flat universe is the simplest choice after Minkowski spacetime: in fact it's just a spatial-compactified  Minkowski spacetime.
As we anticipated above, we can easily quantize the classical Klein-Gordon field by considering the symplectic space $\KGsol$ of real smooth solutions of the equation of motion (\ref{eq:KG-cylindric}).
 
We have at our disposal three different flavours of this space: we choose to work with the Cauchy data version $(\Cdata{\sdc{C}}, \delta_{\sdc{C}})$.
If $\varphi\in\KGsol$, then we denote by ($\Phi,\Pi$) its Cauchy data on $\sdc{C}$, i.e. $\Phi = \varphi_{\rest\sdc{C}}$, $\Pi = \partial_{t}\varphi_{\rest\sdc{C}}$; then the symplectic form reads:
\begin{equation}\label{eq:sigma}
\sigma\left((\Phi,\Pi), (\Phi', \Pi')\right) := \int_\sdc{C} (\Phi'\,\Pi - \Phi\,\Pi') d\theta. 
\end{equation}
\noindent Then we have to fix a poset indexing local observable  algebras. For the sake of simplicity we take a two-steps approach: first we develop the theory with respect to a ``spatial'' poset, i.e. a family of index sets lying on the chosen Cauchy surface $\sdc{C}$; subsequently we construct a ``spacetime'' poset from the spatial one and extend the results obtained in the spatial case, at least as far as possible.

What about the spatial poset $\poset{R}$ ? We choose the simplest one, namely the class of open proper intervals of $\bS^1$.
A \emph{proper interval} of $\bS^1$ is a connected subset $\open{I}\subset \bS^1$  such that both $\Int{\open{I}}$ and $\Int{\bS^1\setminus\open{I}}$ are nonempty. If $\open{I}\in\poset{R}, \Cdata{\open{I}}$ will denote the symplectic subspace of $\Cdata{\sdc{C}}$ of Cauchy data supported in $\open{I}$.
It's worth noting here that $\poset{R}$ is endowed with a causal disjointness relation: two intervals $\open{I}$, $\open{I}_{1}$ are causally disjoint, $\open{I}\perp\open{I}_{1}$, if their intersection is empty. Specializing definitions given in section \ref{sec:causal_disjointness} to the present case we obtain:
\begin{definition}
Let $\open{I}\in\poset{R}$; we define the \emph{causal complement} of $\open{I}$  as the family of sets $\ccompl{\open{I}} := \{\open{I}_{1}\in \poset{R}\,|\, \open{I}_{1}\perp \open{I}\}$.
\end{definition}
\noindent Note that the open set $\sdc{C}^{\ccompl{\open{I}}}:=\cup_{\open{I}_{1}\in \ccompl{\open{I}}} \open{I}_{1}$ coincides with the interior of the ordinary set complement of $\open{I}$, namely $\sdc{C}^{\ccompl{\open{I}}}=\Int{\bS^{1}\setminus\open{I}}$. For the sake of simplicity, we'll adopt henceforth  the shortened notation $\open{I}^{\prime}\equiv\sdc{C}^{\ccompl{\open{I}}}$.
By construction ${\open{I}^{\prime\prime}} = \open{I}$ and $\open{I}^{\prime}\in\poset{R}$ when $\open{I}\in\poset{R}$.

\noindent\textbf{Remark.} The support properties of Cauchy data of solutions of the Klein-Gordon equation depend on the chosen Cauchy surface: if, for example, the Cauchy data $(\Phi, \Pi)$ of $\varphi\in\KGsol$ belongs to $\Cdata{\open{I}}$ for some $\open{I}\in\poset{R}$ with respect to a Cauchy surface $\sdc{C}$, Cauchy data of the same solution $\varphi$ generally fail to fulfill this condition when referring to another Cauchy surface $\sdc{C}^{\prime}$ sufficiently far in time from the former.

\subsection{Local algebras}
In the previous section we have constructed a poset-indexed family $\Cdata{\open{I}}$ of symplectic spaces; now we use these spaces to construct a poset-indexed family of C*-algebras, namely the Weyl algebras built on them. We recall that, given a symplectic space $(V, \omega)$, there exists a $C^{*}$-algebra $\alg{A}$ and a map $W:V\mapsto\alg{A}$ such that, for all $\varphi,\psi\in V$:
\begin{description}
 \item[(i)]$ W(0) = 1$,
 \item[(ii)] $W(-\varphi) = W(\varphi)^{*}$,
 \item[(iii)] $W(\varphi)\cdot W(\psi) = e^{-\imath\,\omega(\varphi, \psi)/2}\,W(\varphi + \psi )$,
 \item[(iv)] $\alg{A}$ is generated, as a $C^{*}$-algebra, by the elements $W(\varphi)$.
 \end{description}
It turns out \cite{Bratteli_Robinson_II} that the pair $(\alg{A}, W)$ is unique up to isomorphisms and is called the \emph{CCR-representation} of $(V,\omega)$, denoted $\operatorname{CCR}(V,\omega)$.
$\alg{A}$ is said to be the \emph{Weyl algebra} built from $(V,\omega)$.
Now, coming back to our case, we see that each symplectic space $\Cdata{\open{I}}$ gives rise to a $C^*$-algebra $\sweyl{\open{I}}$ localized in $\open{I}$, namely its associated Weyl algebra; in the quantization scheme discussed e.g in \cite{Dimock_80} these $\sweyl{\open{I}}$ play the role of \emph{local (abstract) observable algebras}. If we denote by $\alg{W}_{KG}$ the global algebra associated to $\Cdata{\sdc{C}}$, we see that each $\sweyl{\open{I}}$ is a subalgebra of $\alg{W}_{KG}$.

Notice that all the subalgebras share the same unit element and the following two properties are valid:\\
 {\bf isotony}: $\sweyl{\open{I}}\subset\sweyl{\open{J}}$       		       if  $\open{I}\subset\open{J}$,\\
 {\bf spatial locality}: $\left[ \sweyl{\open{I}}, \sweyl{\open{J}} \right] =0$        if $\open{I}\perp\open{J}$.\\
Strictly speaking the family $\{\sweyl{\open{I}}\}_{\open{I}\in\poset{R}}$  isn't a {\em net} of $C^*$-algebras\footnote{Nevertheless we'll systematically adopt this  term with a slight language's abuse.} because the poset $\poset{R}$ is not directed with respect to the partial order relation given by set inclusion (there are pairs $\open{I},\open{J}\in\poset{R}$ with $\open{K}\not\supset\open{I} ,\open{J}$ for every $\open{K}\in\poset{R}$), and thus it is not possible to take the inductive limit defining the overall quasi-local ($C^*$-) algebra containing every $\sweyl{\open{I}}$.
Anyway an {\em universal algebra} $\alg{A}$  generated by $\left\lbrace \sweyl{\open{I}}\right\rbrace_{\open{I}\in\poset{R}}$ can be defined and  $\alg{A}\supset\alg{W}_{KG}$ (see appendix \ref{app:universal_algebras}).

\subsection{Spacetime symmetries}\label{sec:spacetime_symmetries} 

A word  is in order here about spacetime symmetries. Consider a Cauchy surface $\sdc{C}$ belonging to the natural foliation of $\spacetime{M}$; then
$\sdc{C}$ is metrically invariant under the action of $\bR$ viewed as a $\sdc{C}$-isometry group: $r\in \bR$ induces the isometry
$\beta_r :\theta \mapsto \theta + r$.
If the pull-back $\beta_r^*$ is defined as $(\beta^*_rf)(\theta) := f(\theta-r)$ for all $f\in\testfuncv{\bS^{1}}{\bR}$, the $\sdc{C}$-isometry group $\bR$  can be represented in terms of a one-parameter group $\{\alpha_r\}_{r\in\bR}$ of $*$-automorphisms of $\alg{W}_{KG}$, uniquely induced
by 
\begin{equation}\label{eq:isometries}
\alpha_r \left( W(\Phi,\Pi)\right) := W\left(\beta^*_r \Phi, \beta^*_r \Pi\right) \:, \quad\forall\, r\in \bR,\:
(\Phi,\Pi) \in \Cdata{\sdc{C}}.
\end{equation}
 The existence of such $\{\alpha_r\}_{r\in\bR}$ follows immediately from the fact that $\sigma$ is invariant under every $\beta_r^*$ (see e.g. \cite[Proposition (5.2.8)]{Bratteli_Robinson_II}) . 

Now let $\varphi$ be a real smooth solution  of the wave equation and take $s\in \bR$; $\varphi_s$ is the future-translation of $\varphi$ by a time interval $s$, in the sense that $\varphi_s(t,\theta) := \varphi(t-s,\theta)$ for all $t\in \bR$ and $\theta\in \bS^1$. Notice that $\varphi_s$ is again solution of the Klein-Gordon equation because the spacetime is static.
Passing to the Cauchy data (on the {\em same} Cauchy surface at $t=0$), this procedure induces a one-parameter group of transformations $\mu_s: \Cdata{\sdc{C}}\to \Cdata{\sdc{C}}$ such that $\mu_s(\Phi,\Pi)$  are the Cauchy data of $\varphi_s$ when $(\Phi,\Pi)$ are those of $\varphi$.  The maps $\{\mu_s\}_{s\in \bR}$ preserve the symplectic form due  to the invariance of the metric under time displacements. As a consequence we have a one-parameter group of $*$-isomorphisms $\{\tau_s\}_{s\in \bR}$ acting on $\alg{W}_{KG}$  and uniquely defined by the requirement 
\begin{equation}\label{eq:isometries2} 
\tau_s \left( W(\Phi,\Pi)\right)= W\left(\mu_s (\Phi,\Pi)\right) \:, \quad\forall\, r\in \bR,\:
(\Phi,\Pi) \in \Cdata{\sdc{C}}. 
\end{equation}
The groups $\{\alpha_r\}_{r\in\bR}$ and $\{\tau_s\}_{s\in \bR}$ can be combined into an Abelian group of $*$-automorphisms $\{\gamma_{(r,s)}\}_{(r,s)\in \bR^2}$ of $\alg{W}_{KG}$ with $\gamma_{(r,s)} := \alpha_r \circ \tau_s$.
This group represents the action of the unit connected component of the Lie group of spacetime isometries on the Weyl algebra associated with the quantum field.

Solutions of the wave equation  with Cauchy data in $\sdc{C}$ supported in $\open{I}\in\poset{R}$ propagate in $\spacetime{M}$ inside
the subset $\cfuture{\open{I}}\cap\cpast{\open{I}}$ as is well known. Therefore one concludes that if $(\Phi,\Pi)$ is supported in $\open{I}\in \poset{R}$, $\mu_s (\Phi,\Pi)$ is supported in the interval $\open{I}_s \subset \bS^1$ constructed as follows.
Passing to the new variable $\theta':= \theta +c$ for some  suitable constant $c\in \bR$,
one can always represent $\open{I}$ as $(-a,a)$ with $0<a<\pi$. In this representation  $\open{I}_s := (-a-|s|, a+|s|)$
taking the identification
$-\pi \equiv \pi$ into account. Notice in particular that, for $\open{I}\in \poset{R}$, one has  $\open{I}_s \in \poset{R}$ if and only if  $|s| <\pi - \ell(\open{I})/2$
(where $\ell(\open{I})$ is the length of $\open{I} \in \poset{R}$ when $\ell(\bS^1) = 2\pi$), whereas it turns out that $\open{I}_s = \bS^1$
whenever $|s|> \pi - \ell(\open{I})/2$.

\section{Vacuum representations.}\label{sec:vacuum_representations}
%Quasifree states
Now we want to introduce an important class of states on Weyl algebras, namely \emph{quasifree states}. They turn out to enjoy remarkable properties and, above all, the privileged vacuum state we are going to discuss later is a quasifree state.
En passant we also introduce the notion of a two-point function, although we won't exploit it in what follows.
\subsection{Quasifree states}
\label{sec:quasifree_states}
Our subsequent constructions rely heavily on the concept of GNS representation; although it's a well established notion, for the sake of completeness we recall here the statement of the main theorem; further details can be found in \cite{Bratteli_Robinson_I}.

\begin{theorem}[GNS representation]
Let $\omega$ be a state over the $C^{*}$-algebra $\alg{A}$; then there exists a cyclic representation $(\hs{H}_{\omega}, \pi_{\omega}, \Omega_{\omega})$ (called the \emph{GNS representation} of $\alg{A}$ induced by the state $\omega$) such that:
\[ \omega(A) = \ip{\Omega_{\omega}}{\pi_{\omega}(A)\Omega_{\omega}}_{\hs{H}_{\omega}},\]
for all $A\in\alg{A}$ and, consequently, $\norm{\Omega_{\omega}}^{2} = \norm{\omega} = 1$. Moreover this representation is unique up to unitary equivalence.
\end{theorem}
So, let $\omega$ be a state on the Weyl algebra $\weylalg{\sympls{Z}}{\xi}$ associated to some symplectic space $(\sympls{Z}, \xi)$ via the map $W:\sympls{Z}\mapsto\weylalg{\sympls{Z}}{\xi}$ and let $(\hs{H}_{\omega}, \pi_{\omega}, \Omega_{\omega})$ be the GNS representation of $\omega$. Assume that for every $z\in\sympls{Z}$ the unitary one-parameter group $t\mapsto\pi_{\omega}(W(tz))$, $t\in\bR$, is strongly continuos and that $\Omega_{\omega}$ belongs to the domain of definition of its generator $\Phi_{\omega}(z)$. Then we call the $\bC$-bilinear form $\lambda_{\omega}$ on $\sympls{Z}$ defined as:

\[\lambda_{\omega}(z,\tilde{z}):=\ip{\Phi_{\omega}(z)\Omega_{\omega}}{\Phi_{\omega}(\tilde{z})\Omega_{\omega}} \]

\noindent the \emph{two-point function} of $\omega$. In the case (of interest to us) that $(\sympls{Z}, \xi)=(\sympls{K},\kappa)$ we call

\[\Lambda_{\omega}(f,h):=\lambda_{\omega}([f],[h]),\quad\forall\,f,h\in\testfunc{\spacetime{M}},\]

\noindent the \emph{spacetime two-point function of} $\omega$. It can be shown by direct inspection that  $\Lambda_{\omega}$ is a \emph{bi-solution} of the Klein-Gordon equation, i.e.

\[\Lambda_{\omega}((\square + m^{2})f, h) = 0 = \Lambda_{\omega}(f, (\square + m^{2})h),\quad\forall\,f,h\in\testfunc{\spacetime{M}}.\]

\noindent Now assume that the two-point function $\lambda_{\omega}$ of $\omega$ exists; then the following definition holds.
\begin{definition}
If $\omega$ is a state on the Weyl algebra $\weylalg{\sympls{Z}}{\xi}$ of some symplectic space $(\sympls{Z}, \xi)$, it's called \emph{quasifree} if there exists a real scalar product $\mu_{\omega}$ on $\sympls{Z}$ such that, for all $z,\tilde{z}\in\sympls{Z}$:\\
1. $\lambda_{\omega}(z,\tilde{z}) = \mu_{\omega}(z,\tilde{z}) +\frac{\imath}{2}\,\xi(z,\tilde{z})$,\\
2. $[\xi(z,\tilde{z})]^{2}\leq 4\mu_{\omega}(z,z)\mu_{\omega}(\tilde{z},\tilde{z})$,\\
3. $\omega(W(z)) = \exp[-\frac{1}{2}\,\mu_{\omega}(z,z)]$.
\end{definition}
\noindent Quasifree states admit a complete characterization \cite{Kay_Wald_91} in terms of the notion of \emph{one-particle structure}. First of all, we need some notation: given a complex Hilbert space $\hs{H}$, consider its (symmetric) Fock space $\fock{\hs{H}}$ and for all $\chi\in\hs{H}$ construct the operator 

\[W^{F}(\chi):=\exp[\overline{a(\chi)-a^{*}(\chi)}],\]
where $a^{*}$, $a$ are the creation/annihilation operators acting on $\fock{\hs{H}}$. $\Omega^{F}\equiv 1\oplus 0 \oplus 0\ldots$ will denote the Fock vacuum.
\begin{definition}
Let $\omega$ be a state on $\weylalg{\sympls{Z}}{\xi}$. A \emph{one-particle Hilbert space structure} for $\omega$ is pair $(\mathbf{k}, \hs{H})$ where $\hs{H}$ is a complex Hilbert space and $\mathbf{k}:\sympls{Z}\mapsto\hs{H}$ a real-linear injective map, with the properties (for all $z,\tilde{z}\in\sympls{Z}$):\\
1. $\mathbf{k}(\sympls{Z}) + \imath\,\mathbf{k}(\sympls{Z})$ is dense in $\hs{H}$,\\
2. $\ip{\mathbf{k}(z)}{\mathbf{k}(\tilde{z})}_{\hs{H}} = \lambda_{\omega}(z,\tilde{z}) = \mu_{\omega}(z,\tilde{z}) +\frac{\imath}{2}\,\xi(z,\tilde{z})$.
\end{definition}
\noindent We are now in the position to give the quasifree states characterization we have spoken about:
\begin{theorem}\label{th:quasifree}
Let $\omega$ be a quasifree state on $\weylalg{\sympls{Z}}{\xi}$. Then there exists, uniquely up to unitary equivalence, a  one-particle Hilbert space structure $(\mathbf{k}, \hs{H})$ for $\omega$, such that the GNS representation $(\hs{H}_{\omega}, \pi_{\omega}, \Omega_{\omega})$ is given by $(\fock{\hs{H}}, \pi^{F}, \Omega^{F})$, where:
\[\pi^{F}(W(z)):= W^{F}(\mathbf{k}(z)),\quad \forall\, z\in\sympls{Z}.\]
Moreover, $\omega$ is pure if and only if the range of $\mathbf{k}$ is dense in $\hs{H}$.
\end{theorem}
\noindent If $\hs{L}$ is a real-linear subspace of $\hs{H}$ we write

\[W(\hs{L})\equiv\left\lbrace W^{F}(\chi)|\chi\in\hs{L}\right\rbrace ^{\prime\prime}\]
where the double prime accent denotes the weak operatorial closure in $\bop{\fock{\hs{H}}}$.

Now let us specialize this construction to QFT. Consider a Cauchy surface $\sdc{C}$ in the globally hyperbolic spacetime $(\spacetime{M}, g)$, and suppose that the global algebra $\alg{W}_{KG}$ is represented as the Cauchy-data space Weyl algebra $\weylalg{\Cdata{\sdc{C}}}{\delta_{\sdc{C}}}$. Then from theorem \ref{th:quasifree} we know that a quasifree state $\omega$ on $\weylalg{\Cdata{\sdc{C}}}{\delta_{\sdc{C}}}$ is characterized by a one-particle Hilbert space structure  $(\mathbf{k}, \hs{H})$, where $\mathbf{k}$ satisfies 
$$2 \Im m\,\ip{\mathbf{k}(u_{0}\oplus u_{1})}{\mathbf{k}(v_{0}\oplus v_{1})}_{\hs{H}} = \delta_{\sdc{C}}(u_{0}\oplus u_{1},v_{0}\oplus v_{1}),\quad u_{i}, v_{i}\in\testfuncv{\sdc{C}}{\bR},$$
and it gives rise to the correspondence:

\[\Cdata{\open{I}}\mapsto\mathbf{k}(\Cdata{\open{I}})=:\hs{L}(\open{I})\subset\hs{H},\]
where $\open{I}$ denotes an element of the spatial poset on $\sdc{C}$.
In terms of local von Neumann algebras it turns out that:
\[\alg{W}(\hs{L}(\open{I}))= \snet{R}{\open{I}},\]
\noindent where we set $\snet{R}{\open{I}}:=\pi^{F}(\weylalg{\Cdata{\open{I}}}{\delta_{\open{I}}})''$.
Finally we restrict ourselves to ultrastastic spacetimes. In this context we can explicitly construct a quasifree, time-invariant state, that we identify with the (unique) vacuum in the theory.
So consider a $d$-dimensional complete Riemannian manifold $(\Sigma, \gamma)$ with Laplace-Beltrami operator $\Delta_{\gamma}$ and metric-induced measure $\mu_{\gamma}$; then for fixed $m  > 0$ the Klein-Gordon differential operator 

\[-\Delta_{\gamma} + m^{2}: \testfuncv{\Sigma}{\bC}\mapsto\smoothfuncv{\Sigma}{\bC}\] 
is essentially selfadjoint \cite{Chernoff_73}; we denote its closure by $A$. Let $(\spacetime{M}, g)$ be the ultrastastic spacetime foliated by $(\Sigma, \gamma)$, and $\Sigma(t)\equiv\sdc{C}_{t}$, $t\in\bR$, the canonical foliation. For each  $t\in\bR$ define a quasifree state $\omega^{t}$ on 
$\weylalg{\Cdata{\sdc{C}_{t}}}{\delta_{\sdc{C}_{t}}}$ by setting its one-particle Hilbert space structure to $(\mathbf{k}^{t}, \hs{H}^{t})$, where $$\hs{H}^{t}:=\csqint{\Sigma}{\mu_{\gamma}}$$ 
and
\[\mathbf{k}^{t}(u_{0}\oplus u_{1}):=\frac{1}{\sqrt{2}}(A^{\frac{1}{4}}u_{0} + \imath A^{-\frac{1}{4}}u_{1}),\quad u_{0}\oplus u_{1}\in\Cdata{\sdc{C}_{t}}.\]
It turns out that $\omega^{t}$ is pure and invariant under time translations $\tau_{s}$ on $\alg{W}_{KG}$. So we may drop the superscript $t$ denoting $\omega^{t}$ by $\omega$, and calling it the \emph{canonical vacuum state} on the Weyl algebra $\alg{W}_{KG}$ of the (mass $m$) Klein-Gordon field on the ultrastastic spacetime $(\spacetime{M}, g)$.

It's time to introduce the vacuum state for the cylindric flat universe; since it's an ultrastastic spacetime, we have just seen that there exists a unique time-invariant quasifree state, and we also know how to construct it: we have to exploit the generator of the equation of motion, namely eq. (\ref{eq:KG-cylindric}).
It's the  positive symmetric operator 
$$-\frac{d^2}{d\theta^2} + m^2 I : \: \smoothfuncv{\bS^1}{\bC} \to \csqint{\bS^1}{\theta}\:.$$
acting on the complex Hilbert space $\csqint{\bS^1}{\theta}$.
 It is essentially self-adjoint since $\smoothfuncv{\bS^1}{\bC}$ contains a dense set of analytic vectors made of exponentials $\theta \mapsto e^{in\theta}$, $n\in \bZ$, which are the eigenvectors of the operator.
The unique self-adjoint extension of this operator, i.e. its closure, will be denoted by $A : \operatorname{Dom}(A) \to \csqint{\bS^1}{\theta}$. \\ 
Notice that $A$ is strictly positive (being $m>0$) and  thus its real powers $A^{\alpha}$,  $\alpha \in \bR$, are well-defined.\\
The next proposition concerns some basic properties of $A$ and its powers.
 
\begin{proposition} \label{prop:hyperbolic_generators}
The  operators $A^\alpha : \operatorname{Dom}(A^\alpha) \to \csqint{\bS^1}{\theta}$ for $\alpha\in \bR$
enjoy the following properties:\\
{\bf (a)} $\sigma(A^\alpha) = \{(n^2+ m^2)^\alpha\:|\: n= 0,1,\ldots\}$.\\
{\bf (b)} $\overline{\operatorname{Ran}(A^\alpha)} = \csqint{\bS^1}{\theta}$.\\
{\bf (c)} $A^\alpha$ commutes with the standard conjugation $C: \csqint{\bS^1}{\theta} \to \csqint{\bS^1}{\theta}$ with $(Cf)(\theta) := \overline{f(\theta)}$; furthermore $A^\alpha(\smoothfuncv{\bS^1}{\bR}) = \smoothfuncv{\bS^1}{\bR}$ so that $\overline{A^\alpha (\smoothfuncv{\bS^1}{\bC})}
= \csqint{\bS^1}{\theta}$.\\
{\bf (d)}  If $\alpha \leq 0$, $A^{\alpha} : \csqint{\bS^1}{\theta} \to \operatorname{Dom}(A^{-\alpha})$ are bounded with $||A^{\alpha}|| = m^{2\alpha}$.\\
\end{proposition}

\begin{proof}
Consider the family of operators $B^\alpha : \{c_n\}_{n\in \bZ} \mapsto \{(n^2+ m^2)^\alpha c_n\}_{n\in \bZ}$ in $\ell^2(\bZ)$. These are self-adjoint with domains $\operatorname{Dom}(B^\alpha):= \{ \{c_n\}_{n\in \bZ} \in \ell^2(\bZ)\:| \: \{(n^2+ m^2)^\alpha c_n\}_{n\in \bZ} \in \ell^2(\bZ) \}$
so that the spectra are $\sigma(B^\alpha) = \{(n^2+ m^2)^\alpha\:|\: n= 0,1,\ldots\}$. 
By construction $B^\alpha = (B^1)^\alpha$ in the sense of spectral theory. Using the unitary operator  $U: \csqint{\bS^1}{\theta} \to \ell^2(\bZ)$ which associates a function with its Fourier series, one realizes that $U^{-1} B^1 U$ extends $-\frac{d^2}{d\theta^2} + m^2I$
and thus its unique self-adjoint extension $A$ must coincide with $U^{-1} B^1 U$. As a consequence $U^{-1}B^\alpha U = A^\alpha$.
The operators $B^\alpha$  satisfy the corresponding properties in $\ell^2(\bZ)$ of (a), (b), (d) by construction ((b) follows from the fact that
$B^\alpha = (B^{\alpha})^*$, $\operatorname{Ker}(B^{\alpha})$ is trivial and $\overline{\operatorname{Ran}(S)} = \operatorname{Ker}(S^*)^\perp$ for every densely-defined operator $S$ acting on a Hilbert space), so that $A^\alpha$ satisfy  (a), (b), (d).

The proof of (c) is the following. The first statement is an obvious consequence of the Fourier representations of $C$ and
that of $A^\alpha$, $B^\alpha$. Let us come to the second statement. If $f \in \smoothfuncv{\bS^1}{\bC}$ its Fourier coefficients $c_n$ vanish as $n\to \pm \infty$ faster than any power of $n$ and thus they define an element of each $\operatorname{Dom}(B^\alpha)$. Moreover, if $\psi \in \smoothfuncv{\bS^1}{\bC}$, the Fourier coefficients of $A^\alpha \psi$ vanish faster than every powers of $n$ and thus the associated Fourier series converges uniformly with all of $\theta$-derivatives of any order and therefore $A^\alpha \psi$ individuates a function in $\smoothfuncv{\bS^1}{\bC}$; this proves that $A^\alpha (C^\infty(\bS^1, \bR)) \subset \smoothfuncv{\bS^1}{\bR}$. The other inclusion follows trivially applying $A^{-\alpha}$ to both sides and using arbitrariness of $\alpha$.  To extend  this result to $\smoothfuncv{\bS^1}{\bC}$ it suffices exploiting the fact that $A^\alpha$ commute with $C$. The final statement is now  immediate:
$\overline{A^\alpha (\smoothfuncv{\bS^1}{\bC})} = \overline{\smoothfuncv{\bS^1}{\bC}}
= \csqint{\bS^1}{\theta}$.
\end{proof}
\noindent Following the general theory discussed in section \ref{sec:quasifree_states}, we exploit the operator $A$ to define a one-particle Hilbert space structure for the vacuum. Fixing a reference Cauchy surface $\sdc{C}$, for every $(\Phi,\Pi) \in\Cdata{\sdc{C}}$ the quantization map $K :\Cdata{\sdc{C}} \to \csqint{\bS^1}{\theta}$ is defined as, :
\begin{equation}\label{eq:1-p struct}
K(\Phi,\Pi) := \frac{1}{\sqrt{2}} \left(A^{1/4} \Phi + \imath A^{-1/4} \Pi \right), 
\end{equation}
and, together with the Hilbert space $\csqint{\bS^1}{\theta}$, it's the one-particle structure we looked for.
This map will be useful  shortly to determine a preferred unitary irreducible (Fock) representation of the Weyl algebra called the {\em vacuum representation}. 

A natural physical way to introduce $K$ is noticing that the solution $\varphi$ of equation (\ref{eq:KG_general_form}) with Cauchy data $(\Phi,\Pi) \in\Cdata{\sdc{C}}$,
 interpreting the $t$ derivative in the sense of $\csqint{\bS^1}{\theta}$ topology, can be written as 
\begin{equation}\label{eq:solform}
 \varphi(t,\cdot) = \frac{1}{\sqrt{2}}e^{-\imath tA^{1/2}} A^{-1/4} K(\Phi,\Pi) + C \frac{1}{\sqrt{2}}e^{-\imath tA^{1/2}} A^{-1/4} K(\Phi,\Pi),
\end{equation}
 $C : \csqint{\bS^1}{\theta} \to \csqint{\bS^1}{\theta}$ being the standard complex conjugation.
The proof is a trivial consequence of Stone theorem and (d)  of proposition (\ref{prop:hyperbolic_generators}).
The right-hand side of (\ref{eq:solform}) turns out to be  $(t,p)$-jointly  smooth and the $t$ derivative coincides with  that in the $L^2$ sense  \cite{Wald_94}. Thus, by the uniqueness theorem  for solutions of the Klein-Gordon equation with compactly supported data in globally hyperbolic spacetimes,
 varying $t\in \bR$ the right-hand side of (\ref{eq:solform}) defines the proper solution of the wave equation individuated by the Cauchy data $(\Phi,\Pi)$. From (\ref{eq:solform}), interchanging $A^{\pm 1/4}$ with $e^{-\imath tA^{1/2}}$ it arises that $A^{1/2}$ can be seen as the Hamiltonian generator of Killing time displacements, acting on the Hilbert space  of the quantum wavefunctions $K(\Phi,\Pi)$ associated
with the classical solutions with Cauchy data $(\Phi,\Pi)$; that Hilbert space is the so called {\em one-particle space}.
This is the perspective necessary to understand the construction presented in theorem \ref{th:vacuum} from a physical point of view. \\ 
The following fundamental statement about the range of $K$ holds true.
 
\begin{proposition} \label{prop:sympl-scalar}
 With the given definitions for $\Cdata{\sdc{C}}$ and $K$ the following facts are valid.\\
{\bf (a)} The range of $K$  is dense in $\csqint{\bS^1}{\theta}$.\\
{\bf (b)} For every pair $(\Phi, \Pi), (\Phi',\Pi') \in \Cdata{\sdc{C}}$ it holds
\begin{equation}\label{eq:Ksigma}
-\frac{1}{2}\delta_{\sdc{C}}\left((\Phi, \Pi), (\Phi',\Pi') \right) =
 \Im m\ip{K(\Phi, \Pi)}{K(\Phi', \Pi')}; 
 \end{equation}
as a consequence $K$ is injective. 
\end{proposition}
\begin{proof}
(a) is a straightforward consequence of (c) in proposition \ref{prop:hyperbolic_generators}.
The proof of (b) is obtained by direct inspection. Injectivity of $K$ descends immediately from (\ref{eq:Ksigma}) and non-degenerateness of $\sigma$.
\end{proof}

 \noindent Let us construct the {\em vacuum GNS representation} using proposition \ref{prop:sympl-scalar}. Let us remind some terminology. In the following, if $\{\alpha_g\}_{g\in G}$ is a representation of a group $G$ in terms of $*$-automorphisms of an unital  $*$-algebra $\alg{A}$, a state $\lambda: \alg{A} \to \bC$ will be said to be {\bf invariant} under $\{\alpha_g\}_{g\in G}$ if one has $\lambda\left(\alpha_g(a)\right) = \lambda(a)$ for all $g\in G$ and $a\in \alg{A}$.
 Moreover a representation $\{U_g\}_{g\in G}$, where every $U_g$ is a unitary operator defined over the GNS Hilbert  space $\hs{H}_\lambda$ of $\lambda$,
is said to {\bf implement} $\{\alpha_g\}_{g\in G}$ if $\pi_\lambda\left(\alpha_g(a)\right) = U_g \pi_\lambda(a)  U^*_g$ for all $g\in G$ and $a\in \alg{A}$, where $\pi_\lambda$ is the GNS representation of $\alg{A}$ induced by $\lambda$.

\begin{theorem}\label{th:vacuum}
  With the given definitions for $\Cdata{\sdc{C}}$ and $K$, the following facts are valid.\\
{\bf (a)} There is a pure quasifree state $\omega_0 : \alg{W}_{KG} \to \bC$  uniquely induced by linearity and continuity by
\begin{equation} \label{eq:omega}
 \omega_0(W(\Phi,\Pi)) = e^{-\frac{1}{2}\ip{K(\Phi,\Pi)}{K(\Phi,\Pi)}} \: \quad \forall\,(\Phi,\Pi)\in \Cdata{\sdc{C}}.
\end{equation}
 {\bf (b)} The GNS representation of $\omega_0$, $(\varhs{H}_{\omega_0}, \pi_{\omega_0}, \Psi_{\omega_0})$, called the  {\bf vacuum representation} of $\alg{W}_{KG}$, is constructed as follows (up to unitarities):\\ 
 (i) $\varhs{H}_{\omega_0}$ is the symmetrized Fock space with one-particle space $\hs{H}_{\omega_0} := \csqint{\bS^1}{\theta}$;\\
 (ii) the  representation $\pi_{\omega_0}$ is isometric and is induced by linearity and continuity by:
\begin{equation}\label{expaa}
 \pi_{\omega_0}(W(\Phi,\Pi)) = e^{\overline{a(K(\Phi,\Pi)) -  a^*(K(\Phi,\Pi)) }}\:, 
% e^{-i\overline{\sigma((\hat{\Phi},\hat{\Pi}),(\Phi,\Pi))}}
\end{equation}
where $a(K(\Phi,\Pi))$, $a^*(K(\Phi,\Pi))$ are the standard creation and annihilation operators (the latter antilinear in its argument) defined in the dense subspace spanned by  vectors with finite number of particles.\\
(iii) the cyclic vector $\Psi_{\omega_0}$ is the vacuum vector of $\varhs{H}_{\omega_0}$. \\
 {\bf (c)} $\omega_0$ is invariant under $\{\gamma_{(r,s)}\}_{(r,s)\in \bR^2}$, where  $\{\gamma_{(r,s)}\}_{(r,s)\in \bR^2}$ is the Abelian group of $*$-automorphisms  representing the natural  action of the unit connected-component Lie group of isometries of $\spacetime{M}$. \\
 {\bf (d)} The unique unitary representation $\{U_{(r,s)}\}_{(r,s)\in \bR^2}$ on $\varhs{H}_{\omega_0}$
  leaving $\Psi_{\omega_0}$ invariant  and
implementing $\{\gamma_{(r,s)}\}_{(r,s)\in \bR^2}$ fulfills,
 for all $(r,s)\in \bR^2$: $$U_{(r,s)} = e^{-irP^\otimes} e^{isH^\otimes} = e^{-i(rP^\otimes -s H^\otimes)}\:,$$
where the generators $P^\otimes$, $H^\otimes$
 are respectively given by  the tensorialization of the  operators  $P$, $A^{1/2}$ on $\hs{H}_{\omega_0}$, where $P$ denotes the unique self-adjoint extension of $-i\frac{d}{d\theta} : \smoothfuncv{\bS^1}{\bC} \to \csqint{\bS^1}{\theta}$.
\end{theorem}
\begin{proof}
We'll give here only a sketch of the proof. Points (a) and (b) are immediate consequences of proposition 3.1 and lemma A.2 in \cite{Kay_Wald_91}.
 The fact that $\omega_0$ is pure follows from the cited propositions because  $\overline{K(\Cdata{\sdc{C}})} = \csqint{\bS^1}{\theta} =  \hs{H}_{\omega_0}$ as a consequence of Proposition \ref{prop:sympl-scalar}.
 Invariance of $\omega_0$ under $\gamma$ arises by direct inspection. 
A known theorem (see \cite{Araki_99}) establishes, as a consequence, the existence of a unique unitary representation of $\gamma$ which implements $\gamma$ leaving  $\Psi_{\omega_0}$ invariant. 
 The remaining statements are based on standard arguments valid for QFT in static spacetimes and trivial generalizations of the analogs in Minkowski spacetime.
\end{proof}
For notational convenience, from now on we denote $\varhs{H}_{\omega_0}, \pi_{\omega_0}, \Psi_{\omega_0}$ with the symbols $\varhs{H}_0, \pi_0, \Psi_0$ respectively.

\subsection{From abstract algebras to concrete algebras}\label{sec:abstract -> concrete algebras} 
Having selected a reference state, namely the vacuum state $\omega_{0}$, we are now in the position to realize our abstract (Weyl) local algebras as operator algebras on a fixed Hilbert space. All we need to do is to use the GNS representation $(\varhs{H}_{0},\pi_{0}, \Psi_{0})$ of the vacuum state and to take, as the concrete algebras, the images  $\pi_0(\sweyl{\open{I}})$ of the local algebras $\sweyl{\open{I}}$ under the representation $\pi_{0}$ itself. Actually, we'd like to work with von Neumann algebras, so we need to take the (weak operatorial) closure of these algebras; so we end up with a poset-indexed family of von Neumann algebras $\snet{R}{\open{I}}:= \pi_0(\sweyl{\open{I}})''$, $\open{I}\in\poset{R}$.

\section{Properties of local nets}\label{sec:properties} 
\subsection{Introduction}\label{sec:intro_properties}
In this section we discuss some remarkable properties of the vacuum GNS representation $(\varhs{H}_0,\pi_0, \Psi_0)$   and  the associated  class of von Neumann algebras, $\{\snet{R}{\open{I}}\}_{\open{I}\in \poset{R}}$
where $\snet{R}{\open{I}}:= \pi_0(\sweyl{\open{I}})''$. In the following $\alg{R}_{KG} := \alg{R}(\bS^1) = \pi_0(\alg{W}_{KG})''$
 and $\bop{\varhs{K}}$ will denote the algebra of all bounded operators on the  Hilbert space $\varhs{K}$.
The class $\{\snet{R}{\open{I}}\}_{\open{I}\in \poset{R}}$   fulfills the following important properties:\\
{\bf isotony}:  $\snet{R}{\open{I}} \subset \snet{R}{\open{J}}$  if $\open{I} \subset \open{J}$,\\
 {\bf spatial locality}: $\left[ \snet{R}{\open{I}}, \snet{R}{\open{J}} \right] =0$   if $\open{I}\cap \open{J} = \emptyset$,\\
 {\bf irreducibility}:  $\alg{R}_{KG} = \bop{\varhs{H}_0}$.\\
The proof of the first  identity follows from the validity of the analog statement relative to the underlying Weyl algebras and from the fact that the von Neumann algebra  $\snet{R}{\open{I}}$  is the weak operatorial closure of the GNS representation of the Weyl algebra $\sweyl{\open{I}}$.

The proof of irreducibility is quite trivial: first of all notice that $\alg{R}_{KG}^{\prime} = \{c\cdot 1  \:|\: c\in \bC\}\equiv\bC\cdot 1$. Indeed
on one hand $\bC\cdot 1 \subseteq\alg{R}_{KG}^{\prime}$, on the other hand   $\alg{R}_{KG}^{\prime} \subseteq \pi_0(\alg{W}_{KG})' = \{c\cdot 1 \:|\: c\in \bC\}$
because  $\pi_0(\alg{W}_{KG})$ is irreducible.
Taking the commutant once again we find $\alg{R}_{KG} = \bop{\varhs{H}_0}$ because
$\alg{R}_{KG}= \alg{R}_{KG}'' =  \bC\cdot 1' = \bop{\varhs{H}_0}$. Notice that it also holds $\alg{R}_{KG}' \cap \alg{R}_{KG} = \bC\cdot 1$ and hence $\alg{R}_{KG}$ is a factor (the simplest factor of type $I_\infty$).

Our goal is to prove some relevant features of this class of von Neumann algebras. From now on, if $\open{J}\in \poset{R}$ and $\mathbf{t} = (r,s)\in \bR^2$, the set $\open{J}+\mathbf{t}$ is the subset of $\bS^1$ obtained as follows: (1) rotate $\open{J}$ by the angle $r$ positively obtaining the set $\open{I}$;
(2) pass from the set $\open{I}$ to the set $\open{I}_s$ defined in section \ref{sec:spacetime_symmetries}. The obtained set is $\open{J}+ \mathbf{t}$ by definition. \\

\noindent\textbf{Remark.} If $\open{J}\in \poset{R}$ and $\mathbf{t}=(r,s) \in \bR^2$, the length of $\open{J}+\mathbf{t}$
increases continuously with $|s|$, and it rotates positively of an angle $r$.
$\ell(\open{J}+\mathbf{t}) = \ell(\open{J})$ if and only if $s=0$.
As discussed above, $\open{J}+\mathbf{t} \in \poset{R}$ if and only if $\mathbf{t}=(r,s)$ is such that
 $|s| < \pi -\ell(\open{J})/2$. For $s=\pi -\ell(\open{J})/2$, $\open{J}+\mathbf{t}$ coincides with $\bS^1$ without a point.
Finally $\open{J}+\mathbf{t} = \bS^1$
for $|s|> \pi -\ell(\open{J})/2$.

\subsection{Additivity properties}
\noindent Now we come to the central point of this section. First of all we notice that the properties we call {\em spatial additivity} and
{\em spatial weak additivity} hold true. In the following,
if $\{\alg{A}_j\}_{j\in \open{J}}$ is a class of sub $*$-algebras  of a given $*$-algebra $\alg{A}$,
$\vee_{j\in \open{J}} \alg{A}_j$ denotes the $*$-algebra finitely generated by all the algebras $\alg{A}_j$.\\

\begin{proposition}\label{prop:additivity}
Referring to $\{\snet{R}{\open{I}}\}_{\open{I}\in \poset{R}}$, if $\open{I} \in \poset{R}$ the following properties hold.\\
{\bf (a)} {\bf Spatial additivity}: 
if  $\{\open{I}_i\}_{i\in L} \subset \poset{R}$ satisfies $\cup_{i\in L} \open{I}_i = \open{I}$ or respectively $\cup_{i\in L} \open{I}_i = \bS^1$, then
$$\left(\bigcup_{i\in L} \pi_0\left( \sweyl{\open{I}_i}\right) \right)'' = \snet{R}{\open{I}} \quad \mbox{and respectively} \quad 
\left(\bigcup_{i\in L} \pi_0\left( \sweyl{\open{I}_i}\right) \right)'' = \alg{R}_{KG}\:.$$
{\bf (b)} {\bf Spatial weak additivity}:
$$\left(\bigcup_{r\in \bR} \pi_0\left( \sweyl{\open{I}+(r,0)}\right) \right)'' = \alg{R}_{KG} \quad \mbox{and} \quad
\overline{\bigvee_{r\in \bR} \pi_0\left( \sweyl{\open{I}+(r,0)}\right)\Psi_0} = \varhs{H}_0\:.$$ 
\end{proposition}
\begin{proof}
By construction we have the von Neumann algebras identity
\[\left(\bigcup_{i\in L} \pi_0\left( \sweyl{\open{I}_i}\right) \right)''
 = \left(\bigvee_{i\in L} \pi_0\left( \sweyl{\open{I}_i}\right) \right)'', \]
and the latter equals
 $\pi_0\left(\bigvee_{i\in L} \sweyl{\open{I}_i}\right)''$ and thus $\pi_0\left(\overline{\bigvee_{i\in L} \sweyl{\open{I}_i}}\right)''$
by continuity of $\pi_0$.
 Since $\overline{\bigvee_{i\in L} \sweyl{\open{I}_i}} \subset \sweyl{\open{I}}$,
  proving that $\overline{\bigvee_{i\in L} \sweyl{\open{I}_i}} \supset \sweyl{\open{I}}$ is enough to conclude. Take $f \in \testfuncv{\bS^1}{\bR}$ 
   supported
 in $\open{I}$. Since $\supp{f}$ is compact, for some finite number of intervals $\open{I}_{i_1}, \ldots, \open{I}_{i_n}$,
 $\supp f \subset \cup_{j=1}^n \open{I}_{i_j}$. Then, employing a suitable smooth partition of the unit
 on the union of the intervals $\open{I}_{i_1}$,  $1= \sum_{j=1}^n g_j$ with $\supp{g_{j}} \subset \open{I}_{i_j}$, one has
  $f = \sum_{j=1}^n g_j\!\cdot\!f$
 where $  g_j\!\cdot\!f\in \testfuncv{\bS^1}{\bR}$ and $\supp{ g_j\!\cdot\!f} \subset \open{I}_{i_j}$.
 In this way, using Weyl relations, one finds that
every generator of $\sweyl{\open{I}}$ can be written, up to a factor in $\bC$, as a product of generators of each $\sweyl{\open{I}_{i_j}}$,
therefore $\overline{\bigvee_{i\in L} \sweyl{\open{I}_i}} \supset \sweyl{\open{I}}$. The case  $\cup_{i\in L} \open{I}_i = \bS^1$ has the same proof.\\
(b) The proof of (a) is valid also if $\open{I} = \bS^1$ no matter if $\bS^1 \not \in \poset{R}$. The class $\{\open{I}_r\}_{r\in \bR}$
with $\open{I}_r:= \open{I}+(r,0)$ satisfies $\cup_{r\in \bR} \open{I}_r = \bS^1$ and then the proof of (a)
encompasses the proof of the first statement in (b).
The second statement can be proved as follows.
$\alg{R}_{KG}\Psi_0$ is dense in $\varhs{H}_0$ because $\alg{R}_{KG}\supset \pi_0(\alg{W}_{KG})$ and $\pi_0(\alg{W}_{KG})\Psi_0$ is dense in $\varhs{H}_0$ by GNS theorem.
Hence, in view of the first statement, $\left(\bigvee_{r\in \bR} \pi_0\left( \sweyl{\open{I}+(r,0)}\right)\right)''\Psi_0$ is dense in $\varhs{H}_0$.
However $\bigvee_{r\in \bR} \pi_0\left( \sweyl{\open{I}+(r,0)}\right)$ is dense in $\left(\bigvee_{r\in \bR} \pi_0\left( \sweyl{\open{I}+(r,0)}\right)\right)''$
in the strong operatorial topology and thus, in turn, $\bigvee_{r\in \bR} \pi_0\left( \sweyl{\open{I}+(r,0)}\right)\Psi_0$
must be dense in $\varhs{H}_0$. 
\end{proof}

\noindent The last important property we mention is the \emph{Reeh-Schlieder property} for $\{\snet{R}{\open{I}}\}_{\open{I}\in \poset{R}}$.
\begin{proposition}[Spatial Reeh-Schlieder property]\label{th:spatial Reeh-Schlieder}
For every $\open{I}\in \poset{R}$ the vacuum vector $\Psi_0$ is cyclic for $\pi_0(\sweyl{\open{I}})$ and separating for $\snet{R}{\open{I}}$. 
\end{proposition}
\begin{proof}
See appendix \ref{app:proofs}.
\end{proof}

\subsection{Haag duality.}
In this section we establish the validity of a key property, namely (spatial) {\em Haag duality}; the proof uses an extension (not very straightforward due to technical difficulties) of the approach by Leyland-Roberts-Testard \cite{Leyland_Roberts_Testard_78}. Afterwards we'll consider (spatial) {\em local definiteness} and (spatial) {\em punctured Haag duality}. 
{\em Haag duality} in our theory  states that:
\begin{equation}\label{eq:Haag_Duality}
 \snet{R}{\open{I}}' = \alg{R}(\open{I}')\quad \mbox{for every}\quad \open{I}\in\poset{R}. 
\end{equation}
In more general contexts $\open{I}'$ does not belong to the class $\poset{R}$ labeling the local algebras and thus the right-hand side of equation (\ref{eq:Haag_Duality}) has to be defined, imposing spatial additivity, as the von Neumann algebra generated by all local algebras
$\snet{R}{\open{J}}$ with $\open{J} \in \poset{R}$ and $\open{J} \subseteq \open{I}$. However in our case there are no problems since $\open{I}' \in \poset{R}$.

To prove Haag duality we need some preliminary definitions and results. Consider a generic complex Hilbert space $\hs{H}$; as we saw earlier, for all vectors $\psi\in \hs{H}$ the unitary operator
\begin{equation}\label{eq:Weylgen} 
W[\psi] := e^{\overline{a(\psi) -  a^*(\psi)}}
\end{equation}
is well-defined on the symmetrized Fock space $\fock{\hs{H}}$ (see e.g. \cite{Bratteli_Robinson_II}). 
These operators satisfy Weyl relations with respect to the  symplectic form
\begin{equation} \label{eq:Ksigmaext} 
\sigma(\psi,\psi') := - 2\,\Im m \ip{\psi}{\psi'} \:, \quad \mbox{for $\psi,\psi' \in \hs{H}$.}
\end{equation}
In the following, if $M \subset \hs{H}$ is a real (not necessarily closed) subspace, $M'\subset \hs{H}$  denotes the closed real subspace {\bf symplectically orthogonal} to $M$:
$$M':= \left\{ \psi \in \hs{H}\:\: |\:\: \sigma(\psi,\psi')=0 \quad \forall\, \psi'\in M \right\}\:.$$
It arises that $M'= \overline{M}' =  \overline{M'}$.
If $M$ is a closed real subspace of $\hs{H}$, the von Neumann algebra  generated by all of $W[\psi]$ with $\psi \in M$  will be indicated by $\alg{R}[M]$.  The fundamental result by Leylard, Roberts and Testard, based on Tomita-Takesaki theory, is that \cite{Leyland_Roberts_Testard_78}:
\begin{equation}\label{eq:LRT}
\alg{R}[M]' = \alg{R}[M'], 
\end{equation}
\noindent for every real closed subspace $M$ of $\hs{H}$.
Another important result is that \cite{Leyland_Roberts_Testard_78}:
\begin{equation}\label{eq:LRT2}
\alg{R}[M] \cap \alg{R}[N] = \alg{R}[M \cap N], 
\end{equation}
\noindent for any pair of closed real subspaces $M,N \subset \hs{H}$.

We now specialize to the case where $\hs{H}$ is the one-particle space $\csqint{\bS^1}{\theta}$.
If $\open{I}\in\poset{R}$, henceforth $M_{\open{I}}:=\overline{K(\Cdata{\open{I}})}$. 
Notice that $M_{\open{I}} \subset M_{\open{J}}$ when $\open{I}\subset \open{J}$ are elements of $\poset{R}$.
 $\alg{R}[M_{\open{I}}]$ denotes the  von Neumann algebra generated by operators $W[\psi]$  with $\psi \in M_{\open{I}}$.
The symplectic form on $\hs{H}$  defined as in eq. (\ref{eq:Ksigmaext}) is an extension of that initially defined on $\Cdata{\sdc{C}}$  because of eq. (\ref{eq:Ksigma}). On the other hand, since the $\bR$-linear map $K : \Cdata{\sdc{C}}\to \hs{H}$  is injective,  by construction it turns out that $\pi_{0}(W(\Phi,\Pi)) = W[\psi]$ if $\psi = K(\Phi,\Pi)\in K(\Cdata{\sdc{C}})$ and $W:\Cdata{\sdc{C}}\mapsto\alg{W}_{KG}$ is the Weyl map.
Since the $\bR$-linear  map $K(\Cdata{\sdc{C}}) \ni \psi \mapsto  W[\psi]$  is strongly continuous (see for instance \cite{Bratteli_Robinson_II}), we finally obtain that $\snet{R}{\open{I}} = \alg{R}[M_{\open{I}}]$.

We are now ready to state the main theorem.
\begin{theorem}[Spatial Haag duality]
 For every $\open{I}\in \poset{R}$ it holds:
$$\snet{R}{\open{I}}' = \alg{R}(\open{I}')\:.$$ 
\end{theorem}
\begin{proof}
By spatial locality one obtains immediately that $\alg{R}(\open{I}')\subset \snet{R}{\open{I}}'$. It remains to show that  $\snet{R}{\open{I}}' \subset \alg{R}(\open{I}')$. To this end we want to use eq. (\ref{eq:LRT}) when $\hs{H}_0=\csqint{\bS^1}{\theta}$ is the one-particle space for the vacuum representation $\pi_0$. In particular $\fock{\hs{H}_0}= \varhs{H}_0$. 
Therefore the inclusion $\snet{R}{\open{I}}' \subset \alg{R}(\open{I}')$ can be re-written as $\alg{R}[M_\open{I}]' \subset \alg{R}[M_\open{I}']$. In view of eq. (\ref{eq:LRT}), to prove Haag duality is enough to establish that 
\begin{equation} \label{eq:fineH}
(M_\open{I})' \subset M_{\open{I}'}\:, \quad
\mbox{for every $\open{I}\in \poset{R}$.} 
\end{equation}
This is true in view of the subsequent lemma. 
\end{proof}
 
\begin{lemma}\label{lem:central_lemma}
If $\open{I}\in \poset{R}$ it holds: 
\begin{equation} \label{eq:centralnew}
(M_{\open{I}})' = M_{\open{I}'}\:, \quad\mbox{for every $\open{I}\in \poset{R}$.} 
\end{equation}
\end{lemma}

\begin{proof}
Since $\open{I}$ and $\open{I}'$ are disjoint, $\sigma((\Phi',\Pi'), (\Phi,\Pi))=0$ if $(\Phi',\Pi') \in \Cdata{\open{I}'}$ and $(\Phi,\Pi) \in \Cdata{\open{I}}$; taking the closures of the spaces $K(\Cdata{\open{I}})$ and $K(\Cdata{\open{I}'})$ it must hold  $(M_{\open{I}})' \supset M_{\open{I}'}$. Therefore to establish the validity of (\ref{eq:centralnew})  is enough to achieve (\ref{eq:fineH}). So, given $\psi \in (M_{\open{I}})'$, we want to show that  $\psi \in M_{\open{I}'}$. The proof descends from these two facts:
 
 (I) If $\psi \in (M_{\open{I}})'$, for every (sufficiently small)
 $\epsilon>0$, $\psi \in M_{\open{I}'+ (-\epsilon, \epsilon)}$.

 (II) For every $\open{J}\in \poset{R}$,
 \begin{equation}
 \bigcap_{\epsilon >0} M_{\open{J} + (-\epsilon, \epsilon)} \subset M_{\open{J}} \label{eq:due}\:.
 \end{equation}

\noindent {\em Proof of (I)}. Take $\psi \in (M_{\open{I}})'$; since $M_{\open{I}'+ (-\epsilon, \epsilon)} = \overline{K(\Cdata{\open{I}'+ (-\epsilon, \epsilon)})}$ to reach our goal it's sufficient to exhibit a sequence $\{(\Phi_k, \Pi_k)\}_{k\in \bN} \subset \Cdata{\open{I}'+ (-\epsilon, \epsilon)}$ such that $K(\Phi_k, \Pi_k) \to \psi$ as $k\to +\infty$. Looking at (\ref{eq:Ksigma}), we define the distributions $\Phi_\psi,\Pi_\psi \in \distrib{\bS^1}$ individuated by
 \begin{eqnarray}
 \int \Phi_\psi(\theta) f(\theta)\: d\theta &:=& 2\, \Im m\, \ip{\psi}{K(0,f)}\:, \quad \forall\, f \in \smoothfuncv{\bS^{1}}{\bC}\:, \nonumber\\
 \int \Pi_\psi(\theta) g(\theta)\: d\theta &:=& -2\, \Im m\, \ip{\psi}{K(g,0)} \:.\quad \forall\, g \in \smoothfuncv{\bS^{1}}{\bC}\nonumber \:.
 \end{eqnarray}
 (Applying definition \ref{eq:1-p struct} one proves straightforwardly that the linear functionals defined above are continuous in
 the sense of distributions. In the case of  $\Pi_\psi$ one has that the functional individuated -- varying $f$ -- by  
$\ip{\psi}{A^{-1/4} f}
  = \ip{A^{-1/4}\psi}{ f}$
 is trivially  continuous.
 In the case of $\Phi_\psi$ notice that $\ip{\psi}{A^{1/4} f} =
 \lim_{n\to +\infty}\ip{A^{1/4}\psi_n}{f}$ for some sequence  $\Dom{A^{1/4}} \ni  \psi_n \to \psi$ independent from $f$.
 As  each linear functional $\ip{A^{1/4} \psi_n}{\cdot}$ is a distribution, $\Phi_\psi$ is a distribution as well.

 By construction the distributions $\Phi_\psi$ and $\Pi_\psi$  have supports contained in $\overline{\open{I}'}$ because from the definitions descends that
 $\int \Phi_\psi(\theta) f(\theta)\: d\theta =0 $ and
 $\int \Pi_\psi(\theta) f(\theta)\: d\theta =0$  for $\psi \in (M_\open{I})'$,
 whenever $f$ is supported in $\open{I}$. Now consider
 $\rho \in\testfuncv{\bR}{\bR}$ supported in $(-\epsilon/2, \epsilon/2)$ and define (using weak operatorial
 topology)
 $\psi * \rho := \int_\bR \rho(r) e^{-ir P^\otimes} \psi\:dr \:.$
 Fubini-Tonelli theorem and the fact that $e^{-ir P^\otimes}$ commutes with $A^{\alpha}$ (it can be proved immediately passing to Fourier-series representation) entails that
 $\Phi_{\psi * \rho} = \Phi_\psi * \rho$ and $\Pi_{\psi * \rho} = \Pi_\psi * \rho$,
 where $*$ in the right-hand side denotes the standard convolution so that $\Phi_{\psi * \rho}$ and
  $\Pi_{\psi * \rho}$ are smooth functions supported in $\open{I}'+(-\epsilon,\epsilon)$ and thus
  $\psi*\rho \in M_{\open{I}'+
  (-\epsilon,\epsilon)}$. Therefore, assuming the existence of a suitable sequence $\{\rho_k\}$ of real smooth functions supported in
  $(-\epsilon/2,
  \epsilon/2)$, with  $\psi * \rho_k \to \psi$,  the sequence of pairs
  $(\Phi_k, \Pi_k) := (\Phi_{\psi * \rho_k},\Phi_{\psi * \rho_k})$
  turns out to be  made of real smooth functions supported in
  $\open{I}'+ (-\epsilon, \epsilon)$, 
and $K(\Phi_k, \Pi_k) \to \psi$ as requested, proving that $\psi \in M_{\open{I}'+ (-\epsilon, \epsilon)}$.\\
It's only left to prove that such a sequence $\{\rho_k\}$ does exist.
  Consider a family of smooth functions $\rho_k \geq 0$ with $\supp{\rho_k}  \subset [-1/k,1/k]$ and with $\int_\bR \rho_k(r)
  dr =1$. In our
  hypotheses $\norm{\psi * \rho_k - \psi} = \norm{\int \rho_k(r) e^{-ir P^\otimes} \psi dr  - \psi}$ can be re-written as
\begin{equation}
\begin{split}
\norm{\int_{-1/k}^{1/k}  \left( \rho_k(r) e^{-ir P^\otimes}  - \rho_k(r)\right)  \psi\: dr} &\leq \int_{-1/k}^{1/k} \rho_k(r) \norm{\left(  e^{-ir P^\otimes}  - I \right)  \psi}\: dr \\
  &\leq\sup_{r\in [-1/k,1/k] }\norm{\left(  e^{-ir P^\otimes}  - I \right)  \psi}\\ 
\end{split}
\end{equation}
\noindent and the last term vanishes as $k\to +\infty$ because $r\mapsto e^{-ir P^\otimes}$ is strongly continuous. Q.E.D.

\noindent {\em Proof of  II}.
  We have to establish the validity of (\ref{eq:due}). The proof is based on the following technical proposition.
  (This is a quite difficult point whose proof technically differentiates QFT on $\bS^1$ from
  that in Minkowski space as done in  \cite{Leyland_Roberts_Testard_78}.) 
\end{proof}
 
\begin{proposition} \label{prop:last}
Take $\open{J}\in \poset{R}$ and assume
 $\open{J} \equiv (-a,a) \subset (-\pi,\pi] \equiv \bS^1$ with a suitable choice of the origin of $\bS^{1}$. There is a class of operators   $D_\lambda : \csqint{\bS^1}{\theta}\to \csqint{\bS^1}{\theta}$, with $\lambda$ ranging in a neighborhood $\open{O}$ of $1$,  such that, if $\psi \in M_L$ with $\poset{R} \ni L\subsetneq \open{J}$ :
 
  {\bf (a)} $D_\lambda \psi \in M_{\lambda L}$ and
 
  {\bf (b)} $D_\lambda \psi \to \psi$ as $\lambda \to 1$   
 \end{proposition}
 \begin{proof}
  See  appendix \ref{app:proofs}.  
 \end{proof}

Notice that the requirement $\open{J} \equiv (-a,a)$ does not imply any true restriction since all the theory is invariant under rotations of the circle.  To go on with the main proof, one sees by direct inspection that, for  $\lambda \in (0,1)$,  there is $\epsilon_\lambda>0$ with
   $\lambda \left( \open{J} + (-\epsilon_\lambda, \epsilon_\lambda)\right) \subset \open{J}$ \textbf{(c)}.

  If $\psi \in  \bigcap_{\epsilon >0} M_{\open{J} + (-\epsilon, \epsilon)}$ then  $\psi \in M_{ \open{J} + (-\epsilon_\lambda, \epsilon_\lambda)}$
   for every $\lambda \in (0,1)$, so that  from (a)  $D_\lambda \psi \in M_{ \lambda(\open{J} + (-\epsilon_\lambda, \epsilon_\lambda))}$. Therefore, by (c),  $D_\lambda \psi \in M_{\open{J}}$. Finally, taking the limit as $\lambda \to 1^-$ and using (b) and the fact that $M_{\open{J}}$ is closed, one achieves
  $\psi \in M_{\open{J}}$.  $\Box$\\
  
\noindent\textbf{Remark.} Since, by construction  $\bigcap_{\epsilon >0} M_{\open{J} + (-\epsilon, \epsilon)} \supset M_{\open{J}}$, validity of statement
 (II) is in fact equivalent to
 the {\em outer regularity} property:
\begin{equation}\label{eq:outer-regolarity}\
 \bigcap_{\epsilon >0} M_{\open{J} + (-\epsilon, \epsilon)} = M_{\open{J}}.
\end{equation}
    
\subsection{Local definiteness, factoriality and punctured Haag duality.}
We pass to prove some important properties of the class $\{\snet{R}{\open{I}}\}_{\open{I}\in\poset{R}}$ and the analogous class of von Neumann algebras $\{\alg{R}_\lambda(\open{I})\}_{\open{I} \in \poset{R}}$ associated with any pure state $\lambda$  which is {\bf locally unitarily equivalent} to $\omega_0$, i.e.
for every $\open{I} \in \poset{R}$ there is a unitary  operator $U_{\open{I}} : \varhs{H}_\lambda \to \varhs{H}_0$ with:
\begin{equation}\label{eq:locun}
U_{\open{I}}\, \pi_\lambda (a) U^{-1}_{\open{I}} = \pi(a) \:, \quad \mbox{for all}\, a\in \sweyl{\open{I}}. 
\end{equation}
\noindent where as usual the GNS triple for $\lambda$ is denoted by $(\varhs{H}_{\lambda}, \pi_{\lambda}, \Psi_{\lambda})$.\\
\noindent\textbf{Remark.} Notice that $\alg{R}_\lambda(\open{I}):= \pi_\lambda(\sweyl{\open{I}})''$ and,  exploiting  strong operator topology and bijectivity of $U_{\open{I}}$, (\ref{eq:locun})
implies $U_{\open{I}}\alg{R}_\lambda(\open{I}_1)U_{\open{I}}^{-1} = \alg{R}(\open{I}_1)$ and $U_{\open{I}}\alg{R}_\lambda(\open{I}_1)'U_{\open{I}}^{-1} = \alg{R}(\open{I}_1)'$, for every $\open{I}_1 \in \poset{R}$ with $\open{I}_1 \subset \open{I}$ (including the case $\open{I}_1=\open{I}$). \\

\noindent First we shall be concerned with  {\em spatial local definiteness}: it simply means that
  the algebra of observables associated with a single point $p\in \bS^1$ is the trivial one $\{c\!\cdot\! 1\}_{c\in \bC}$, $1$ being the unit element of $\alg{R}_{KG}$. Since $\{p\} \not \in \poset{R}$, the algebra associated with $\{p\}$ is obtained by taking the intersection of algebras $\alg{R}_\lambda(\open{I})$ for all $\open{I}\in \poset{R}$ with $\open{I} \ni p$.

Secondly we shall examine the validity of {\em spatial punctured  Haag duality}, i.e. spatial Haag duality valid on the restricted space $\bS^1 \setminus \{p\}$ for every fixed $p \in \bS^1$.
%TODO: add general definition in the poset chapter ?
In details, fix $p\in \bS^1$ and, for every $\open{I}\in \poset{R}$ with $p \not\in \cl{\open{I}}$, define $\open{I}'_1, \open{I}'_2 \in \poset{R}$ as the disjoint, not containing $p$, sets such that $\open{I}' = \open{I}'_1 \cup \open{I}'_2 \cup \{p\}$.
With this definition  validity of  spatial punctured  Haag duality means that for every $\open{I}\in\poset{R}$ it holds:
$$\alg{R}_\lambda(\open{I})' = \left(\pi_\lambda(\sweyl{\open{I}^{\prime}_1}) \cup \pi_\lambda(\sweyl{\open{I}^{\prime}_2})\right)''\:.$$
Finally we shall focus on {\em factoriality}. Validity of factoriality for $\{\alg{R}_\lambda(I)\}_{\open{I} \in \poset{R}}$ means that each $\alg{R}_\lambda(\open{I})$ is a factor, for every $\open{I} \in\poset{R}$.
The proofs are based on the following important fact.
\begin{lemma}\label{lem:Tomita}
If $I,\open{J} \in\poset{R}$ and $\open{I} \cap \open{J} = \emptyset$ then $M_{\open{I}} \cap M_{\open{J}} = \{0\}$ and thus it also holds $\snet{R}{\open{I}} \cap \snet{R}{\open{J}} = \bC\cdot 1$. 
\end{lemma}
\begin{proof}
See  appendix \ref{app:proofs}.
\end{proof}
\begin{theorem}\label{th:genprop}
Let $\lambda : \alg{W}_{KG} \to \bC$ be a pure state, with GNS triple $(\varhs{H}_\lambda, \pi_\lambda, \Psi_\lambda)$, locally unitarily equivalent to $\omega_0$ (thus $\lambda = \omega_0$ in particular).
The spatially additive, isotonous, local class of von Neumann algebras $\{\alg{R}_\lambda(\open{I})\}_{\open{I}\in \poset{R}}$ with $\alg{R}_\lambda(\open{I}):= \pi_\lambda(\sweyl{\open{I}})''$
fulfills the following properties.\\
(i)   {\bf Spatial local definiteness}.\\
(ii)  {\bf Spatial Haag duality}.\\
(iii) {\bf Spatial punctured Haag duality}.\\
(iv)  {\bf Factoriality}.\\
(v)   {\bf Borchers property}.\\
Finally,  if $\open{I},\open{J} \in \poset{R}$ one has
\begin{equation} \label{eq:speriamo2}
\snet{R}{\open{I}} \cap \snet{R}{\open{J}} = \snet{R}{\open{I} \cap \open{J}},
\end{equation}
where $\snet{R}{\open{I} \cap \open{J}}:= \bC\cdot 1$ if $\open{I}\cap \open{J} = \emptyset$, or $\snet{R}{\open{I} \cap \open{J}}:= (\snet{R}{\open{K}_1} \cup \snet{R}{\open{K}_2}) ''$, if $\open{I} \cap \open{J}$ is disconnected, $\open{K}_1,\open{K}_2 \in \poset{R}$ being the two components of $\open{I}\cap \open{J}$. 
\end{theorem}
\begin{proof}
First of all we consider the case of $\lambda = \omega_0$. We have to demonstrate the validity of (i), (iii) and (iv).\\
\textbf{(i)} We have to show that $\bigcap_{\open{J} \in \poset{R}, \open{J}\ni p} \snet{R}{\open{J}}= \bC\cdot 1$.
 It is clear that $\bC\cdot 1 \subset \bigcap_{\open{J} \in \poset{R}, \open{J}\ni p} \snet{R}{\open{J}}$, so we have to prove the converse.
Fix $\bS^1\ni q\neq p$ and consider the disjoint sets $\open{I}_1,\open{I}_2 \in \poset{R}$ uniquely determined by
assuming that $\partial \open{I}_1= \partial \open{I}_2 = \{p,q\}$.
If $A \in \bigcap_{\open{J} \in \poset{R}, \open{J}\ni p} \snet{R}{\open{J}}$ it must in particular commute with $\pi_0(W(\Phi,\Pi))$
for every choice of $(\Phi,\Pi)$ whose supports are both contained either in $\open{I}_1$ or $\open{I}_2$.  Therefore
$A\in \pi_0(\sweyl{\open{I}_1})' = \pi_0(\sweyl{\open{I}_1})''' =  \alg{R}(\open{I}_1)'$ and
$ A\in \pi_0(\sweyl{\open{I}_2})' = \pi_0(\sweyl{\open{I}_2})''' =  \alg{R}(\open{I}_2)'$. We conclude that $A \in \alg{R}[M_{\open{I}_1}]' \cap \alg{R}[M_{\open{I}_2}]'$.
That is $A \in \alg{R}[M_{\open{I}_2}' \cap M_{\open{I}_1}']$, where we employed (\ref{eq:LRT}) and (\ref{eq:LRT2}). Now, using lemma
\ref{lem:central_lemma}, we may assert that $A \in \alg{R}[M_{\open{I}_1} \cap M_{\open{I}_2}]$ using the fact that $M_{\open{I}_2}' = M_{\open{I}'_2} = M_{\open{I}_1}$
and $M_{\open{I}_1}' = M_{\open{I}'_1} = M_{\open{I}_2}$. By lemma \ref{lem:Tomita} we conclude that $A \in \alg{R}[\{0\}] = \bC\cdot 1$.\\
\noindent \textbf{(iv)} We have to establish that every $\snet{R}{\open{I}}$ is a factor, i.e.
 $\snet{R}{\open{I}}' \cap \snet{R}{\open{I}} = \bC\cdot 1$.
 The proof of (i) shows, in particular, that each $A \in \alg{R}(\open{I}_1)' \cap \alg{R}(\open{I}_2)'$ must be of the form $c\!\cdot\!1$ if
 $\open{I}_1$ and $\open{I}_2$ are disjoint and fulfill $\partial\open{I}_1= \partial\open{I}_2 = \{p,q\}$ for arbitrarily fixed $p,q \in \bS^1$.
This applies to the case
 $\open{I}_1 := \open{I} \in \poset{R}$, $\open{I}_2 = \open{I}'$. Therefore we know that $\snet{R}{\open{I}}' \cap \alg{R}(\open{I}')' = \bC\cdot 1$. However, by
 Haag duality it can be re-written  $\snet{R}{\open{I}}' \cap \snet{R}{\open{I}} = \bC\cdot 1$, that is $\snet{R}{\open{I}}$ is a factor.\\
\textbf{(iii)} For every $p\in \bS^1$ and $\open{I} \in \poset{R}$, we have to prove that
 $$\alg{R}(\open{I})' = \left(\pi_0(\sweyl{\open{I}^{\prime}_1}) \cup \pi_0(\sweyl{\open{I}^{\prime}_2})\right)'',$$
$\open{I}'_1, \open{I}'_2 \in \poset{R}$ being the disjoint, not containing $p$,
    sets  such that $\open{I}' = \open{I}'_1 \cup \open{I}'_2 \cup \{p\}$.
By spatial additivity $\alg{R}(\open{I}') \subset (\pi_0(\sweyl{\open{I}^{\prime}_1}) \cup \pi_0(\sweyl{\open{I}^{\prime}_2}) \cup \pi_0(\sweyl{\open{J}}))''$
where $\open{J}\in \poset{R}$ is any open set with $\open{J} \ni p$. Since $\open{J}$ is arbitrary, we have:
\begin{equation}
\begin{split}
\alg{R}(\open{I}') &\subset \bigcap_{\open{J}\in \poset{R}, \open{J}\ni p} \left(\pi_0(\sweyl{\open{I}^{\prime}_1}) \cup \pi_0(\sweyl{\open{I}^{\prime}_2}) \cup \pi_0(\sweyl{\open{J}})\right)''=\\ 
&=\left(\pi_0(\sweyl{\open{I}^{\prime}_1}) \cup \pi_0(\sweyl{\open{I}^{\prime}_2}) \cup  \bigcap_{\open{J}\in \poset{R}, \open{J}\ni p}  \pi_0(\sweyl{\open{J}})\right)''=\\
&=\left(\pi_0(\sweyl{\open{I}^{\prime}_1}) \cup \pi_0(\sweyl{\open{I}^{\prime}_2}) \right)''\\
\end{split}
\end{equation}
where we have used local definiteness in the last step:

$$\bigcap_{\open{J}\in \poset{R}, \open{J}\ni p}  \pi_0(\sweyl{\open{J}} \subset \bigcap_{\open{J}\in \poset{R}, \open{J}\ni p}  \snet{R}{\open{J}} = \bC\cdot 1.$$
We have found:
$\alg{R}(\open{I}') \subset \left(\pi_0(\sweyl{\open{I}^{\prime}_1}) \cup \pi_0(\sweyl{\open{I}^{\prime}_2}) \right)''$.
The other inclusion is trivially true because $\open{I}'_1 \cap \open{I}= \open{I}'_2 \cap \open{I} = \emptyset$ and so
$$\pi_0(\sweyl{\open{I}^{\prime}_1}) \cup \pi_0(\sweyl{\open{I}^{\prime}_2}) \subset \snet{R}{\open{I}}' = \alg{R}(\open{I}')$$ 
by Haag duality. Summarizing: $\alg{R}(\open{I}')= \left(\pi_0(\sweyl{\open{I}^{\prime}_1}) \cup \pi_0(\sweyl{\open{I}^{\prime}_2}) \right)''$.
Using Haag duality once again, we finally get:
$$\snet{R}{\open{I}}'= \alg{R}(\open{I}')= \left(\pi_0(\sweyl{\open{I}^{\prime}_1}) \cup \pi_0(\sweyl{\open{I}^{\prime}_2}) \right)''.$$

To conclude the case $\lambda=\omega_0$, let us prove (\ref{eq:speriamo2}). First of all notice that, if $\open{I} \in \poset{R}$
and $\open{I}_1,\open{I}_2,\ldots,\open{I}_n \in \poset{R}$ are pairwise disjoint subsets of $\open{I}$ with $\Int{\left(\cup_{i=1}^n \overline{\open{I}_i}\right)} = \open{I}$ then
\begin{equation}\label{eq:speriamo}
 \left(\bigcup_{i=1}^n \alg{R}(\open{I}_i)\right)'' = \snet{R}{\open{I}}\:.
\end{equation}
The proof for $n=2$ is a straightforward consequence of Haag duality and punctured Haag duality used together.
Iterating the procedure, replacing $\open{I}$ with $\open{I}_1$ and $\open{I}_2$ one achieves the general case with $n$ arbitrary but finite.
Using (\ref{eq:speriamo}) one gets (\ref{eq:speriamo2}); let us prove it. To this end decompose $\open{I}$ into pairwise disjoint sets as  $\open{I}= (\open{I}\setminus \overline{\open{J}})
\cup (\open{I} \cap \open{J}) \cup (\open{I} \cap \partial \open{J})$ where $\open{I}\setminus \overline{\open{J}},\open{I} \cap \open{J} \in \poset{R}$ if these sets are nonempty and connected,
whereas  $(\open{I} \cap \partial \open{J})$ is empty or it contains up to two points. (We assume that $\open{I}\setminus \overline{\open{J}}$, $\open{J}\setminus
\overline{\open{I}}$ and $\open{I}\cap \open{J}$ are connected, the case where someone of those sets, say $\open{A}$,
 has two components $A_1, A_2$ (which must belong to $\poset{R}$ by construction) is a straightforward generalization replacing in the following $\alg{R}(\open{A})$ with $\alg{R}(\open{A}_1) \cup \alg{R}(\open{A}_2)$. Finally, when  $\open{A}=\emptyset$ one has to replace
 $\alg{R}(\open{A})$ with $\bC\cdot 1$. 
Using (\ref{eq:speriamo}) and following the same procedure for $\open{J}$ one achieves the relations
$\snet{R}{\open{I}}  = \left(\alg{R}\left(\open{I}\setminus \overline{\open{J}}\right) \cup \alg{R}\left(\open{I} \cap \open{J}\right)\right)''\quad
\mbox{and}\quad \snet{R}{\open{J}}  = \left(\alg{R}\left(\open{J}\setminus \overline{\open{I}}\right) \cup \alg{R}\left(\open{I} \cap \open{J}\right)\right)''.$
 Putting all together we have
$$\snet{R}{\open{I}} \cap \snet{R}{\open{J}} = \left(\alg{R}\left(\open{I}\setminus \overline{\open{J}}\right) \cup \alg{R}\left(\open{I} \cap \open{J}\right)\right)''\cap
 \left(\alg{R}\left(\open{J}\setminus \overline{\open{I}}\right) \cup \alg{R}\left(\open{I} \cap \open{J}\right)\right)''$$ 
that is
\begin{multline}
\snet{R}{\open{I}} \cap \snet{R}{\open{J}} = \left(\left(\alg{R}_{\open{I}\setminus\bar{\open{J}}} \cup \alg{R}_{\open{I}\cap \open{J}}\right)\cap \left(\alg{R}_{\open{J}\setminus\bar{\open{I}}} \cup \alg{R}_{\open{I}\cap \open{J}}\right)\right)''= \\
=\left(\left(\alg{R}_{\open{I}\setminus \bar{\open{J}}} \cap\alg{R}_{\open{J}\setminus\bar{\open{I}}}\right)\cup
\left(\alg{R}_{\open{I}\setminus \bar{\open{J}}} \cap \alg{R}_{\open{I}\cap \open{J}} \right)
\cup\left(\alg{R}_{\open{J}\setminus\bar{\open{I}}} \cap  \alg{R}_{\open{I}\cap \open{J}}\right) \cup \alg{R}_{\open{I}\cap \open{J}}\right)'',
\end{multline} 
where we have put $\alg{R}_{I\cap J}:=\snet{R}{I\cap J}$, $\alg{R}_{I\setminus\bar{J}}:=\snet{R}{I\setminus\bar{J}}$ and so on in order to shorten the formulas.
 However $(\open{I}\setminus \overline{\open{J}}) \cap (\open{J}\setminus \overline{\open{I}}) = (\open{I}\setminus \overline{\open{J}}) \cap (\open{I} \cap \open{J}) = (\open{J}\setminus \overline{\open{I}}) \cap (\open{I} \cap \open{J}) = \emptyset$ so that, applying
 lemma \ref{lem:Tomita} all factors are trivial barring $\alg{R}\left(\open{I} \cap \open{J}\right)$. The final result is
 $\snet{R}{\open{I}} \cap \snet{R}{\open{J}}= \alg{R}\left(\open{I} \cap \open{J}\right)'' = \alg{R}\left(\open{I} \cap \open{J}\right)$.\\

\noindent Let us pass to the case where $\lambda \neq \omega_0$ but $\lambda$ is locally unitarily equivalent to $\omega_0$.
Spatial additivity, locality and isotony for $\{\pi_\lambda(\sweyl{\open{I}})''\}_{\open{I}\in \poset{R}}$ have the same proofs
  as for  $\{\pi_0(\sweyl{\open{I}})''\}_{\open{I}\in \poset{R}}$.\\
\textbf{(i)} Spatial local definiteness is an immediate consequence of the definition of local unitary equivalence
and closedness of von Neumann algebras with respect to  strong operatorial topology. \\
\textbf{(ii)} Let us prove Haag duality for the state $\lambda$. Fix $\open{I}\in \poset{R}$ and consider $\open{J}_1, \open{J}_2 \in \poset{R}$ such
that $\open{J}_1 \cup \open{J}_2 = \bS^1$ and $\open{I} \subset \open{J}_1 \cap \open{J}_2$. As
$\open{I}\subset \open{J}_1$,
we have (see the remark after (\ref{eq:locun}))
$ U_{\open{J}_1} \alg{R}_\lambda(\open{I})' \cap \alg{R}_\lambda(\open{J}_1) U_{\open{J}_1}^{-1} = 
\snet{R}{\open{I}}' \cap \alg{R}(\open{J}_1) =  \alg{R}(\open{I}') \cap \alg{R}(\open{J}_1)$;
but we know that $ \alg{R}(\open{I}') \cap \alg{R}(\open{J}_1) = \alg{R}(\open{I}'\cap \open{J}_1)$ due to (\ref{eq:speriamo2}),
so that
$$\alg{R}_\lambda(\open{I})' \cap \alg{R}_\lambda(\open{J}_1) = \alg{R}_\lambda(\open{I}' \cap \open{J}_1)$$
and similarly
$$\alg{R}_\lambda(\open{I})' \cap \alg{R}_\lambda(\open{J}_2) = \alg{R}_\lambda(\open{I}' \cap \open{J}_2).$$
As a consequence:
$$\alg{R}_\lambda(\open{I})' \cap \left( \alg{R}_\lambda(\open{J}_1) \cup \alg{R}_\lambda(\open{J}_2) \right) =
\alg{R}_\lambda(\open{I}' \cap \open{J}_1) \cup \alg{R}_\lambda(\open{I}' \cap \open{J}_2)\:.$$
Taking the double commutant of both members we achieve:
$$\alg{R}_\lambda(\open{I})' \cap \left( \alg{R}_\lambda(\open{J}_1) \cup \alg{R}_\lambda(\open{J}_2) \right)'' =
\left(\alg{R}_\lambda(\open{I}' \cap \open{J}_1) \cup \alg{R}_\lambda(\open{I}' \cap \open{J}_2)\right)''\:.$$
Since $\open{J}_1 \cup \open{J}_2 = \bS^1$ and $(\open{I}' \cap \open{J}_1) \cup (\open{I}' \cap \open{J}_2) = \open{I}'$, additivity implies that the identity written above can be re-written as:
$\alg{R}_\lambda(\open{I})' \cap \alg{R}_\lambda(\bS^1) = \alg{R}_\lambda(\open{I}')$.
Finally $\alg{R}_\lambda(\bS^1) = \bop{\varhs{H}_\lambda}$ because $\lambda$ is pure and thus its GNS representation is irreducible.
We have obtained that $\alg{R}_\lambda(\open{I})' = \alg{R}_\lambda(\open{I}')$ as wanted. \\
\textbf{(iii}) The proof of punctured Haag duality is now exactly the same
as given for $\pi_0$ using spatial additivity, local definiteness, Haag duality and locality. \\
\textbf{(iv)} Taken $\open{I} \in \poset{R}$, factoriality means that
$\alg{R}_\lambda(\open{I}) \cap \alg{R}_\lambda(\open{I})' = \bC\cdot 1$.
In our hypotheses $U_I \alg{R}_\lambda(\open{I}) U^{-1}_{\open{I}} = \snet{R}{\open{I}}$
and $U_{\open{I}} \alg{R}_\lambda(\open{I})' U^{-1}_{\open{I}} = \snet{R}{\open{I}}'$ hold true for some unitary operator $U_{\open{I}}$. Now, making use of
factoriality for the vacuum representation,
$U_{\open{I}}   \alg{R}_\lambda(\open{I}) \cap  \alg{R}_\lambda(\open{I})' U^{-1}_{\open{I}} =  \snet{R}{\open{I}} \cap  \snet{R}{\open{I}}' = \bC\cdot 1$,
so that $\alg{R}_\lambda(\open{I}) \cap \alg{R}_\lambda(\open{I})' = \bC\cdot 1$.\\
\noindent\textbf{(v)} To conclude let us prove the Borchers property. We start from the vacuum case; referring to appendix \ref{app:misc}, we need to show that, given $\open{I}\in\poset{R}$, there exists $\poset{R}\ni \open{I}_{1}\subsetneq\open{I}$ such that every non-zero projection $E\in\snet{R}{\open{I}_{1}}$ is equivalent to $1$ in $\snet{R}{\open{I}}$. So, choose an interval $\open{I}\in\poset{R}$, take $\open{I}_{1}\in\poset{R}$ to be an arbitrary subinterval of $\open{I}$, and fix a projection $E \neq 0$ in $\snet{R}{\open{I}_{1}}$. If we can prove that the von Neumann algebras inclusion $\snet{R}{\open{I}_{1}}\subset\snet{R}{\open{I}}$ satisfies standard split inclusion \ref{def:vNa_inclusion}, proposition \ref{prop:standard_split} implies property $B$ so $E\sim 1$ in $\snet{R}{\open{I}}$ and we have done. In turn, standard split inclusion holds true if the following two conditions are satisfied:
\begin{itemize}
\item the inclusion $\snet{R}{\open{I}_{1}}\subset\snet{R}{\open{I}}$ is split;
\item there exists a vector $\Omega\in\varhs{H}_0$ cyclic and separating for $\snet{R}{\open{I}_{1}}$, $\snet{R}{\open{I}}$ and $\snet{R}{\open{I}_{1}}^{'}\cap\snet{R}{\open{I}}$.
\end{itemize}
The former condition holds since the generator of rotations on $\bS^1$ satisfies the trace class condition that implies the split property for the observable net (see \cite[Th. 3.2] {D'Antoni_Longo_Radulescu_98}. 

The latter follows at once from Reeh-Schlieder property for the vacuum (theorem \ref{th:spatial Reeh-Schlieder}); in fact the vacuum vector $\Psi_{0}$ is obviously cyclic and separating for $\snet{R}{\open{I}_{1}}$ and $\snet{R}{\open{I}}$; concerning $\snet{R}{\open{I}_{1}}^{'}\cap\snet{R}{\open{I}_{1}}$, we have the following chain of inclusions:

 \[\snet{R}{\open{I}}\supset{\snet{R}{\open{I}_{1}}^{'}}\cap\snet{R}{\open{I}}=\snet{R}{\open{I}_{1}^{'}\cap\open{I}}\supseteq\snet{R}{\open{J}}\]
for some $\open{J}\in\poset{R}$, taking into account Haag duality. This ends the proof for the vacuum case.

The general case $\lambda\neq\omega_0$ follows straightforwardly, from local unitary equivalence of local algebras. 
\end{proof}

\subsection{DHR sectors}

 To conclude we prove that, referring to the vacuum representation $\pi_0$, there are no {\em DHR sectors} for the massive Klein-Gordon theory on $\bS^1$ we are discussing. DHR sectors were introduced \cite{DHR_69_I, DHR_69_II, DHR_71, DHR_74} to describe localized charges; a representation belongs to a DHR sector if it coincides with the reference (vacuum) representation everywhere but on bounded regions of the spacetime. To be more precise, we shall say that an irreducible representation $\widetilde{\pi}$ of $\alg{W}_{KG}$ valued on $\bop{\varhs{H}_{0}}$ and
satisfying (spatial) Haag duality fulfills the \emph{DHR selection criterion} if, for every $\open{I}\in \poset{R}$, there is a unitary operator $U_{\open{I}'} : \varhs{H}_{0} \to \varhs{H}_{0}$ with  $U_{\open{I}'} \widetilde{\pi}(\sweyl{\open{I}'}) U_{\open{I}'}^* = {\pi_0}(\sweyl{\open{I}'})$. Notice that in our case, as $\open{I}'\in \poset{R}$ and the map  $\poset{R} \ni \open{I} \mapsto \open{I}' \in \poset{R}$ is bijective, the mentioned requirement is equivalent to say that $\widetilde{\pi}$ and $\pi_0$ are locally unitarily equivalent. A \emph{DHR sector} is a {\em global} unitary equivalence class of irreducible representations $\widetilde{\pi}$ verifying the DHR selection criterion.  
\begin{theorem}
Every representation $\widetilde{\pi}$ of $\alg{W}_{KG}$ valued on $\bop{\varhs{H}_0}$, locally unitarily equivalent to the vacuum representation $\pi_0$ and satisfying spatial Haag duality, is (globally) unitarily equivalent to $\pi_0$ itself. As a consequence there is only one DHR sector, i.e. that containing the vacuum representation.
\end{theorem}

\begin{proof}
Define $\widetilde{\alg{R}}(\open{I}):= \widetilde{\pi}(\sweyl{\open{I}})''$. We first prove that the unitaries $U_{\open{I}}$ which verify $U_{\open{I}}\,\widetilde{\pi}(\sweyl{\open{I}}) U^*_{\open{I}} = {\pi_0}(\sweyl{\open{I}})$ in view of local unitary equivalence, also satisfy
\begin{equation}\label{eq:II'}
 U_{\open{I}'}= e^{\imath\gamma_{\open{I}}}U_{\open{I}}\:, \quad \mbox{for every $\open{I} \in \poset{R}$ and some $\gamma_{\open{I}}\in\bR$.} 
\end{equation} 
For a fixed $\open{I}\in \poset{R}$, one has $U_{\open{I}} \widetilde{\alg{R}}(\open{I}) U_{\open{I}}^* = \snet{R}{\open{I}}$ and thus  
$U_{\open{I}} \widetilde{\alg{R}}(\open{I})'U_{\open{I}}^* = \snet{R}{\open{I}}'$, that is  
 (i) $U_{\open{I}}\widetilde{\alg{R}}(\open{I}')U_{\open{I}}^*  = \alg{R}(\open{I}') $ by Haag duality.
  On the other hand (ii) $U_{\open{I}'}\widetilde{\alg{R}}(\open{I}') U_{\open{I}'}^* = \alg{R}(\open{I}')$ is true  by definition.
Relations (i) and (ii) imply
 $U_{\open{I}'} U^*_{\open{I}} \in  \alg{R}(\open{I}')' = \snet{R}{\open{I}}$. Interchanging the role of $\open{I}$ and $\open{I}'$ one finds in the same way that $U_{\open{I}} U^*_{\open{I}'} \in  \alg{R}(\open{I}') = \snet{R}{\open{I}}'$, which produces $U_{\open{I}'} U^*_{\open{I}} \in  \snet{R}{\open{I}}'$ by hermitean conjugation.
 Since $U_{\open{I}'} U^*_{\open{I}} \in \snet{R}{\open{I}}\cap \snet{R}{\open{I}}'$ and factoriality holds, 
  it has to be $U_{\open{I}'} U^*_{\open{I}} = c\!\cdot\!$ for some  $c\in \bC$. As $U_{\open{I}'} U^*_{\open{I}}$ is unitary, it arises $|c|=1$ and (\ref{eq:II'}) follows
  immediately.\\ 
To go on fix $\open{I} \in \poset{R}$: we shall prove that there is a character $\chi : \Cdata{\sdc{C}}\to U(1)$ with 
 \begin{equation}\label{eq:UI}
 U_{\open{I}} \widetilde{\pi}(W(\Phi, \Pi)) U^*_{\open{I}} = \chi(\Phi,\Pi)\pi_0(W(\Phi, \Pi)) =: \pi_\chi(W(\Phi, \Pi)), 
\end{equation}
for every $(\Phi,\Pi) \in \Cdata{\sdc{C}}$.
 Thus $\widetilde{\pi}$ is unitarily equivalent to a character representation $\pi_\chi$. As $\widetilde{\pi}$
 is locally unitarily equivalent to $\pi_0$, $\pi_\chi$ is such. Since $\pi_\chi$ is irreducible it is the GNS representation of 
 a pure state $\omega_\chi$ which can be fixed as any unitary vector of $\varhs{H}_0$. Local unitary equivalence, 
 by (d) in theorem \ref{prop:DMP}, entails that $\pi_\chi$ is globally unitarily equivalent
  to $\pi_0$ and thus $\widetilde{\pi}$ is such proving the thesis of the theorem.\\
 Let us demonstrate (\ref{eq:UI}) to conclude. Consider a partition of the unity 
  $f_{\open{I}},f_{\open{I}'},f_\open{F}$ subordinated to the covering  $\open{I},\open{I}',\open{F}$ of $\bS^1$,
   where $\open{F}$ is the disconnected union of two neighborhoods $\open{J}_1,\open{J}_2 \in \poset{R}$ 
   of  $p_1,p_2 \in \partial \open{I}$ respectively. Decomposing $(\Phi,\Pi)$ making use of this partition of the unity,
    taking advantage
   of (\ref{eq:II'}) and employing Weyl relations one realizes that there is $\beta \in \bC$ (depending on $(\Phi,\Pi)$ and the chosen partition of the
   unity) with:
\begin{multline}
U_{\open{I}}\, \widetilde{\pi}\left(W(-\Phi,-\Pi) \right) U^*_{\open{I}}\, \pi_0\left(W(\Phi,\Pi) \right) =\\ 
=e^{\imath\beta}
   U_{\open{I}}\,\widetilde{\pi}\left(W(-f_\open{F}\Phi,-f_\open{F}\Pi) \right) U^*_{\open{I}}\, \pi_0\left(W(f_\open{F}\Phi,f_\open{F}\Pi) \right) 
\end{multline}
 Notice that the left-hand side does not depend on $\open{F}$. If $(\Phi_1,\Pi_1)$ has compact support contained either in $\open{I}$ or $\open{I}'$,
 taking $\open{F}$ sufficiently shrinked about $p_1$ and $p_2$ one achieves by direct inspection that the right hand side must commute with 
 $\pi_0\left(W(\Phi_1,\Pi_1) \right)$ (which coincides with  $U_{\open{I}}\,\widetilde{\pi}\left(W(\Phi_1,\Pi_1) \right) U^*_{\open{I}}$ up to a phase).
 Therefore $U_{\open{I}}\,\widetilde{\pi}\left(W(-\Phi,-\Pi) \right) U^*_{\open{I}}\,\pi_0\left(W(\Phi,\Pi) \right) \in \snet{R}{\open{I}}' \cap \alg{R}(\open{I}')' = 
 \bC\cdot 1$, by Haag duality and factoriality. Since $U_{\open{I}}\,\widetilde{\pi}\left(W(-\Phi,-\Pi) \right) U^*_{\open{I}}\,\pi_0\left(W(\Phi,\Pi) \right)$ is unitary we have found that there is a map 
 $\chi : \Cdata{\sdc{C}}\to U(1)$ with
 $$\pi_0\left(W(\Phi,\Pi) \right) = \chi(\Phi,\Pi)\:  U_{\open{I}}\,\widetilde{\pi}\left(W(\Phi,\Pi) \right) U^*_{\open{I}}\:.$$
 The fact that $\chi$ is additive, i.e. a character, follows immediately from Weyl relations. 
\end{proof}

\section{Spacetime formulation}\label{sec:spacetime formulation} 
Until now, we have worked with the ``spatial'' poset $\poset{R}$ made of open intervals of $\bS^{1}$ whose set-complement has non-empty interior, partially ordered under set inclusion; we have also endowed this poset with a causal disjointness relation $\perp$.

What we would like to do at this point is to reformulate the previous results, as far as possible, in terms of a ``spacetime'' poset $\poset{K}$, i.e. a poset being a topological basis for the entire spacetime; obviously we require $\poset{K}$ to be compatible with $\poset{R}$, in a certain sense to be specified. 

First of all we recall a notion belonging to Lorentzian geometry: given a subset $S$ of a time-oriented Lorentzian manifold $\spacetime{M}$, its \emph{Cauchy development} $D(S)$ is the set of points  of $\spacetime{M}$ through which every inextendible causal curve in $\spacetime{M}$ meets $S$.
For further details, see e.g \cite{Bar_Ginoux_Pfaffle_06} or \cite{O'Neill_83}. 

Our choice of the spacetime poset is quite natural: we take the set of diamonds based on elements of the spatial poset. Specializing definition \ref{def:diamonds}, given an element $\open{I}\in\poset{R}$ the diamond $\DIAMOND{\open{I}}$ based on $\open{I}$ is the set $\Int{D(\open{I}}$. Then we need to ``translate'' these diamonds, based on a specific Cauchy surface, to cover the entire spacetime $\spacetime{M}$. In our case, the cylindric flat universe, the spacetime admits a $1$-parameter group of isometries acting as time traslations, sending the point $x\equiv(\theta,s)$ to the point $x+(0,t)\equiv(\theta, t+s)$; so we can translate the poset $\poset{R}$ lying on the chosen Cauchy surface $\sdc{C}$ on another Cauchy surface belonging to the same foliation. If we denote with $\sdc{C}_{t}$ the Cauchy surface corresponding to the Cauchy time $t$ (so that by definition $\sdc{C}\equiv\sdc{C}_{t=0})$ we define the translated poset $\poset{R}_{t}$ on $\sdc{C}_{t}$ as the family of elements of the form $\open{I} + (0,t)$, $\open{I}\in\poset{R}$.
So we can construct the diamonds based on $\poset{R}_{t}$, and it's easy to see that the set $\poset{K}=\left\lbrace\DIAMOND{\open{I}}\,|\,\open{I}\in\poset{R}_{t},\:t\in\bR\right\rbrace$
 of the diamonds based on a generic Cauchy surface in the foliation is really a poset, endowed with an order relation (set inclusion) and a causal disjointness relation (that induced by the causal structure of $\spacetime{M})$.

Before we proceed some remarks are in order about the explicit form of geometrical objects (Cauchy surfaces, diamonds and Cauchy developments) on the cylindric flat universe; this kind of information will be of use in what follows.
Obviously, what we are going to say is a trivial consequence of Minkowskian geometry, taking into account that in our universe ``spatial endpoints'' are identified.
First of all, the spacetime admits a ``natural'' smooth foliation in Cauchy surfaces: with respect to canonical coordinates $(\theta, t)$ these are the hypersurfaces of equation $t=cost$, and are of course isometric to $\bS^{1}$. 
Diamonds, i.e. members of the spacetime poset $\poset{K}$, are very simple: they are squares, with a diagonal coinciding with an element of a spatial poset $\poset{R}_{t}$. We note that in this case $D(\open{I})$ is open, so $\DIAMOND{I}\equiv D(\open{I})$. 
Also causal complements of diamonds are simple: they are diamonds, too, based on the (spatial) causal complement of their basis. In symbols: \[\ccompl{\DIAMOND{\open{I}}}\equiv\DIAMOND{\open{I}^{\prime}},\quad \open{I}\in\poset{R}.\]
Finally, regarding Cauchy developments it's enough to observe that, by definition: $\DIAMOND{\open{I}^{\prime}}=\spacetime{M}\setminus\cdevel{\open{I}}$.

At this point we need to relate the net of local algebras defined on the Cauchy surface(s) with that defined on the entire spacetime.
In our case the symplectic space associated with the diamond $\open{O}\equiv\DIAMOND{\open{I}}$, $\open{I}\in\poset{R}_{t}$, is given by the (equivalence classes of) test functions supported in $\open{O}$: 	
\[\sympl{\open{O}}=\left\lbrace f\in\testfuncv{\open{O}}{\bR}\,|\,f \sim g \Leftrightarrow \cpropagator f = \cpropagator g \right\rbrace.\]
\noindent We previously saw that the casual propagator establishes an isomorphism between the global symplectic space of Cauchy data (on a given Cauchy surface $\sdc{C}_{t})$ and (equivalence classes of) test functions on the entire spacetime; consequently, if we restrict ourserlves to the symplectic subspace of Cauchy data supported in a open interval $\open{I}\in\poset{R}_{t}$ we obtain a symplectic subspace of test functions on the entire spacetime. Then the question arises of how we can describe this subspace; the following proposition gives an answer, at least in the present case. We start with a technical lemma \cite{Wald_94}.

\begin{lemma}\label{lem:Cauchy-slices}
Let $\spacetime{M}$ be a globally hyperbolic spacetime, $\psi\in\KGsol$ a smooth solution of the Klein-Gordon equation with compactly supported Cauchy data on every Cauchy surface, and $\sdc{C}_{1}$, $\sdc{C}_{2}$ two Cauchy surfaces belonging to the same foliation; suppose that $t$ is a Cauchy time for the foliation $\spacetime{M}=\sdc{C}\times\bR$ and that $\sdc{C}_{i}\equiv \sdc{C}(t_{i}),\: t_{1} < t_{2}$. Then there exists a test function $f\in\testfuncv{\spacetime{M}}{\bR}$ supported in the set $\tfuture{\sdc{C}_{1}}\cap\tpast{\sdc{C}_{2}}\cap\supp{\psi}$ such that $\psi=\cpropagator f$.
\end{lemma}
\begin{proof}
Choose $\chi\in\smoothfunc{\spacetime{M}}$ such that 

\[\chi(t) = 
\left\lbrace 
\begin{array}{rl}
0 & \mbox{if}\quad t\le t_{1}\\
1 & \mbox{if}\quad t\ge t_{2}
\end{array}
\right. \]
\noindent and take $f:= - \KGop(\chi\psi)$, where $K$ is the Klein-Gordon operator; then it's easy to see that $\supp{f}\subseteq\tfuture{\sdc{C}_{1}}\cap\tpast{\sdc{C}_{2}}\cap\supp{\psi}$ and is compact. 
It remains to show that $\psi=\cpropagator f$. We recall that $\cpropagator f$ is the difference between the advanced and retarded fundamental solutions with source term $f$. But is straigthforward to verify that:

\[
\left\lbrace 
\begin{array}{rl}
Af & = (1-\chi)\psi\\
Rf & = -\chi\psi
\end{array}
\right. \]
taking into account the uniqueness of the advanced and retarded fundamental solutions; then we obtain $\cpropagator f = (A-R)f=\psi$ as required.

\end{proof}

\begin{proposition}\label{prop:st_sympl}
Let $\poset{K}\ni\open{O}\equiv\DIAMOND{\open{I}}$, $\open{I}\in\poset{R}_{t}$, be a diamond; then the symplectic space $\sympl{\open{O}}$ is isomorphic to the symplectic space $\Cdata{\open{I}}$ of Cauchy data localized in $\open{I}$.
\end{proposition}
\begin{proof}
Let $[f]\in\sympl{\open{O}}$ such that $f$ is a test function supported in $\open{O}$; then $\psi\equiv\cpropagator f$ is the corresponding solution of the KG equation. Let $(\Phi, \Pi)$ be the Cauchy data on $\sdc{C}_{t}$ associated with $\psi$; then from the fact that $\supp{\psi}\subseteq\cdevel{\supp{f}}$ it follows that $\supp{\Phi}$, $\supp{\Pi}\subseteq\open{I}$, namely $(\Phi, \Pi)\in\Cdata{\open{I}}$ . So we have shown that $\sympl{\open{O}}\subseteq\Cdata{\open{I}}$.

On the other hand take $(\Phi, \Pi)\in\Cdata{\open{I}}$ and let $\psi$ be the induced solution of the KG equation ; we want to show that exists $[f]\in\sympl{\DIAMOND{\open{I}}}$ such that $\cpropagator f = \psi$; we already know that $\supp{\psi}\subseteq\cdevel{\supp{\Phi}\cup\supp{\Psi}}$. 

Given two Cauchy surfaces in the standard foliation for $\spacetime{M}$, say $\sdc{C}_{t_{1}}$ and $\sdc{C}_{t_{2}}$ with $t_{1} < t_{2}$ , we denote by $\open{N}$ the open set delimited by them, i.e. $\open{N}:=\cfuture{\sdc{C}_{t_{1}}}\cap\cpast{\sdc{C}_{t_{2}}}$; then, by lemma \ref{lem:Cauchy-slices} it follows that we can find $\open{O}\subset\open{N}\cap\supp{\psi}$ and $f\in\sympl{\open{O}}$ such that $\cpropagator f = \psi$ so the result is proven if we can chose $t_{1}< t_{2}$ small enough so $\open{N}\cap\cdevel{\supp{\Phi}\,\cup\,\supp{\Psi}}\subset\DIAMOND{\open{I}}$. But in the present case we can see (taking into account the actual geometry of Cauchy surfaces, diamonds and Cauchy developments), that it's true indeed.
In conclusion we have shown that $\sympl{\open{O}}\supseteq\Cdata{\open{I}}$, and the thesis follows.
\end{proof}
Following the construction outlined in section \ref{sec:abstract -> concrete algebras}, the net of symplectic spaces $\sympl{\open{O}}$ induces a net of local von Neumann algebras $\stnet{R}{\open{O}}$, $\open{O}\in\poset{K}$; keeping in mind that every member of the spacetime poset is a diamond based on an element of a spatial poset, proposition \ref{prop:st_sympl} implies at once that $\stnet{R}{\DIAMOND{\open{I}}} = \snet{R}{\open{I}}$, for every $\open{I}$ in a spatial poset $\poset{R}_{t}$. 

We could ask if properties discussed in section \ref{sec:properties} admit an extension to concrete spacetime algebras $\stnet{R}{\open{O}}$. The answer is affirmative, and we give it below with the following few propositions.

\begin{proposition}
In the cylindric flat universe, endowed with the net of local algebras $\net{R}{K}$ indexed by the poset $\poset{K}$ of diamonds, the following properties hold:\\
\textbf{Reeh-Schlieder}: For every diamond $\open{O}\in\poset{K}$ the vacuum vector $\psi\in\varhs{H}_0$ is cyclic for $\pi_0(\stweyl{\open{O}})$ and separating for $\stnet{R}{\open{O}}$.\\
\textbf{Local definiteness}: The local algebra associated with a point in the spacetime is trivial, i.e. 
\[\stnet{R}{{x}}:=\bigcap_{x\in\open{O}\in\poset{K}}\stnet{R}{\open{O}}= \bC\!\cdot\! 1,\quad\forall\, x\in\spacetime{M}.\]\\
\textbf{Factoriality}: Every local algebra  $\stnet{R}{\open{O}}$ is a factor in $\bop{\varhs{H}_0}$.\\
\textbf{Additivity}: Let $\open{O}$, $\left\lbrace \open{O}_{i}\right\rbrace_{i\in L}\in\poset{K}$ be such that $\open{O}=\cup_{i\in L} \open{O}_{i}$; then
\[ \left( \bigcup_{i\in L} \pi_0\left( \stweyl{\open{O}_{i}}\right)\right) ^{\prime\prime}= \stnet{R}{\open{O}}. \]
\end{proposition}

\begin{proof}
  \noindent\\
\textit{Reeh-Schlieder}: since $\pi_0(\stweyl{\open{O}})=\pi_0(\sweyl{\open{I}})$ and $\stnet{R}{\open{O}}=\snet{R}{\open{I}}$ the statement is trivial.\\  
\textit{Local definiteness}: let $\sdc{C}$ be a Cauchy surface containing $x$; if $\{x\}\subset\open{I}\in\poset{R}_t$ then obviously $\{x\}\subset\DIAMOND{\open{I}}\in\poset{K}$; so 
\[\bigcap_{x\in\open{O}\in\poset{K}}\stnet{R}{\open{O}}\subseteq\bigcap_{\poset{R}_t\ni\open{I}\supset\{x\}}\snet{R}{\open{I}} = \bC\!\cdot\!1,\] 
from the spatial version of the statement. But 
$$\bigcap_{\poset{K}\ni\open{O}\supset\{x\}}\stnet{R}{\open{O}}$$ is a subalgebra of $\bop{\varhs{H}_0}$, so is all $\bC\!\cdot\!1$.\\
\textit{Factoriality}:	 Obvious since $\stnet{R}{\open{O}}=\snet{R}{\open{I}}$, that is a factor.\\	
\textit{Additivity}:	The proof is identical to the spatial one, substituting $\open{I}$ with $\open{O}$.

\end{proof}

It remains to discuss Haag duality and punctured Haag duality; for the sake of convenience, we recall some definitions about causal complements.
\begin{definition}
Let $\open{I}\in\poset{R}_{t}$; we define the \emph{causal complement} of $\open{I}$ as the family of sets $\ccompl{\open{I}} := \{\open{I}_{1}\in\poset{R}_{t}\,|\,\open{I}_{1}\perp \open{I}\}$.\\ 
We note that the open set  
\[\sdc{C}_{t}^{\ccompl{\open{I}}}:=\bigcup_{\open{I}_{1}\in \ccompl{\open{I}}} \open{I}_{1}\]
coincides with the ordinary causal complement of $\open{I}$ in $\sdc{C}_{t}$.
\end{definition}

\begin{definition}
Let $\open{O}\in\poset{K}$; we define the \emph{causal complement} of $\open{O}$ as the family of sets $\ccompl{\open{O}} := \{\open{O}_{1}\in\poset{K}\,|\, \open{O}_{1}\perp \open{O}\}$.\\ 
We note that the open set  
\[\spacetime{M}^{\ccompl{\open{O}} }:=\bigcup_{\open{O}_{1}\in \ccompl{\open{O}}} \open{O}_{1}\]
coincides with the ordinary causal complement of $\open{O}$.
\end{definition}

\begin{proposition}[Haag Duality]
In the spacetime $\spacetime{M}$, endowed with the net of local (concrete) algebras $\net{R}{K}$, Haag duality holds, i.e. given $\open{O}\in\poset{K}$ the following identity is satisfied:
\begin{equation}
\stnet{R}{\open{O}}=\bigcap_{\open{O}_{1}\in\ccompl{\open{O}}}\stnet{R}{\open{O}_{1}}^{\prime} 
\end{equation} 

\begin{proof}
We already know that \emph{spatial} Haag duality holds in the cylindric flat spacetime; we'll show that \emph{spacetime} Haag duality follows from this.
Let $\sdc{C}$ be the Cauchy surface where $\open{O}$ is based; then 
\[\bigcap_{\open{O}_{1}\in\ccompl{\open{O}}}\stnet{R}{\open{O}_{1}}^{\prime} = \bigcap_{\open{I}_{1}\in\ccompl{\open{I}}}\snet{R}{\open{I}_{1}}^{\prime}. \]
In fact we know that $\forall\, \open{I}_{1}\in\ccompl{\open{I}}$, $\open{O_{1}}:=\DIAMOND{\open{I_{1}}}\in\ccompl{\open{O}}$ and $\stnet{R}{\DIAMOND{\open{I_{1}}}} = \snet{R}{\open{I_{1}}}$, so \[\bigcap_{\open{O}_{1}\in\ccompl{\open{O}}}\stnet{R}{\open{O}_{1}}^{\prime}\subseteq\bigcap_{\open{I}_{1}\in\ccompl{\open{I}}}\snet{R}{\open{I}_{1}}^{\prime}.\]
Now let $\open{O}_{1}\in\ccompl{\open{O}}$; then $\open{O_{1}}\subseteq\DIAMOND{\open{I}^{\prime}}$ and by isotony we have \[\stnet{R}{\open{O_{1}}}\subseteq\stnet{R}{\DIAMOND{\open{I}^{\prime}}}=\snet{R}{\open{I}^{\prime}}.\] 
Applying the algebra commutant to both sides of the inclusion we have 
\[\stnet{R}{\open{O_{1}}}^{\prime}\supseteq\snet{R}{\open{I}^{\prime}}^{\prime}.\] 

\noindent Doing some algebra we obtain:
\begin{equation}
\begin{split}
&\stnet{R}{\open{O_{1}}}^{\prime}\cap\left(\bigcap_{\open{I}_{1}\in\ccompl{\open{I}}}\snet{R}{\open{I}_{1}}^{\prime}\right)        =\stnet{R}{\open{O}^{\prime}}^{\prime}\cap\left(\bigcap_{\open{I}_{1}\in\ccompl{\open{I}}}\stnet{R}{\DIAMOND{\open{I}_{1}}}^{\prime}\right)\supseteq\\
&\supseteq\snet{R}{\open{I}^{\prime}}^{\prime}\cap\left(\bigcap_{\open{I}_{1}\in\ccompl{\open{I}}}\snet{R}{\open{I}_{1}}^{\prime}\right)=\bigcap_{\open{O}_{1}\in\ccompl{\open{O}}}\stnet{R}{\open{O}_{1}}^{\prime},\\  
\end{split}
\end{equation}
given that $\open{I}_{1}\in\ccompl{\open{I}}$.
\noindent Repeating this argument for all $\open{O}_{1}\in \ccompl{\open{O}}$ we have, as claimed: 

\[\bigcap_{\open{O}_{1}\in\ccompl{\open{O}}}\stnet{R}{\open{O}_{1}}^{\prime}=\bigcap_{\open{I}_{1}\in\ccompl{\open{I}}}\snet{R}{\open{I_{1}}}^{\prime}.\] 

\noindent Finally we conclude observing that \[\stnet{R}{\open{O}}=\snet{R}{\open{I}}=\bigcap_{\open{I}_{1}\in\ccompl{\open{I}}}\snet{R}{\open{I}_{1}}^{\prime}=\bigcap_{\open{O}_{1}\in\ccompl{\open{O}}}\stnet{R}{\open{O}_{1}}^{\prime}.\]
\end{proof}
\end{proposition}

\noindent Now we examine punctured Haag duality. We recall \cite{Ruzzi_05} that given a point $x$ in $\spacetime{M}$, the \emph{causal puncture} of the spacetime poset $\poset{K}$ is the poset $\poset{K}_{x}$ defined as the collection 
\[
\poset{K}_{x}=\left\lbrace \open{O}\in\poset{K}|\cl{\open{O}}\perp x\right\rbrace 
\]
ordered under inclusion, where $\cl{\open{O}}\perp x$ means that $\cl{\open{O}}\subseteq\spacetime{M}\setminus\cdevel{x}$. We also recall that its ``topological realization'' $\spacetime{M}_{x}:=\cup\left\lbrace \open{O}\in\poset{K}|\open{O}\in\poset{K}_{x}\right\rbrace$ coincides with $\spacetime{M}\setminus\cdevel{x}=D(\sdc{C}\setminus {x})$, where $\sdc{C}$ is a (spacelike) Cauchy surface containing $x$.
Considered as a spacetime on its own, $\spacetime{M}_{x}$ is globally hyperbolic and $\poset{K}_{x}$ is a topological basis of $\spacetime{M}_{x}$. Now, for any $\open{O}\in\poset{K}_{x}$ we define the \emph{causal complement} of $\open{O}$ in $\poset{K}_{x}$ as the family  of sets:

\[
\ccompl{\open{O}}|_{\poset{K}_{x}}:=\left\lbrace \open{O}_{1}\in\poset{K}_{x}|\open{O}_{1}\perp \open{O}\right\rbrace 
\]
We remark that in general an element of the poset $\poset{K}_{x}$ it's not a diamond of the spacetime $\spacetime{M}_{x}$, although it's a diamond of the spacetime $\spacetime{M}$ by definition.

\noindent We are now ready to formulate punctured Haag duality in a spacetime context.

\begin{proposition}[Punctured Haag duality]
Let $x\in\spacetime{M}$, $\poset{K}_{x}$ the causal puncture of $\poset{K}$ induced by $x$, $\open{O}\in\poset{K}_{x}$; then punctured Haag duality holds. In other words the following identity is true:
\begin{equation}
\stnet{R}{\open{O}}=\bigcap_{\open{O}_{1}\in\ccompl{\open{O}}_{x}}\stnet{R}{\open{O}_{1}}^{\prime} 
\end{equation} 
\end{proposition}
\begin{proof}
Let $\open{O}\in\poset{K}_{x}$; denoting by $\ccompl{O}$ the causal complement of $\open{O}$ $\in \spacetime{M}$, we already know that $\stnet{R}{\open{O}}^{\prime} = \stnet{R}{\ccompl{O}}$ (Haag duality); we have to show that $\stnet{R}{\open{O}}^{\prime}$ = $\stnet{R}{\ccompl{O}\setminus\left(\open{V}_{1}\cup\open{V}_{2}\right)}$. 
Then we see it is enough to prove that $\stnet{R}{\ccompl{O}}$ = $\stnet{R}{\ccompl{O}\setminus\left(\open{V}_{1}\cup\open{V}_{2}\right)}$.
It's easy to see that each of the subsets $\open{V}_{i}$, $i=1,\ldots,4$ making up $\ccompl{\open{O}}$ is a globally hyperbolic spacetime (see lemma \ref{lem:global_hyperbolicity_condition} below). 
Take a piece of a ``straight'' Cauchy surface (namely belonging to the natural foliation of $\spacetime{M})$ around $x$, say $\Delta_{n}$, and extend it\footnote{According to \cite{Bernal_Sánchez_06} this construction is always possible.} to a Cauchy surface $\sdc{C}$ of $\ccompl{\open{O}}$ such that $\sdc{C}$ = $\sdc{C}_{1}\cup\sdc{C}_{2}\cup\Delta_{n}$, where  $\sdc{C}_{1}(\sdc{C}_{2})$ is a Cauchy surface of $\open{V}_{3} (\open{V}_{4})$.
Now transport the local algebra $\stnet{R}{\ccompl{\open{O}}}$ on $\sdc{C}$ (passing from test functions to Cauchy data, as we have seen in \ref{sec:symplectic_spaces}); using a partition of unity, split Cauchy data in three parts, supported respectively in $\sdc{C}_{1}$, $\sdc{C}_{2}$ and $\Delta_{n}$. Passing to local algebras we obtain 
\begin{equation}\label{eq:punct-proof-split}
\stnet{R}{\ccompl{O}} = \snet{R}{\sdc{C}} = \snet{R}{\sdc{C_{1}}}\vee\snet{R}{\sdc{C_{2}}}\vee\snet{R}{\Delta_{n}} 
\end{equation}
If we choose $\Delta_{n+1}$ such that $\Delta_{n+1}\subset\Delta_{n}$ we get $\snet{R}{\Delta_{n+1}}\subseteq\snet{R}{\Delta_{n}}$; because (\ref{eq:punct-proof-split}) is valid for all $n\in\bN$ we finally obtain
\begin{eqnarray}
\stnet{R}{\ccompl{O}} & = &  \snet{R}{\sdc{C_{1}}}\vee\snet{R}{\sdc{C_{2}}}\vee\snet{R}{\Delta_{n}} = \nonumber \\
=\snet{R}{\sdc{C_{1}}}\vee\snet{R}{\sdc{C_{2}}}\vee\bC\cdot 1 & = & \snet{R}{\sdc{C_{1}}}\vee\snet{R}{\sdc{C_{2}}} = \nonumber \\
=\stnet{R}{\open{V}_{3}}\vee\stnet{R}{\open{V}_{4}} & = & \stnet{R}{\ccompl{O}\setminus\left(\open{V}_{1}\cup\open{V}_{2}\right)}.
\end{eqnarray}
\end{proof}

We conclude this section with some material about Lorentzian geometry used in the proof above; to start with we recall the causality condition for spacetimes.
\begin{definition}
A Lorentzian manifold is said to satisfy the \emph{causality condition} if it doesn't contain any closed causal curve; it's said to satisfy the \emph{strong causality condition} if there are no almost closed causal curves. More precisely, for each point $p\in\spacetime{M}$ and for each open neighborhood $\open{U}$ of $p$ there exists an open neighborhood $\open{V}\subset\open{U}$ of $p$ such that each causal curve in $\spacetime{M}$ starting and ending in $\open{V}$ is entirely contained in $\open{U}$.
\end{definition}
\noindent Now we give another characterization of globally hyperbolic spacetimes.
\begin{lemma}\label{lem:global_hyperbolicity_condition} 
A connected time-oriented Lorentzian manifold is globally hyperbolic if and only if it satisfies the strong causality condition and for all $p,q\in\spacetime{M}$ the intersection $\cpast{p}\cap\cfuture{q})$ is compact.
\end{lemma}

\chapter{Topological cocyles}\label{chap:topological_cocycles}
We finally come to the core of this work; in this chapter we exhibit a class of topologically non-trivial cocycles living in the cylindric universe. 
What about the significance of such a result ? As we'll see presently, the cocycle machinery has a twofold purpose: from the one hand, it allows us to build unitary representations of the first homotopy group of the spacetime (via theorem \ref{th:cocycle-representations}), from the other hand it can be used to construct net-representations.  So existence of non-trivial 1-cocycles implies both the existence of non-trivial representations of the fundamental group and non-trivial net-representations of the observables net. In this respect, we can say that spacetime topology interacts with algebraic QFT. 
\section{Introduction} 
In this section we show how  the non-trivial topology of $\bS^1$ allows the existence non-trivial unitary representations of the first homotopy group of the spacetime acting in the Hilbert space of the vacuum. As before, we refer to the (spatial) poset $\poset{R}$ and the vacuum representation induced by the vacuum state $\omega_{0}$,  with GNS Hilbert space  $\varhs{H}_{0}:= \fock{\hs{H}_{0}}$ and one-particle space $\hs{H}_{0}:= \sqint{\bS^1}{\theta}$. Hence the reference net of local algebras indexed by $\poset{R}$ will be 
$$\alg{R} : \poset{R} \ni \open{I} \to \snet{R}{\open{I}} \subset \bop{\varhs{H}_{0}},$$
satisfying {\em irreducibility}, {\em Haag duality}, {\em punctured Haag duality}. 

Moreover we shall establish that non-trivial cocycles individuate generalized net representations of the algebra of observables satisfying the (topological) selection criterion defined in \cite{Brunetti_Ruzzi_08}.
The found situation is however puzzling here, because some of the most important results presented in the mentioned paper (which are valid for dimension of the spacetime $\geq 3$) cannot be used here and, as matter of fact, they do not hold true. In particular the theorem of localization of fundamental group, theorem 4.1 in \cite{Brunetti_Ruzzi_08}, fails to be valid
specializing to our definition of cocycles as it can be verified by direct inspection. Therefore, some relevant consequences 
of that theorem may be false in our low-dimension context, and the actual situation has to be analyzed carefully.
In particular, the absence of  irreducible cocycles different from characters when the fundamental group of the manifold is Abelian, as established in  Corollary 6.8 in \cite{Brunetti_Ruzzi_08}, no longer holds. Indeed, the irreducible cocycles we are going to 
 present below are a clear counterexample. Furthermore the proof of Theorem 4.3 
in \cite{Brunetti_Ruzzi_08}, which demonstrates 
 equivalence of the category of $1$-cocycles and the category of representations of  the algebra of observables satisfying the (topological) selection criterion, does not hold since it relies upon theorem 4.1 in \cite{Brunetti_Ruzzi_08}. However, it turns out that various results carry over to our model, although their proofs often require substantial modifications. 

\section{General theory}
First of all we need to set the stage for subsequent constructions. As we said, we choose the poset $\poset{R}$ of proper, open intervals on $\bS^{1}$, equipped with the partial order relation $\subseteq$ and the causal disjointness relation $\perp$ (see section \ref{sec:the_cylindric_flat_universe}). Our reference net of  observables will be the net of unital $C^*$-algebras (Weyl algebras)  
$$\net{W}{R} : \poset{R} \ni \open{I} \to \sweyl{\open{I}}.$$ 
\noindent The vacuum representation\footnote{For the sake of notational simplicity, and to be consistent with \cite{Brunetti_Ruzzi_08}, henceforth we denote the vacuum representation $\pi_0$ with the symbol $\iota$, emphasizing the fact that $\pi_0$ is our reference representation.} $\iota$ on the vacuum GNS Hilbert-Fock space $\varhs{H}_{0}$ discussed in section \ref{sec:vacuum_representations}, induces a unitary net representation $\{\iota,\bI\}$ of the net of observables $\net{W}{R}$, namely a family of representations $\iota_{\open{I}}: \sweyl{\open{I}}\to\iota_{\open{I}}(\sweyl{\open{I}})\subseteq \bop{\varhs{H}_{0}}$; in this case the ``transition functions'' $ \bI_{\widetilde{\open{I}}\open{I}}$ are the identity operators for every $\widetilde{\open{I}}$, $\open{I}\in\poset{R}$. The unitary net representation $\{\iota,\bI\}$ gives also rise to a von Neumann algebras net $\poset{R} \ni \open{I} \mapsto \iota(\sweyl{\open{I}})''=: \snet{R}{\open{I}}$ satisfying {\em irreducibility}, {\em Haag duality}, {\em punctured Haag duality} and {\em Borchers property}.
Finally $(\iota, \bI)$  is topologically trivial since the cocycle $\zeta^\iota$ associated with it is the simplest co-boundary (hence the associated unitary representation of $\stdhomgroup{\bS^1}$ is trivial, too).   
We refer to notation, definitions and results discussed in section \ref{sec:vacuum_representations}, noting en passant that corollary \ref{cor:not_directed_poset} entails that $(\poset{R}, \subseteq)$ cannot be directed (as we know independently).

Finally it is worth remarking that every irreducible unitary representation of $\bZ$, $\{\lambda_x(n)\}_{n\in\bZ}$, is one-dimensional, $\lambda_x(n) : \bC \to \bC$, as the group is Abelian. All those representations are labeled by $x\in \bR$ and have the form:
\begin{equation} \label{eq:lambdax}
 \lambda_x(n) : \bC \ni \psi \mapsto e^{\imath n x} \psi\:,\quad\mbox{for all $n \in \bZ$}.
\end{equation}
Some remarks about the poset $\poset{R}$ are in order too. 
In view of section \ref{sec:the_first_homotopy_group_of_a_poset} and since $\bS^1$ is Hausdorff, arcwise connected and $\poset{R}$ is a topological basis of
$\bS^1$, it turns out that $\shomgroup{R}$ doesn't depend on the basepoint $\open{I}_0$ and coincides with the fundamental group of $\bS^1$, i.e. $\shomgroup{R} = \bZ$ in our case. 
For the sake of clarity, it's worth restating theorem \ref{th:cocycle-representations} specialized to our case:
\begin{theorem}\label{th:cocycle_representations_spec}
 Fix a base $0$-simplex $\open{I}_0$; given $\coc{z}\in\cocyclecat{\net{R}{R}}$, define the following map sending an equivalence class of (poset) paths to an operator in $\bop{\varhs{H}_0}$:
\begin{equation}\label{eq:cocycle_repr_spec}
\pi_{\coc{z}}([\path{p}]):=\coc{z}(\path{p}),\quad[\path{p}]\in\shomgroup{R}. 
\end{equation}
The correspondence $\cocyclecat{\net{R}{R}}\ni\coc{z}\mapsto\pi_{\coc{z}}$ maps $1$-cocycles $\coc{z}$, equivalent in $\bop{\varhs{H}_0}$, into equivalent unitary representations $\pi_{\coc{z}}$ of $\shomgroup{R}$ in $\varhs{H}_0$; up to equivalence this map is injective. 
\end{theorem}
\noindent Notice that, as a consequence, $z\in\cocyclecat{\net{R}{R}}$ is trivial if and only if the associated representation of
$\pi_z$ is the trivial one made of the unit operator only.
\subsection{Topological superselection sectors}

Now we introduce a selection criterion which generalizes DHR selection criterion; we refer to \cite{Brunetti_Ruzzi_08} for motivations and further details. Consider as before the unitary net representation $(\iota, \bI)$ of $\net{W}{R}$ over $\varhs{H}_{0}$; following \cite{Brunetti_Ruzzi_08} we say that a net representation $\{\pi,\psi\}$ over $\varhs{H}_{0}$
is {\bf a sharp excitation of the reference representation} $\{\iota, \bI\}$, if for any $\open{J}\in \poset{R}$ and for any simply connected open set $\open{N} \subset \bS^1$, such that $\bar{\open{J}} \subset \open{N}$, it holds
\begin{equation}\label{eq:SS}
 \{\pi,\psi\}\rest_{\open{J}'\cap \open{N}} \quad \cong  \quad \{\iota, \bI\}\rest_{\open{J}'\cap \open{N}}\:. 
\end{equation}
This amounts to saying that there is a family $W^{\open{N}\open{J}} 
:= \{W^{\open{N}\open{J}}_\open{I} \:|\:  \overline{\open{I}} \subset \open{N}\:, \open{I} \subset  \open{J}'\}$ of
unitary operators in $\varhs{H}_{0}$ such that:\\

1.   $W^{\open{N}\open{J}}_\open{I}\pi_\open{I} = \iota_\open{I} W^{\open{N}\open{J}}_\open{I}$;\\

2.    $W^{\open{N}\open{J}}_\open{I}\psi_{\open{I}\widetilde{\open{I}}} = W^{\open{N}\open{J}}_{\widetilde{\open{I}}}$,      for all $\widetilde{\open{I}} \subset \open{I}$;\\

3.   $W^{\open{N}\open{J}}  = W^{\open{N}_1\open{J}}$  for any simply connected open set  $\open{N}_1$ with $\open{N} \subset  \open{N}_1$.\\

\noindent These three requirements represent the {\bf selection criterion}.
It turns out \cite[p. 14]{Brunetti_Ruzzi_08} that  $W^{\open{N}\open{J}}$ is independent from the region $\open{N}$. 
The unitary equivalence classes of irreducible unitary net representations satisfying the selection criterion are the 
{\bf superselection sectors} and the analysis of their charge structure and topological content, in the case of a generic globally hyperbolic spacetime  with dimension $\geq 3$ was the scope of the work \cite{Brunetti_Ruzzi_08}. We are dealing with a (particular) $2$-dimensional spacetime, so we expect that some of the results found there cannot apply.

\subsection{Localized cocycles}
One of the most important results, established in \cite[Th. 4.3]{Brunetti_Ruzzi_08}, is that, for globally hyperbolic spacetimes with dimension $\geq 3$, the $C^*$-category whose objects are sharp excitations of $\{\iota, \bI\}$ with arrows given by 
intertwiners, is equivalent to the subcategory of localized $1$-cocycles $\cocyclecat{\net{R}{R}}$ (with respect to $\iota$ and the associated net of von Neumann algebras $\net{R}{R} : \poset{R} \ni \open{I} \mapsto \iota(\open{I})'' =: \snet{R}{\open{I}}$). 

Although this $C^*$-category can be defined without any modifications in our 2-dimensional model, there is no guarantee for the validity of the equivalence theorem. In any cases, the proof fails to go through in the same form as for larger dimension, since an important  lemma exploited in the proof of Theorem 4.3, Theorem 4.1 of  \cite{Brunetti_Ruzzi_08} is not valid in our spacetime.\\
 
\noindent\textbf{Remark.} The absence of  irreducible cocycles different from characters when the fundamental group of the manifold is Abelian, as established in  Corollary 6.8 in \cite{Brunetti_Ruzzi_08}, no longer holds in low dimension. 
Indeed, the irreducible cocycles we are going to  present are a  counterexample.

\section{Topological cocycles}

In this section we show explicitly a topological $1$-cocycle; actually, we'll construct an entire family of (localized) $1$-cocycles with respect to the reference vacuum representation $\{\iota, \bI\}$. 
To achieve this result we need a preliminary construction. 
First of all, fix an orientation on $\bS^1$ and, in the following, refer to that orientation 
the former and the latter endpoints of $0$-simplices.
Afterwards,  define an assignment of smooth functions to $0$-simplices
\begin{equation}\label{eq:assignchi}
 \chi: \simplex{0}{\poset{R}} \ni a \mapsto \chi^{(a)} \in \smoothfuncv{a}{\bR} 
\end{equation}
 such that:
 
(i)  $\chi^{(a)}(\theta) \in [0,1]$,
 
(ii) $\chi^{(a)}(\theta) =0$ in a  neighborhood of the former endpoint of $a$,

(iii) $\chi^{(a)}(\theta) =1$ in a  neighborhood  of the latter endpoint  of $a$. 

\noindent Now we pass to $1$-simplices. Consider a $1$-simplex $b$; extend
$\chi^{(\partial_1b)}$ and $\chi^{(\partial_0b)}$ smoothly and uniquely as constant functions over $|b| \setminus \partial_1b$ and $|b| \setminus \partial_0b$; the functions so extended over the whole $|b|$ will be denoted by $\chi^{(\partial_1b)}$ and $\chi^{(\partial_0b)}$ again; finally, for every $b\in \simplex{1}{\poset{R}}$ we define the function $\chi^{(b)} \in \smoothfuncv{|b|}{\bR}$
\begin{equation} \label{eq:assignchi2} 
\chi^{(b)} := \chi^{(\partial_1b)} -  \chi^{(\partial_0b)}.
\end{equation}
Notice that this function vanishes in a neighborhood of each endpoint of $|b|$, because it must hold 
$\chi^{(\partial_1b)}(\theta) -  \chi^{(\partial_0b)}(\theta) = 0-0$ around the former endpoint, and 
$\chi^{(\partial_1b)}(\theta) -  \chi^{(\partial_0b)}(\theta) = 1 - 1$ around the latter. Therefore 
$\chi^{(b)}$ can be extended uniquely to a smooth function defined on the whole circle $\bS^1$ and supported in $|b|$. We shall
denote by $\chi^{(b)}$ again this unique extension.

Let us come to $1$-cocycles. 
Adopting the notation (\ref{eq:Ksigmaext}), we define (where as before $\iota$ is the vacuum representation of $\net{W}{R}$)
\begin{equation}\label{eq:Vf}
Z(f,g):= \iota\left(W(f,g) \right) =  W\left[2^{-1/2} (A^{1/4} f + \imath A^{-1/4} g)\right] 
\end{equation}
for $(f,g) \in \Cdata{\sdc{C}}$. In the following, to define a $1$-cocycle localized at $b\in \simplex{1}{\poset{R}}$, we shall replace the arguments $f$ 
and $g$ with $\chi^{(b)}$-smeared restrictions of those functions to 
$0$-simplices $|b|$ for any $1$-simplex $b$. The restriction is necessary in order to fulfill the localization requirement of $1$-cocycles.
 The smearing procedure is necessary too, at least for the entry of $A^{1/4}$, whose domain generally
does not include elements $\chi_{|b|}g$,
$\chi_{|b|}$ being the characteristic function of the set $|b|$.
 It, however, includes  every smoothed function
 $\chi^{(b)} g$ when $(f,g) \in \Cdata{\sdc{C}}$. \\
We are now in place to state our result defining, in fact, a family of $1$-cocycles labeled by $\chi$, $f$, $g$. The following theorem also establishes independence from 
$\chi$, up to equivalence  of the defined $1$-cocycles, irreducibility of cocycles and the fact that they are inequivalent 
if $f\neq f'$ or $g\neq g'$. \\

\begin{theorem}\label{th:propcocycle}
Fix an orientation of $\bS^1$, an assignment   $\chi : \simplex{0}{\poset{R}} \ni a \mapsto \chi^{(a)}$ as
in (\ref{eq:assignchi}) and define $\chi^{(b)}$ as in (\ref{eq:assignchi2}).
For every choice of $(f,g) \in \Cdata{\sdc{C}}$ the map 
\begin{equation} \label{eq:cocyclef}
\coc{z}^{(\chi)}_{(f,g)} : \simplex{1}{\poset{R}} \ni b \mapsto Z\left( \chi^{(b)}f, \chi^{(b)}g\right)\:,
\end{equation}
 is a  $1$-cocycle of $\cocyclecat{\net{R}{R}}$. The following further facts hold.\\
{\bf (a)} Every $1$-cocycle $\coc{z}^{(\chi)}_{(f,g)}$ is irreducible.\\
{\bf (b)}  For fixed $(f,g)\in \Cdata{\sdc{C}}$, but different assignments $\chi_1,\chi_2$, the cocycles 
 $\coc{z}^{(\chi_1)}_{(f,g)}$ and $\coc{z}^{(\chi_2)}_{(f,g)}$ are equivalent.\\
 {\bf (c)}   For a fixed assignment $\chi$, the $1$-cocycles  $\coc{z}^{(\chi)}_{(f,g)}$ and $\coc{z}^{(\chi)}_{(f',g')}$ are equivalent if and only 
 if $f=f'$ and $g=g'$.\\
{\bf (d)} If the assignment $\chi : \simplex{0}{\poset{R}} \ni a \mapsto \chi^{(a)}$ is covariant\footnote{Covariant assignments $\chi : \simplex{0}{\poset{R}} \ni a \mapsto \chi^{(a)}$
with respect to the isometry group of $\theta$-displacements on $\bS^1$ do exist
 as the reader can easily prove.}
with respect to the isometry group of $\theta$-displacement on $\bS^1$:
\begin{equation}\label{eq:covset}
\chi^{(\beta_r(a))} = \beta^*_r\chi^{(a)} \quad \forall\, a \in \simplex{0}{\poset{R}}, \forall\, r\in \bR,  
\end{equation}
then, for every $(f,g) \in \Cdata{\sdc{C}}$  and for every $r\in\bR$ and $b\in \simplex{1}{\poset{R}}$,
$$U_{(r,0)} \coc{z}^{(\chi)}_{(f,g)}(b) U_{(r,0)}^* = \coc{z}^{(\chi)}_{(\beta^*_r (f) ,\beta^*_r(g))} (\beta_r(b)),$$
where $U_{(r,0)}$ is the one-parameter unitary group implementing $\theta$-displacements $\beta_r$ and leaving the vacuum invariant, introduced in theorem \ref{prop:hyperbolic_generators}, and $\beta_r^*$ is the pull-back action of $\theta$-displacements on functions defined on $\bS^1$. 
\end{theorem}
\begin{proof}
Let us prove that (\ref{eq:cocyclef})  defines a localized $1$-cocycle. First we notice that the unitary operator
$$Z(\chi^{(b)}f,\chi^{(b)}g):= W\left[2^{-1/2} (A^{1/4} (\chi^{(b)}f) + \imath A^{-1/4} (\chi^{(b)}g))\right]$$
is an element of $\snet{R}{|b|}$ since $\supp{\chi^{(b)}}\subset |b|$ as previously noticed.
 So, only the identity (\ref{eq:cocycle_identity}) remains to be proved. It can be established as follows. Consider a 
 $2$-cocycle $c$. To simplify the notation we define $b_k := \partial_k c$ for $k=0,1,2$.
 Since $|c|$ cannot coincide with the whole circle (and this is the crucial point),
 all functions $\chi^{(\partial_1 b_j)}$ can uniquely and smoothly be extended to functions defined on $|c|$ by defining them as constant functions outside their original domain. The extension procedure does not affect the definition of the functions $\chi^{(b_i)}$. We shall exploit this extension from now on.
 We have to show that $\coc{z}^{(\chi)}_{(f,g)}(b_0)\coc{z}^{(\chi)}_{(f,g)}(b_2) = \coc{z}^{(\chi)}_{(f,g)}(b_1)$, that is
 $$Z\left(\chi^{(b_0)}f, \chi^{(b_0)} g\right)
Z\left(\chi^{(b_2)}f, \chi^{(b_2)} g\right) = Z\left(\chi^{(b_1)}f, \chi^{(b_1)} g\right)$$
that is, in turn, 
\begin{equation}\label{eq:idcocycle} 
\begin{split}
&Z\left((\chi^{(\partial_1b_0)}-\chi^{(\partial_0b_0)})f,(\chi^{(\partial_1b_0)}-\chi^{(\partial_0b_0)})g\right)\cdot\\
&\cdot Z\left((\chi^{(\partial_1b_2)}-\chi^{(\partial_0b_2)})f,(\chi^{(\partial_1b_2)}-\chi^{(\partial_0b_2)})g\right)=\\  
&= Z\left((\chi^{(\partial_1b_1)}-\chi^{(\partial_0b_1)})f,(\chi^{(\partial_1b_1)}-\chi^{(\partial_0b_1)})g\right). 
\end{split}
\end{equation}
Now notice that, in view of the very definition of a $2$-cocycle, $\partial_1 b_1= \partial_1b_2$, 
$\partial_0 b_0= \partial_0b_1$ and $\partial_0 b_2 = \partial_1 b_0$, so that the left-hand side of (\ref{eq:idcocycle})
can be recast as
\begin{equation}
\begin{split}
&Z\left((\chi^{(\partial_1b_1)}-\chi^{(\partial_0b_2)})f,(\chi^{(\partial_1b_1)}-\chi^{(\partial_0b_2)})g\right)\cdot\\
&\cdot Z\left((\chi^{(\partial_0b_2)}-\chi^{(\partial_0b_1)})f,(\chi^{(\partial_0b_2)}-\chi^{(\partial_0b_1)})g\right)
\end{split}
\end{equation}
where {\em all} functions $\chi^{(\partial_ib_j)}$ are now defined on the whole $|c|$
and the differences $\chi^{(\partial_0b_2)}-\chi^{(\partial_0b_1)}$, $\chi^{(\partial_1b_1)}-\chi^{(\partial_0b_2)}$ are  defined everywhere on $\bS^1$ and compactly supported in $|c|$. Finally, making use of Weyl relations, taking the definition (\ref{eq:Vf}) of $Z(f,g)$ into account, we find that the terms $\pm\chi^{(\partial_0b_2)}$ cancel each other in the final exponent, and the left-hand side of (\ref{eq:idcocycle}) is:
\begin{equation}
\begin{split}
&Z\left((\chi^{(\partial_1b_1)}-\chi^{(\partial_0b_1)})f,(\chi^{(\partial_1b_1)}-\chi^{(\partial_0b_1)})g\right)\cdot\\
&\cdot e^{\imath \int_{\bS^1} (\chi^{(\partial_1b_1)}-\chi^{(\partial_0b_2)})(\chi^{(\partial_0b_2)}-\chi^{(\partial_0b_1)})(fg-gf)d\theta}\\ 
\end{split}
\end{equation}
Since the phase vanishes, we have found the very right-hand side of  (\ref{eq:idcocycle})  as wanted.\\
\textbf{(a)} Let us pass to the irreducibility property of the defined cocycles.
Let $\itwiner{V}: \simplex{0}{\poset{R}} \ni \open{I} \mapsto \itwiner{V}_\open{I} \in \snet{R}{\open{I}}$ be a field of
 unitary operators such that
\begin{equation}\label{eq:equivalence3}
\itwiner{V}_{\partial_0 b}  = \coc{z}^{(\chi)}_{(f,g)}(b) \itwiner{V}_{\partial_1 b}\, \coc{z}^{(\chi)*}_{(f,g)}(b)\:, \quad  \quad \mbox{for all $b\in \simplex{1}{\poset{R}}$.} 
\end{equation}
Since $\itwiner{V}_{\partial_1 b} \in \snet{R}{\partial_1b}$, then $\itwiner{V}_{\partial_1 b} = \sum_k c_k \iota(W(r_k,s_k))$ where $r_k,s_k$ are smooth real functions supported in $\partial_1b$, $c_k \in \bC$, and the series converges in the strong operatorial topology. Therefore, using Weyl relations,
\begin{equation}
\begin{split}
\coc{z}^{(\chi)}_{(f,g)}(b)\itwiner{V}_{\partial_1 b}\,\coc{z}^{(\chi)*}_{(f,g)}(b) &= \sum_k c_k \coc{z}^{(\chi)}_{(f,g)}(b)\,\iota(W(r_k,s_k))\,\coc{z}^{(\chi)}_{(f,g)}(b)^{*}=\\ 
&= \sum_k c_k\,\iota(W(r_k,s_k))\exp\{\imath\varphi_k\}\\
\end{split}
\end{equation}
for some reals $\varphi_k$. The final series converges in the strong operatorial topology, too.
Since $\iota(W(r_k,s_k)) \in \snet{R}{\partial_1 b}$ for hypotheses, 
$c_k\iota(W(r_k,s_k))\exp\{\imath\varphi_k\}\in \snet{R}{\partial_1 b}$ for every $k$, and thus we also have
$ \coc{z}^{(\chi)}_{(f,g)}(b)\itwiner{V}_{\partial_1 b}\,\coc{z}^{(\chi)*}_{(f,g)}(b)\in\snet{R}{\partial_1 b}$
since $\snet{R}{\partial_1 b}$ is closed with respect to the strong operatorial topology. As a consequence, it must hold
$$\itwiner{V}_{\partial_0 b}  = \coc{z}^{(\chi)}_{(f,g)}(b)\itwiner{V}_{\partial_1 b}\,\coc{z}^{(\chi)*}_{(f,g)}(b) \in \snet{R}{\partial_1 b}.$$
Since $\partial_0b, \partial_1 b \in \simplex{0}{\poset{R}}$ and $\itwiner{V}_{\partial_0 b}$ are generic, we have found that 
$\itwiner{V}_a \in \bigcap_{\open{I} \in \poset{R}} \snet{R}{\open{I}} = \bC\cdot 1$ in view of the irreducibility property of the net.\\
\textbf{(b)} Let us establish the equivalence of cocycles associated to different maps $\chi$ but with the same $(f,g)$.
 If $\chi_1$ and $\chi_2$ are defined as in (\ref{eq:assignchi}), for every $a\in \simplex{0}{\poset{R}}$ the function
 $\Delta\chi^{(a)} := \chi_1^{(a)}-\chi_2^{(a)}$ is smooth and compactly supported in the open set  $a$, so that it can uniquely and 
 smoothly be extended to a smooth function over $\bS^1$ supported in $a$. As usual,  we indicate by $\Delta\chi^{(a)}$ this unique extension.
Now, define the field of unitaries $\itwiner{V}: \simplex{0}{\poset{R}} \ni a \mapsto \itwiner{V}_a :=  Z\left( \Delta \chi^{(a)} f, \Delta \chi^{(a)} g\right) \in \snet{R}{a}$.
For every $b\in \simplex{1}{\poset{R}}$, we get
\begin{equation}
\begin{split}
\itwiner{V}_{\partial_0b} \coc{z}^{(\chi_1)}_{(f,g)}(b) &=  Z\left( \Delta \chi^{(\partial_0 b)} f, \Delta \chi^{(\partial_0 b)} g\right)
 Z\left((\chi_1^{(\partial_1 b)}- \chi_1^{(\partial_0 b)})f, (\chi_1^{(\partial_1 b)}- \chi_1^{(\partial_0 b)}) g\right)\\
  &= Z\left((\chi_1^{(\partial_1 b)}- \chi_1^{(\partial_0 b)} + \chi_1^{(\partial_0 b)} - \chi_2^{(\partial_0 b)})f, 
(\chi_1^{(\partial_1 b)}- \chi_1^{(\partial_0 b)} + \chi_1^{(\partial_0 b)} - \chi_2^{(\partial_0 b)}) g\right)\\
&= Z\left((\chi_1^{(\partial_1 b)}- \chi_2^{(\partial_0 b)})f,(\chi_1^{(\partial_1 b)} - \chi_2^{(\partial_0 b)}) g\right)\\ 
\end{split}
\end{equation}
where, passing from the first to the second line, we have omitted a phase arising from Weyl relations, since it vanishes as before.
With an analogous computation we similarly find:
$$ \coc{z}^{(\chi_2)}_{(f,g)}(b) \itwiner{V}_{\partial_1b} = Z\left((\chi_1^{(\partial_1 b)}- \chi_2^{(\partial_0 b)})f, 
(\chi_1^{(\partial_1 b)} - \chi_2^{(\partial_0 b)}) g\right)$$
so that $\itwiner{V}_{\partial_0b} \coc{z}^{(\chi_1)}_{(f,g)}(b) =  \coc{z}^{(\chi_2)}_{(f,g)}(b) \itwiner{V}_{\partial_1b}$ as wanted. \\
\textbf{(c)} Consider a positively-oriented $1$-symplex $b$ with $\overline{\partial_0 b}\cap\overline{\partial_1 b}=\emptyset$; let us indicate by $\open{I}_b\in \simplex{0}{\poset{R}}$ the open proper 
segment lying between $\partial_0 b$ and  $\partial_1 b$.
If  $\coc{z}^{(\chi)}_{(f,g)}$ and $\coc{z}^{(\chi)}_{(f',g')}$ are equivalent, we may write
$\itwiner{V}_{\partial_0b}  =  \coc{z}^{(\chi)}_{(f',g')}(b)\itwiner{V}_{\partial_1b} \coc{z}^{(\chi)*}_{(f,g)}(b),$
 for some unitaries $\itwiner{V}_{\partial_j b} \in \snet{R}{\partial_jb}$. Therefore
 $\itwiner{V}_{\partial_0b}  =  \coc{z}^{(\chi)}_{(f',g')}(b)\coc{z}^{(\chi)*}_{(f,g)}(b)\widetilde{\itwiner{V}_{\partial_1b}}\:,$
 and thus $\itwiner{V}_{\partial_0b} \widetilde{\itwiner{V}_{\partial_1b}}^*  =  \coc{z}^{(\chi)}_{(f',g')}(b)\coc{z}^{(\chi)*}_{(f,g)}(b)$
where we have introduced the unitary operator $\widetilde{\itwiner{V}_{\partial_1b}}:= \coc{z}^{(\chi)}_{(f,g)}(b)
 \itwiner{V}_{\partial_1b}\,\coc{z}^{(\chi)*}_{(f,g)}(b)$. As ${\itwiner{V}_{\partial_1b}} \in \snet{R}{\partial_1b}$, 
 following an argument as that in the proof 
 of (a), we achieve $\widetilde{\itwiner{V}_{\partial_1b}} \in  \snet{R}{\partial_1b}$ and so
 $\widetilde{\itwiner{V}_{\partial_1b}}^{*} \in  \snet{R}{\partial_1b}$. 
 The term $\coc{z}^{(\chi)}_{(f',g')}(b)  \coc{z}^{(\chi)*}_{(f,g)}(b)$ can be computed and, in view of Weyl relations, it finally arises
\begin{equation}\label{eq:centralN} 
\itwiner{V}_{\partial_0b} \widetilde{\itwiner{V}_{\partial_1b}}^*  =  \coc{z}^{(\chi)}_{(f'-f,g'-g)}(b) e^{\imath\varphi},
\end{equation}
where $\varphi \in \bR$ depends on $f,f',g,g'$ and $\chi$.
Now consider two real smooth functions $r,s$ supported in $\open{I}_b$. The Weyl generator $Z(r,s)$ belongs to $\snet{R}{\open{I}_b}$ and thus it commutes with
both $\widetilde{\itwiner{V}_{\partial_1b}}^*$ and $\itwiner{V}_{\partial_0b}$ so that (\ref{eq:centralN}) produces (notice that $\chi^{(b)}=1$
on $\open{I}_b$)
\begin{eqnarray}
\itwiner{V}_{\partial_0b} \widetilde{\itwiner{V}_{\partial_1b}}^*  &=& Z(r,s) \coc{z}^{(\chi)}_{(f'-f,g'-g)}  Z(r,s)^* e^{\imath\varphi}=\nonumber\\
&=& \coc{z}^{(\chi)}_{(f'-f,g'-g)} e^{\imath\varphi} \exp\left\{\imath \int_{\bS_1} ((f'-f)s - (g'-g)r) d\theta \right\}.\nonumber 
\end{eqnarray}
Comparing with  (\ref{eq:centralN}) we conclude that it must be 
$$\exp\left\{\imath \int_{\bS_1} ((f'-f)s - (g'-g)r) d\theta \right\}  = 1\:.$$
Since $r,s$ are arbitrary we have that $f'-f= 0$, $g'- g =0$ on $\open{I}_b$. Since the procedure can be implemented fixing 
$\open{I}_b$ as a sufficiently small neighborhood of every point on $\bS^1$, we conclude that $f=f'$ and $g=g'$ everywhere on $\bS^1$.
\textbf{(d)} Referring to theorem \ref{th:vacuum} one finds 
\begin{equation}
\begin{split}
U_{(r,0)} \coc{z}^{(\chi)}_{(f,g)}(b) U_{(r,0)}^* &= e^{-\imath r P^\otimes} \coc{z}^{(\chi)}_{(f,g)}(b) e^{\imath r P^\otimes}  =\\
&=W\left[e^{-\imath r P} 2^{-1/2} (A^{1/4} (\chi^{(b)}  f) + \imath A^{-1/4} (\chi^{(b)} g))\right]=\\
 &= W\left[ 2^{-1/2} A^{1/4} e^{-\imath r P}   (\chi^{(b)}  f)
+ \imath\,2^{-1/2} A^{-1/4} e^{-\imath r P}  (\chi^{(b)}  g)\right]\\
\end{split}
\end{equation}
where we used the fact that $e^{-\imath r P}$ and the spectral measure of $A$ commute as one can easily prove
exploiting the definition of $A$ and the fact that $P$ is the unique self-adjoint extension of the essentially self-adjoint operator $-\imath \frac{d}{d\theta} : 
\smoothfuncv{\bS^1}{\bC}\to \smoothfuncv{\bS^1}{\bC}$.  On the other hand, by direct 
inspection and working in Fourier representation, one finds that $(e^{-\imath r P} h)(\theta) = h(\theta -r) =: (\beta^*_r(h))(\theta)$
for every $h\in \sqint{\bS^1}{\theta}$. This, joined to (\ref{eq:covset}), implies immediately the validity 
of the thesis.
\end{proof}

\subsection{Topological representations of $\stdhomgroup{\bS^1}$} 

Int this section we state and prove some properties of the unitary representations of $\stdhomgroup{\bS^1}$ associated to the found cocycles.
 
\begin{theorem}  \label{th:top_cocycles_repr}
Consider the $1$-cocycle $\coc{z}^{(\chi)}_{(f,g)}$ defined in theorem \ref{th:propcocycle} and the associated 
 representation (\ref{eq:cocycle_repr_spec}) of $\stdhomgroup{\bS^1}\equiv \bZ$. In this case the representation reads
 \begin{equation}\label{eq:repzf}
 \pi_{(f,g)} : \bZ \ni n \mapsto Z(nf,ng),
 \end{equation}
 where it is manifest that the representation does not depend on the choice of $\chi$; moreover, it enjoys the following properties:\\ 
{\bf (a)} 
$\pi_{(f,g)}$ is trivial
-- equivalently $\coc{z}^{(\chi)}_{(f,g)}$ is trivial -- if and only if $f= g = 0$.\\
{\bf (b)}  For every pair  $(f,g), (f',g') \in\Cdata{\sdc{C}}\times\Cdata{\sdc{C}}$ with $(f,g) \neq (0,0) \neq (f',g')$,
 the unitary representations $\pi_{(f,g)}$ and $\pi_{(f',g')}$ are unitarily
equivalent.\\
{\bf (c)} For every $(f,g) \in \Cdata{\sdc{C}}$,  every $r\in\bR$ and every $n\in \bZ$, it holds
$$U_{(r,0)} \pi_{(f,g)}(n) U_{(r,0)}^* = \pi_{(\beta^*_r (f) ,\beta^*_r(g))} (n)\:,$$
where $U_{(r,0)}$ is the one-parameter unitary group implementing $\theta$-displacements $\beta_r$ and leaving the vacuum invariant, and $\beta_r^*$ is the pull-back action of $\theta$-displacements on functions defined on $\bS^1$. \\
{\bf (d)} If $(0,0) \neq (f,g)\in \Cdata{\sdc{C}}$, the space $\varhs{H}_0 = \fock{\hs{H}_0}$ decomposes as a countably infinite Hilbert sum of closed pairwise orthogonal 
subspaces $\varhs{H}_0 = \bigoplus_{k=0}^{+\infty} \varhs{H}^{(f,g)}_k$ such that the following holds for $k\in \bN$.

(i) $\varhs{H}^{(f,g)}_k$ is invariant under $\pi_{(f,g)}$.

(ii) There is a unitary map $U^{(f,g)}_k : \varhs{H}^{(f,g)}_k \to\sqint{\bR}{x}$ such that $\pi_{(f,g)} \rest_{\varhs{H}_k}$
admits a direct integral  decomposition into one-dimensional irreducible representations $\lambda_x$ of $\bZ$ (\ref{eq:lambdax})
as
\begin{equation}\label{eq:dirdec} 
U^{(f,g)*}_k \pi_{(f,g)} \rest_{\varhs{H}_k} U^{(f,g)}_k = \int_{\bR}^{\oplus} dx\:\: \lambda_x , 
\end{equation}
where $\sqint{\bR}{x}= \int_{\bR}^{\oplus} dx\:\:\hs{H}_x $, with $\hs{H}_x := \bC$ and $dx$ being the Lebesgue measure on $\bR$. 
\end{theorem}
\begin{proof}
Let us first prove (\ref{eq:repzf}) for $n=1$. Since we know that (\ref{eq:cocycle_repr_spec}) gives rise to a group representation of $\pi_1(\bS^1)$
when $\coc{z}^{(\chi)}_{(f,g)}$ is a cocycle, to  prove (\ref{eq:repzf}) for  $n=1 \in \bZ = \stdhomgroup{\bS^1}$
 i.e.,  
\begin{equation}\label{eq:z1}
\coc{z}^{(\chi)}_{(f,g)}(p) = Z(f,g)\:, \quad \mbox{for $p\in 1$} 
\end{equation}
 it is enough to prove it for a fixed path  $p\in 1$, because the result depends only on the equivalence class. 
 To this end, if $\bS^1 = [-\pi,\pi]$ where $-\pi\equiv \pi$, consider the path $p\in 1$
made of the simplices $b$, with
$|b|:= (-\frac{\pi}{2}-\epsilon, \frac{\pi}{2} + \epsilon)$, $\partial_1 b := (-\frac{\pi}{2}-\epsilon, -\frac{\pi}{2} + \epsilon)$, 
$\partial_0 b := (\frac{\pi}{2}-\epsilon, \frac{\pi}{2} + \epsilon)$,
and  $b'$ with $|b'|:= (\frac{\pi}{2}-\epsilon, \pi] \cup [-\pi, -\frac{\pi}{2} + \epsilon)$,
$\partial_1 b' := (\frac{\pi}{2}-\epsilon, \frac{\pi}{2} + \epsilon)$, $\partial_0 b' := (-\frac{\pi}{2}-\epsilon, -\frac{\pi}{2} + \epsilon)$, where $\epsilon>0$ is so small that $\partial_0 b \cap \partial_1 b =\emptyset$. Using the definition of $\chi^{(b)}$
and $\chi^{(b')}$, it follows immediately that $\chi^{(b)} + \chi^{(b')} =1$ everywhere on $\bS^1$. 
Therefore we have that $\coc{z}^{(\chi)}_{(f,g)}(p)$ equals 
\begin{equation*}
\begin{split}
&Z(\chi^{(b')}f,\chi^{(b')}g)Z(\chi^{(b)}f,\chi^{(b)}g) = \\
&=Z((\chi^{(b')}+ \chi^{(b)})f,(\chi^{(b')}+ \chi^{(b)})g)\,e^{\imath\int_{\bS^1}\chi^{(b')}\chi^{(b)}(fg-gf) d\theta} = Z(f,g)\\
\end{split}
\end{equation*}
We have established (\ref{eq:z1}), i.e. (\ref{eq:cocycle_repr_spec}) for $n=1$. Let us generalize the result for $n\in \bZ$.
By the definition of $Z$ and making  use of Weyl commutation 
relations one gets  
\begin{equation}\label{eq:abelrel}
Z(nf,ng)Z(mf,mg) = Z((n+m)f, (n+m)g),\quad\forall\, n,m\in \bZ.
\end{equation}
Using the fact that $\pi_{\coc{z}^{(\chi)}_{(f,g)}}$ as defined in (\ref{eq:cocycle_repr_spec}) is a group representation
of  $\stdhomgroup{\bS^1}=\bZ$, which is Abelian and generated by $1$, one has that (\ref{eq:abelrel}) and (\ref{eq:z1}) together  yield eq. (\ref{eq:cocycle_repr_spec}) in the general case.\\
 Let us pass to prove (a).
 As a consequence of  (\ref{eq:repzf}), it is clear that this representation is trivial, that is $\coc{z}^{(\chi)}_{(f,g)}$ is such due to theorem \ref{eq:cocycle_repr_spec},
 if and only if $Z(f,g)=1$. It is equivalent to say $W\left[2^{-1/2}(A^{1/4}f+ \imath A^{-1/4} g)\right] =1$. By theorem \ref{prop:hyperbolic_generators}
 we know that
 $$\ip{\Psi_0}{W\left[2^{-1/2}\left(A^{1/4} f + \imath A^{-1/4} g\right)\right] \Psi_0} = e^{-\frac{1}{4}\left(\ip{f}{A^{1/2} f} + \ip{g}{A^{-1/2} g}\right)},$$
for all $f,g \in\testfuncv{\bS^1}{\bR}$. Since $\norm{\Psi_0}=1$ we have finally that $Z(f,g)=1$ entails  $\ip{f}{A^{1/2} f} + \ip{g}{A^{-1/2} g}=0$ and so $f,g =0$ because $A^{-1/4}$ and $A^{1/4}$
 are strictly positive. We have found that triviality of $\pi_{(f,g)}$ implies $f,g=0$. The converse is obvious and so
  the proof of (a) is concluded. \\Let us demonstrate (b).
 Assume that
$$\norm{2^{-1/2}(A^{1/4}f + \imath A^{-1/4}g)} = \norm{2^{-1/2}(A^{1/4}f' + \imath A^{-1/4}g')} = a \neq 0$$ 
(the case ``$=0$'' being obvious). Defining
 $\psi_1 := 2^{-1/2}(A^{1/4}f + \imath A^{-1/4}g)$
 we can complete this vector to a maximal orthogonal system $\{\psi_n\}_{n\in \bN}$ of $\sqint{\bS^1}{\theta}$
where $\norm{\psi_n}=a$ for every $n\in \bN$. Similarly, defining $\phi_1 := 2^{-1/2}(A^{1/4}f' + \imath A^{-1/4}g')$
 we can complete this vector to a maximal orthogonal system $\{\phi_n\}_{n\in \bN}$ of $\sqint{\bS^1}{\theta}$,
where $\norm{\psi_n}=a$ for every $n\in \bN$.  There is a unique unitary operator $U: \sqint{\bS^1}{\theta} \to \sqint{\bS^1}{\theta}$
completely individuated by imposing the conditions $U\psi_n = \phi_n$ for every $n\in \bN$.
It is a known property of Weyl generators $W[\psi] = e^{\overline{a(\psi)-a^*(\psi)}}$ that
$$V_\otimes W[\psi] V_\otimes^* = W[V\psi]$$
where the unitary operator $V_\otimes$ in the Fock space is defined by tensorialization of the unitary operator $V$
in the one-particle
space, with the requirement that $V_\otimes$ reduces to the identity acting on the vacuum vector.
As a consequence $U_\otimes W[\psi_1] U_\otimes^* = W[U\psi_1] = W[\phi_1]$ or, equivalently, $U_\otimes Z(f,g) U^*_\otimes
= Z(f',g')$ and thus $U_\otimes Z(nf,ng) U^*_\otimes
= Z(nf',ng')$, making use of (\ref{eq:abelrel}); we have found that $\pi_{(f,g)}$ and $\pi_{(f',g')}$ are unitarily equivalent.
 Let us pass to the the case
 $0 \neq \norm{2^{-1/2}(A^{1/4}f + \imath A^{-1/4}g)} \neq \norm{2^{-1/2}(A^{1/4}f' + \imath A^{-1/4}g')}  \neq 0$
 and define the real number
$$r := \norm{2^{-1/2}(A^{1/4}f + \imath A^{-1/4}g)}/\norm{2^{-1/2}(A^{1/4}f' + \imath A^{-1/4}g')}.$$
With the procedure used in the former case one achieves the existence of a unitary operator $V$ on the Fock space such that
$$W\left[2^{-1/2}(A^{1/4}f + \imath A^{-1/4}g)\right] = VW\left[r2^{-1/2}(A^{1/4}f' + \imath A^{-1/4}g')\right]V^*\:.$$
To conclude it is sufficient to establish the existence of a second unitary operator $E$ (depending on $f$, $g$ and $r$)
with
$$W\left[2^{-1/2}(A^{1/4}f' + \imath A^{-1/4}g')\right] = EW\left[r \: 2^{-1/2}(A^{1/4}f' + \imath A^{-1/4}g')\right]E^*\:.$$
This is an immediate consequence of the following result proved in appendix \ref{app:proofs}. 
\begin{lemma} \label{lem:lemmaE} 
Let $\hs{H}$ be a complex Hilbert space with associated bosonic Fock space $\fock{\hs{H}}$. Define the unitary Weyl
generators $W[\psi]$ as in (\ref{eq:Weylgen}) for every $\psi \in \hs{H}$. For every fixed $\psi\in \hs{H}$ with $\norm{\psi}=1$
there is a strongly continuous one-parameter group of unitary operators $\{E^{(\psi)}_\lambda\}_{\lambda\in \bR}$ such that
\begin{equation}\label{eq:exp}
E^{(\psi)}_\lambda W[\psi] E^{(\psi)*}_\lambda = W\left[e^\lambda\psi\right]\:, \quad \mbox{for all $\lambda \in \bR$.}
\end{equation}
\end{lemma}
\noindent The proof of (c) follows immediately from (d) in theorem \ref{th:propcocycle} taking the independence from $\chi$
into account.\\
Finally we prove (d). Fix $(f,g)\in \Cdata{\sdc{C}}$. In view of Weyl commutation relations, the unitary operators
$$U(a,b):= Z\left(\frac{(a+\imath b)f}{\norm{A^{1/4} f + \imath A^{-1/4}g}}, \frac{(a+ib)g}{\norm{A^{1/4} f + \imath A^{-1/4}g}}\right)\:,\quad  (a,b)\in \bR^2\:,$$ fulfill the one-dimensional Weyl  relations
$$U(a,b)U(a',b')= U(a+a',b+b') e^{-\imath (ab'-a'b)/2}\:,\quad U(a,b)^* = U(-a,-b)\:.$$
From the well-known uniqueness result, due to Stone, von Neumann and Mackey \cite{Strocchi_05}, the space $\varhs{H}_0$ decomposes into a direct sum of 
pairwise orthogonal closed subspaces $\varhs{H}_k$ where each $\varhs{H}_k$ is unitarily equivalent to $\sqint{\bR}{x}$ and the relevant 
unitary map satisfies $$U^{(f,g)*}_k U(a,b) \rest_{\varhs{H}_k} U^{(f,g)}_k = \exp{\imath\{\overline{aX+bP}\}}\:,$$ $X,P$ being the standard position and impulse
operators on the real line ($aX+bP$ is defined on the core given by the Schwartz space). As a consequence
$$U^{(f,g)*}_k \pi_{(f,g)}(n) \rest_{\varhs{H}_k} U^{(f,g)}_k = U^{(f,g)*}_k Z(nf,ng) U^{(f,g)}_k = e^{\imath n c X}\:, $$
with $c= \norm{2^{-1/2} (A^{1/4} f + \imath A^{-1/4} g)}>0$.
Then the spectral decomposition of $cX$ gives rise to (\ref{eq:dirdec}) immediately. To end the proof the only thing remained to show  is that the number of spaces $\varhs{H}_k$ is infinite. Since $\varhs{H}_0$ is separable that infinity must be countable.
It is known by the general theory of Weyl algebras on finite-dimensional symplectic spaces that the $\varhs{H}_k$'s can be obtained as follows. Using weak operator topology, define the operator
$$P=\frac{1}{2\pi}\int_{\bR^2} e^{- (u^2 +v^2)/4} U(u,v)\: dudv$$ which turns out to be a nonvanishing orthogonal projector.
If $\{\phi_k\}_{k\in G}$ is a Hilbert basis for the subspace $P(\varhs{H}_0)$, for any fixed $k\in G$, $\varhs{H}_k$ is the closed space generated by all of $U(a,b)\phi_k$ as $a,b \in \bR$. To conclude it is sufficient to prove that $G$ must be infinite. To this end consider a Hilbert basis in $\hs{H}_0$, 
$\psi_1= (A^{1/4} f +\imath A^{-1/4} g)/\norm{A^{1/4} f + \imath A^{-1/4} g}$, $\psi_2$, $\psi_3$, $\ldots$ and an associated orthonormal (not necessarily complete) system in $\varhs{H}_0$:
$\Psi_1:=\Psi_0$ (the vacuum), $\Psi_2 := a^*(\psi_2)\Psi_0$, $\Psi_3 := a^*(\psi_3)\Psi_0$, $\ldots$.
By construction, one can verify that
\begin{equation}
\begin{split}
\ip{P\Psi_h}{P\Psi_k} &= \frac{1}{2\pi}\int_{\bR^2} e^{- (u^2 +v^2)/4}\ip{\Psi_h}{U(u,v)\Psi_k}\: dudv=\\
&= \frac{\delta_{hk}}{2\pi}\int_{\bR^2} e^{- (u^2 +v^2)/4} e^{- (u^2 +v^2)/4}\: dudv.\\ 
\end{split}
\end{equation}
Therefore, up to normalization, $P \Psi_1, P\Psi_2, \ldots \in P(\varhs{H}_0)$ is an infinite orthonormal system in $P(\varhs{H}_0)$. This means that
$P(\varhs{H}_0)$ admits an infinite Hilbert basis.
\end{proof}

\subsection{Topological superselection sectors}
 In this section we show how it is possible to associate each localized cocycle $\coc{z}^{(\chi)}_{(f,g)}$
  with a sharp excitation of the reference vacuum representation $\{\iota, \bI\}$ in such a way that unitarily inequivalent pair of cocycles individuates unitarily inequivalent pair of unitary net representations fulfilling the selection criterion, and thus two different superselection sectors. The idea is similar to that exploited to define a relevant functor in the proof of theorem 4.3 in \cite{Brunetti_Ruzzi_08}, although there are two important differences. First of all, here we are dealing with a proper subset of 
cocycles and not with the whole category $\cocyclecat{\net{R}{R}}$. Secondly, as we shall see in details shortly, the map which associates 
cocycles to net representations in the proof of Theorem 4.3 in \cite{Brunetti_Ruzzi_08} does not turn out to be well posed in our lower dimensional model and has to be modified. 
Consider $\coc{z}^{(\chi)}_{(f,g)} \in \cocyclecat{\net{R}{R}}$ and, for $\open{I}, \widetilde{\open{I}} \in \poset{R}$ with $\widetilde{\open{I}} \subseteq \open{I}$ define
\begin{eqnarray}
\pi^{\coc{z}^{(\chi)}_{(f,g)}}_\open{I}(A) &:=& \coc{z}^{(\chi)}_{(f,g)}(b_\open{I}) \iota_\open{I}(A) \coc{z}^{(\chi)}_{(f,g)}(b_\open{I})^*\:, \quad A \in \sweyl{\open{I}}\label{eq:1pi}\\
\psi^{\coc{z}^{(\chi)}_{(f,g)}}_{\open{I},\tilde{\open{I}}} &: =& \coc{z}^{(\chi)}_{(f,g)}\left(\open{I},\widetilde{\open{I}}\right).\label{eq:1psi}
\end{eqnarray}
Above, $b_\open{I}$ is a $1$-simplex with final endpoint $\partial_0 b_\open{I} := \open{I}$ and initial endpoint $\partial_1 b_\open{I} := \open{J}$ where $\open{J}\perp\open{I}$ and, finally,  $b_\open{I}$ is positively-oriented; the $1$-simplex  $(\open{I},\widetilde{\open{I}})$
is that with $\partial_1\left(\open{I},\widetilde{\open{I}}\right) = \widetilde{\open{I}}$ and  $\partial_0\left(\open{I},\widetilde{\open{I}}\right) =\open{I} =
\left|\left(\open{I},\widetilde{\open{I}}\right)\right|$.\\
Finally define
\begin{equation}\label{eq:1pipsi}
 \pi^{\coc{z}^{(\chi)}_{(f,g)}} : \poset{R} \ni \open{I} \mapsto \pi^{\coc{z}^{(\chi)}_{(f,g)}}_\open{I} \:, \quad\mbox{and}
\quad \psi^{\coc{z}^{(\chi)}_{(f,g)}} : \poset{R}\times\poset{R} \ni (\open{I},\widetilde{\open{I}}) \mapsto \psi^{\coc{z}^{(\chi)}_{(f,g)}}_{\open{I},\tilde{\open{I}}}\quad 
\end{equation}
for $\open{I},\widetilde{\open{I}} \in \poset{R}$ and 
$\widetilde{\open{I}} \subseteq \open{I}$. We claim that $\{\pi^{\coc{z}^{(\chi)}_{(f,g)}}, \psi^{\coc{z}^{(\chi)}_{(f,g)}}\}$ is a net representation which satifies the topological selection criterion.\\ 

\noindent\textbf{Remark} The definitions given in (\ref{eq:1pi}) and (\ref{eq:1psi}) are the same as that used in theorem 4.3 
in \cite{Brunetti_Ruzzi_08} with the only difference that $b_\open{I}$ is now a (positive oriented) $1$-simplex  rather than a path. This is due to the fact that, if we adopted the very definition as in \cite{Brunetti_Ruzzi_08}
in our case, the defined objects would depend on the chosen path, differently from the higher dimensional case; we shall come back to this issue later.

\begin{theorem}
 If $\coc{z}^{(\chi)}_{(f,g)}\in \cocyclecat{\net{R}{R}}$, the pair $\{\pi^{\coc{z}^{(\chi)}_{(f,g)}}, \psi^{\coc{z}^{(\chi)}_{(f,g)}}\}$ defined as in (\ref{eq:1pipsi})
  is a unitary  net representation of $\alg{W}_{KG}$ over $\varhs{H}_0$, which is independent from the choice of the simplices $b_\open{I}$ made in (\ref{eq:1pi}). Moreover, this net representation enjoys the following properties.\\
{\bf (a)} $\{\pi^{\coc{z}^{(\chi)}_{(f,g)}}, \psi^{\coc{z}^{(\chi)}_{(f,g)}}\}$ is irreducible and satisfies the selection criterion, thus defines a 
sharp excitation of the reference vacuum net representation $\{\iota, \bI\}$, giving rise to a superselection sector.\\
{\bf (b)} $\{\pi^{\coc{z}^{(\chi)}_{(f,g)}}, \psi^{\coc{z}^{(\chi)}_{(f,g)}}\}$ and 
$\{\pi^{\coc{z}^{(\chi')}_{(f',g')}}, \psi^{\coc{z}^{(\chi')}_{(f',g')}}\}$ belong to the same superselection sector 
(i.e. they are unitarily equivalent) if and only if  $f=f'$ and $g=g'$. \\
{\bf (c)} The $1$-cocycle associated with the net representation $\{\pi^{\coc{z}^{(\chi)}_{(f,g)}}, \psi^{\coc{z}^{(\chi)}_{(f,g)}}\}$ 
as in (\ref{eq:zetapi}) coincides with $\coc{z}^{(\chi)}_{(f,g)}$ itself.
\end{theorem}
\begin{proof}
 First of all we have to show that conditions (\ref{eq:UNR}) are fulfilled. By direct inspection, exploiting the definition
of $\coc{z}^{(\chi)}_{(f,g)}$, we find
\begin{equation}\label{eq:zz}
 \psi^{\coc{z}^{(\chi)}_{(f,g)}}_{\open{I},\tilde{\open{I}}} =
  Z\left((\chi^{(\tilde{\open{I}})}- \chi^{(\open{I})})f,(\chi^{(\tilde{\open{I}})}- \chi^{(\open{I})})g \right) ,
\end{equation}
where as before the functions $\chi^{(\tilde{\open{I}})}$ and $\chi^{(\tilde{\open{I}})}- \chi^{(\open{I})}$ has been extended to the whole circle.  With this definition the second identity in (\ref{eq:UNR}) straightforwardly arises from  (\ref{eq:zz}) and the Weyl relations.
Let us examine the first identity in (\ref{eq:UNR}).
By linearity and continuity, it suffices to consider only the case where $A= W(\Phi,\Pi)$ is a local Weyl generator, with $\Phi,\Pi$ supported in ${\tilde{\open{I}}}$. Remind that in our case $j_{\open{I}\tilde{\open{I}}}$ can be omitted, interpreting the elements of local Weyl algebras as elements of the global Weyl algebra $\alg{W}_{KG}$.
 By direct inspection, keeping in mind that $\iota\left(W(\Phi,\Pi)\right) = Z\left(\Phi,\Pi\right)$, 
  making use of Weyl relations and applying the definition of $\coc{z}^{(\chi)}_{(f,g)}$,
 one finds that if $\Phi,\Pi$ are supported in ${\open{I}}$
 \begin{equation}\label{eq:pp} 
\pi^{\coc{z}^{(\chi)}_{(f,g)}}_{\open{I}}\left( W(\Phi,\Pi)\right)= Z(\Phi,\Pi)\, 
 e^{\left\{\imath\sigma\left((\Phi,\Pi),((1-\chi^{(\open{I})}) f,(1-\chi^{(\open{I})}) g)\right)\right\}}.
\end{equation}
 Notice that only $\open{I}$ appears in the right-hand side, so different choices of $b_{\open{I}}$ yields 
the same result and the choice of $b_{\open{I}}$ is immaterial. Equations (\ref{eq:pp}) and (\ref{eq:zz}) entail, in view of Weyl identities
\begin{equation}
\begin{split}
&\psi^{\coc{z}^{(\chi)}_{(f,g)}}_{\open{I},\tilde{\open{I}}} \pi^{\coc{z}^{(\chi)}_{(f,g)}}_{\tilde{\open{I}}}\left( W(\Phi,\Pi)\right)
\psi^{\coc{z}^{(\chi)}_{(f,g)}*}_{\open{I},\tilde{\open{I}}} = \\
&=Z(\Phi,\Pi)\,e^{\imath\sigma\left((\Phi,\Pi),((1-\chi^{(\tilde{\open{I}})}) f,(1-\chi^{(\tilde{\open{I}})}) g)\right)}\,  e^{\imath\sigma\left((\Phi,\Pi),((\chi^{(\tilde{\open{I}})}-\chi^{(\open{I})}) f,(\chi^{(\tilde{\open{I}})}-\chi^{(\open{I})}) g)\right)}
 =\\ 
&=Z(\Phi,\Pi)\,e^{\imath\sigma\left((\Phi,\Pi),((1-\chi^{(\open{I})}) f,(1-\chi^{(\open{I})}) g)\right)} =\pi^{\coc{z}^{(\chi)}_{(f,g)}}_{\open{I}} (W(\Phi,\Pi))\\ 
\end{split}
\end{equation}
This is nothing but the identity (\ref{eq:UNR}). Let us prove the remaining statements of the theorem.\\
\textbf{(a)} If $\open{O} \in \poset{R}$ let  $\open{N} \subset \bS^1$ a (connected) simply connected open set (so that either $\open{N}\in \poset{R}$
or $\open{N}= \bS^1 \setminus \{p\}$ for some $p\in \bS^1$) with $\overline{\open{O}} \subset \open{N}$.
Fix $\open{I}\in \poset{R}$ with both $\overline{\open{I}} \subset \open{N}$ and $\open{I}\perp\open{O}$. We can define $W_\open{I}^{\open{N}\open{O}}$ as
\begin{equation}\label{eq:WINO}
W_\open{I}^{\open{N}\open{O}} := \coc{z}^{(\chi)}_{(f,g)}(b_\open{I})^*,
\end{equation}
where $b_\open{I}\in \simplex{1}{\poset{R}}$ is chosen as in (\ref{eq:1pi})
but  $|b_\open{I}| \subset \open{N}$. Given this definition the three conditions coming from (\ref{eq:SS}) turn out to be satisfied.
The first requirement is verified automatically in view of  (\ref{eq:1pi}), the remaining two have straightforward proofs
based on  Weyl relations and proceeding as above. 
The proof of the irreducibility of $\{\pi^{\coc{z}^{(\chi)}_{(f,g)}}, \psi^{\coc{z}^{(\chi)}_{(f,g)}}\}$ will be postponed at the end the 
of the proof of (b).\\
\textbf{(b)} In view of (b) and (c) in theorem \ref{th:propcocycle}, the thesis  is equivalent to say that 
$\{\pi^{\coc{z}^{(\chi)}_{(f,g)}}, \psi^{\coc{z}^{(\chi)}_{(f,g)}}\}$ and $\{\pi^{\coc{z}^{(\chi')}_{(f',g')}}, \psi^{\coc{z}^{(\chi')}_{(f',g')}}\}$
are unitarily equivalent if and only if $\coc{z}^{(\chi)}_{(f,g)}$ and $\coc{z}^{(\chi')}_{(f',g')}$ are unitarily equivalent. Suppose that $\itwiner{T} \in (\coc{z}^{(\chi)}_{(f,g)}, \coc{z}^{(\chi')}_{(f',g')})$ is unitary; as a consequence
$\itwiner{T} \in (\{\pi^{\coc{z}^{(\chi)}_{(f,g)}}, \psi^{\coc{z}^{(\chi)}_{(f,g)}}\},\{\pi^{\coc{z}^{(\chi')}_{(f',g')}}, \psi^{\coc{z}^{(\chi')}_{(f',g')}}\})$.
Indeed take $A\in\sweyl{\open{I}}$ and remind that $\itwiner{T}_{\partial_1b_\open{I}}\in \snet{R}{\partial_1b_\open{I}}$ and thus 
$\itwiner{T}_{\partial_1b_\open{I}}$ and $\itwiner{T}^*_{\partial_1b_\open{I}}$
 commute with $\iota(A)$ because $\partial_1b_\open{I} \perp \open{I}$. Hence
\begin{equation}
\begin{split}
\itwiner{T}_\open{I} \pi^{\coc{z}^{(\chi)}_{(f,g)}}(A) \itwiner{T}^*_\open{I} &= \itwiner{T}_\open{I} \coc{z}^{(\chi)}_{(f,g)}(b_\open{I}) \iota(A)
\coc{z}^{(\chi)*}_{(f,g)}(b_\open{I}) \itwiner{T}^*_\open{I} =\\  
&=\coc{z}^{(\chi')}_{(f',g')}(b_\open{I}) \itwiner{T}_{\partial_1 b_\open{I}}\iota(A)
(\itwiner{T}_\open{I} \coc{z}^{(\chi)}_{(f,g)}(b_\open{I}))^* = \\
&=\coc{z}^{(\chi')}_{(f',g')}(b_\open{I}) \iota(A) \itwiner{T}_{\partial_1 b_{\open{I}}} (\coc{z}^{(\chi')}_{(f',g')}(b_\open{I}) \itwiner{T}_{\partial_1 b_{\open{I}}})^*=\\
 &= \coc{z}^{(\chi')}_{(f',g')}(b_\open{I}) \iota(A)
 \itwiner{T}_{\partial_1 b_{\open{I}}} \itwiner{T}^*_{\partial_1b_{\open{I}}} \coc{z}^{(\chi')*}_{(f',g')}(b_\open{I})=\\ 
 &= \coc{z}^{(\chi')}_{(f',g')}(b_\open{I}) \iota(A) \coc{z}^{(\chi')*}_{(f',g')}(b_\open{I}) = \pi^{\coc{z}^{(\chi')}_{(f',g')}}(A).\\ 
\end{split}
\end{equation}
Similarly, directly by the definition of $\psi^{\coc{z}^{(\chi)}_{(f,g)}}$ one also gets, if $\widetilde{\open{I}}\subseteq \open{I}$,
$$\itwiner{T}_{\open{I}} \psi^{\coc{z}^{(\chi)}_{(f,g)}}_{\open{I},\tilde{\open{I}}} =
\itwiner{T}_\open{I} \coc{z}^{(\chi)}_{(f,g)}(\open{I},\widetilde{\open{I}}) = \coc{z}^{(\chi')}_{(f',g')}(\open{I},\widetilde{\open{I}}) \itwiner{T}_{\tilde{\open{I}}}
=  \psi^{\coc{z}^{(\chi')}_{(f',g')}}_{\open{I},\tilde{\open{I}}} \itwiner{T}_{\tilde{\open{I}}}.$$
Summarizing, we showed that unitary equivalence of cocycles entails unitary equivalence of associated net representations. 
Let us prove the converse. To this end suppose that 
$\itwiner{T} \in (\{\pi^{\coc{z}^{(\chi)}_{(f,g)}}, \psi^{\coc{z}^{(\chi)}_{(f,g)}}\},\{\pi^{\coc{z}^{(\chi')}_{(f',g')}}, \psi^{\coc{z}^{(\chi')}_{(f',g')}}\})$
is unitary.
For every $\open{I}\in \poset{R}$, define the unitary operator
\begin{equation}
\itwiner{t}_\open{I} := \coc{z}^{(\chi')}_{(f',g')}(b_{\open{O}\open{I}})^* \itwiner{T}_\open{O}\coc{z}^{(\chi)}_{(f,g)}(b_{\open{O}\open{I}})\:,
\end{equation}
where $b_{\open{O}\open{I}}\in \simplex{1}{\poset{R}}$ is a positive oriented simplex 
such that $\partial_1  b_{\open{O}\open{I}} =\open{I}$, $\partial_0  b_{\open{O}\open{I}} =\open{O}$ and $\open{I}\perp\open{O}$. We want to prove that $\itwiner{t}_\open{I}$ defines a 
localized intertwiner between $\coc{z}^{(\chi)}_{(f,g)}$ and $\coc{z}^{(\chi ')}_{(f',g')}$ .
First of all we notice that $\itwiner{t}_\open{I}$ doesn't depend on the choice of $\open{O}\perp\open{I}$ because, if $\widetilde{\open{O}} \subseteq \open{O}$ one has
\begin{equation}
\begin{split}
\itwiner{t}_\open{I} &:= \coc{z}^{(\chi')}_{(f',g')}(b_{\open{O}\open{I}})^* \itwiner{T}_\open{O} \coc{z}^{(\chi)}_{(f,g)}(b_{\open{O}\open{I}}) =
\coc{z}^{(\chi')}_{(f',g')}(b_{\open{O}\open{I}})^* \itwiner{T}_\open{O} \coc{z}^{(\chi)}_{(f,g)}(\open{O},\tilde{\open{O}})
\coc{z}^{(\chi)}_{(f,g)}(b_{\tilde{\open{O}}\open{I}})
 =\\ 
&=\coc{z}^{(\chi')}_{(f',g')}(b_{\open{O}\open{I}})^* \itwiner{T}_\open{O}\psi_{\open{O}\tilde{\open{O}}}^{\coc{z}^{(\chi)}_{(f,g)}}
\coc{z}^{(\chi)}_{(f,g)}(b_{\tilde{\open{O}}\open{I}}) 
= \coc{z}^{(\chi')}_{(f',g')}(b_{\open{O}\open{I}})^*  \psi_{\open{O}\tilde{\open{O}}}^{\coc{z}^{(\chi')}_{(f',g')}}\itwiner{T}_{\tilde{\open{O}}} \coc{z}^{(\chi)}_{(f,g)}(b_{\tilde{\open{O}}\open{I}})=\\ 
&=\coc{z}^{(\chi')}_{(f',g')}(b_{\open{O}\open{I}})^* \coc{z}^{(\chi')}_{(f',g')}(\open{O},\tilde{\open{O}})\itwiner{T}_{\tilde{\open{O}}} \coc{z}^{(\chi)}_{(f,g)}(b_{\tilde{\open{O}}\open{I}}) = \coc{z}^{(\chi')}_{(f',g')}(b_{\tilde{\open{O}}\open{I}})^* 
\itwiner{T}_{\tilde{\open{O}}} \coc{z}^{(\chi)}_{(f,g)}(b_{\tilde{\open{O}}\open{I}}). 
\end{split}
\end{equation}
Using a suitable chain of $1$-simplices and the above identity, one can pass to the initial $\open{O}\perp\open{I}$ to any other $\open{O}_1\perp\open{I}$. Now notice that, if $A\in \sweyl{\open{O}}$
\begin{equation}
\begin{split}
\itwiner{t}_{\open{I}}\iota_\open{O}(A) &= \coc{z}^{(\chi')}_{(f',g')}(b_{\open{O}\open{I}})^* \itwiner{T}_\open{O} \coc{z}^{(\chi)}_{(f,g)}(b_{\open{O}\open{I}})\iota_\open{O}(A)=\\ 
&=\coc{z}^{(\chi')}_{(f',g')}(b_{\open{O}\open{I}})^* \itwiner{T}_\open{O}\pi^{\coc{z}^{(\chi)}_{(f,g)}}_\open{O}(A) \coc{z}^{(\chi)}_{(f,g)}(b_{\open{O}\open{I}})=\\
&=\coc{z}^{(\chi')}_{(f',g')}(b_{\open{O}\open{I}})^*  \pi^{\coc{z}^{(\chi')}_{(f',g')}}_\open{O}(A) \itwiner{T}_\open{O}\coc{z}^{(\chi)}_{(f,g)}(b_{\open{O}\open{I}})  = \iota_\open{O}(A) \itwiner{t}_\open{I}.\\
\end{split}
\end{equation}
So $\itwiner{t}_\open{I} \in \iota_\open{O}(\sweyl{\open{O}})'$.
By Haag duality, and using the fact that $\open{O}\perp\open{I}$ is generic, we conclude that $\itwiner{t}_\open{I} \in \iota_\open{I}(\sweyl{\open{I}})'' = \snet{R}{\open{I}}$. 
Finally let us prove that $\itwiner{t}$ is a cocycle intertwiner. Consider $b\in \simplex{1}{\poset{R}}$ with $\partial_0b = \open{I}$.
Fix $\open{O}\perp |b|$ and choose two positively-oriented $1$-simplices ending at $\open{O}$ and starting from $\partial_0b$ and $\partial_1b$ respectively, say $b_{\open{O}\partial_0 b}$ and $b_{\open{O}\partial_1 b}$, such that
 $b_{\open{O}\partial_0 b}\ast b= b_{\open{O}\partial_1 b}$ (thus $b_{\open{O}\partial_0 b} = b_{\open{O}\partial_1 b}\ast\bar{b}$); then we can write
\begin{equation}
\begin{split}
\itwiner{t}_{\partial_0 b} \coc{z}^{(\chi)}_{(f,g)}(b) &= \coc{z}^{(\chi')}_{(f',g')}(b_{\open{O}\partial_0 b})^* \itwiner{T}_\open{O}\coc{z}^{(\chi)}_{(f,g)}(b_{\open{O}\partial_0 b}) \coc{z}^{(\chi)}_{(f,g)}(b) =\\
&=\left(\coc{z}^{(\chi')}_{(f',g')}(b_{\open{O}\partial_1 b})\cdot\coc{z}^{(\chi')}_{(f',g')}(\bar{b})\right)^{*}\itwiner{T}_\open{O}\coc{z}^{(\chi)}_{(f,g)}(b_{\open{O}\partial_1 b)} =\\
&=\coc{z}^{(\chi')}_{(f',g')}(b)\cdot\coc{z}^{(\chi')}_{(f',g')}(b_{\open{O}\partial_1 b})^{*}\itwiner{T}_\open{O}\coc{z}^{(\chi)}_{(f,g)}(b_{\open{O}\partial_1 b)} = \coc{z}^{(\chi')}_{(f',g')}(b)\itwiner{t}_{\partial_1 b}.\\
\end{split}
\end{equation}
Let us now prove that $\{\pi^{\coc{z}^{(\chi)}_{(f,g)}}, \psi^{\coc{z}^{(\chi)}_{(f,g)}}\}$ is irreducible.
Suppose there is a unitary intertwiner $\itwiner{U} \in (\{\pi^{\coc{z}^{(\chi)}_{(f,g)}}, \psi^{\coc{z}^{(\chi)}_{(f,g)}}\},\{\pi^{\coc{z}^{(\chi)}_{(f,g)}}, \psi^{\coc{z}^{(\chi)}_{(f,g)}}\})$\:.
As a consequence the  operators $\itwiner{t}_\open{I}:= \coc{z}^{(\chi)}_{(f,g)}(b_{\open{O}\open{I}})^* \itwiner{U}_\open{O} \coc{z}^{(\chi)}_{(f,g)}(b_{\open{O}\open{I}})$, where $\open{I},\open{O}\in \poset{R}$, 
$\open{O}\perp\open{I}$ and $b_{\open{O}\open{I}}$ is positively oriented, define a unitary intertwiner
$\itwiner{t}\in (\coc{z}^{(\chi)}_{(f,g)}, \coc{z}^{(\chi)}_{(f,g)})$. Statement (a) in theorem \ref{th:propcocycle}
implies that the $\itwiner{t}_\open{I}$ are all of the form $c\cdot 1$ with $c\in \bC$ and $|c|=1$. Therefore the $\itwiner{U}_\open{O}$ have the same form and so,
since $\open{O}$ can be chosen arbitrarily in $\poset{R}$, 
$\{\pi^{\coc{z}^{(\chi)}_{(f,g)}}, \psi^{\coc{z}^{(\chi)}_{(f,g)}}\}$ is irreducible.\\
(c) This statement is an immediate consequence of (\ref{eq:zz}) and Weyl relations.
\end{proof}

\noindent\textbf{Remark.} The definition of $\{ \pi^{\coc{z}_{(f,g)}^{(\chi)}}, \psi^{\coc{z}_{(f,g)}^{(\chi)}}\}$ can be modified 
changing the requirements on the simplex $b_\open{I}$. These changes don't affect the results in higher dimensions as 
established in Theorem 4.3 in \cite{Brunetti_Ruzzi_08} where $b_\open{I}$ can be replaced by any path $\path{p}_\open{I}$ ending on $\open{I}$
but starting from $\partial_1\path{p} \perp \open{I}$. Remarkably,  the situation is different here.
Replacing the $1$-simplex $b_\open{I}$ in (\ref{eq:1pi}) with a path $\path{p}_\open{I}$
ending in $\open{I}$ which winds $n\in \bZ$ times around the circle before reaching $\open{I}$ and such that the final $1$-simplex
ending on $\open{I}$ is positive oriented, with the first endpoint in $\open{I}'$, 
\begin{eqnarray}
\rho^{\coc{z}^{(\chi)}_{(f,g)}}_\open{I}(A) &:=& \coc{z}^{(\chi)}_{(f,g)}(\path{p}_\open{I}) \iota_\open{I}(A) \coc{z}^{(\chi)}_{(f,g)}(\path{p}_\open{I})^*\:, \quad A \in \sweyl{\open{I}}\label{1pi'}\\
\phi^{\coc{z}^{(\chi)}_{(f,g)}}_{\open{I},\tilde{\open{I}}} &: =& \coc{z}^{(\chi)}_{(f,g)}\left(\open{I},\widetilde{\open{I}}\right) \label{1psi'}\:.
\end{eqnarray}
defines a net representation which is not encompassed in the class of representations considered in the theorem just proved.
 However this new net representation turns out to be unitarily equivalent to  $\{{\pi}^{\coc{z}_{(f,g)}^{(\chi)}}, {\psi}^{\coc{z}_{(f,g)}^{(\chi)}}\}$, it being
\begin{eqnarray}
{\rho}^{\coc{z}_{(f,g)}^{(\chi)}}(A) &=&  Z(nf,ng) \pi^{\coc{z}_{(f,g)}^{(\chi)}}(A) Z(nf,ng)^*,\nonumber\\   
{\phi}^{\coc{z}_{(f,g)}^{(\chi)}} &=& Z(nf,ng)^* {\psi}^{\coc{z}_{(f,g)}^{(\chi)}} Z(nf,ng)= {\psi}^{\coc{z}_{(f,g)}^{(\chi)}}\nonumber. 
\end{eqnarray}
Another, more dramatic change may be performed in the definition (\ref{eq:1pi}), if one assumes that the $1$-simplex $b_\open{I}$ with end points  $\open{I}$ and $\open{J}$ 
is {\em negative} oriented. In this case one is committed to  replace also $\coc{z}^{(\chi)}_{(f,g)}\left(\open{I},\widetilde{\open{I}}\right)$ with 
$\coc{z}^{(\chi)*}_{(f,g)}\left(\open{I},\widetilde{\open{I}}\right)$ in the definition (\ref{eq:1psi}), in order to 
obtain a net representation. With these changes, definitions (\ref{eq:1pi}) and (\ref{eq:1psi}) work anyway 
and give rise to a different net representation
$\{ \tilde{\pi}^{\coc{z}_{(f,g)}^{(\chi)}}, \tilde{\psi}^{\coc{z}_{(f,g)}^{(\chi)}}\}$.
Also this net representation is not encompassed in the class of representations considered in the theorem.
However it is globally unitarily equivalent to a representation as those in the theorem, {\em but associated with a different cocycle}.
In fact it turns out to be unitarily equivalent to $\{\pi^{\coc{z}_{(-f,-g)}^{(\chi)}}, \psi^{\coc{z}_{(-f,-g)}^{(\chi)}}\}$, where we stress that the sign in front of $f$ and $g$, and thus the cocycle, has changed.
Indeed, after a trivial computation based on the explicit form of cocycles, one finds that:
\begin{eqnarray}
\tilde{\pi}^{\coc{z}_{(f,g)}^{(\chi)}}(A) &=&  Z(f,g) \pi^{\coc{z}_{(-f,-g)}^{(\chi)}}(A) Z(f,g)^*,\nonumber\\   
\tilde{\psi}^{\coc{z}_{(f,g)}^{(\chi)}} &=& Z(f,g)^* \tilde{\psi}^{\coc{z}_{(-f,-g)}^{(\chi)}} Z(f,g)\nonumber
= \tilde{\psi}^{\coc{z}_{(-f,-g)}^{(\chi)}}.
\end{eqnarray}

\appendix
\chapter{Proof of some propositions.} \label{app:proofs}
 
\noindent {\bf  Proof of theorem \ref{th:spatial Reeh-Schlieder}}.

%TODO: cite Strohmaier_00
First we prove that $\Psi_0$ is cyclic for $\pi_0(\sweyl{\open{I}})\Psi_0$, $\forall\;\open{I}\in \poset{R}$; it is sufficient to show that every $\Upsilon \in \varhs{H}_0$ with $\Upsilon \perp \pi_0(\sweyl{\open{I}})\Psi_0$ has to vanish.
If $\open{I}\in \poset{R}$, take another  $\open{I}_1 \in \poset{R}$ with $\overline{\open{I}_1} \subset \open{I}$. Recalling the remark made in subsection \ref{sec:intro_properties}, we can found a neighborhood $\open{O}$ of $(0,0) \in \bR^2$
such that $\mathbf{t} \in \open{O}$ implies $\open{I}_1+\mathbf{t} \subset \open{I}$. Consider operators $A_1,\ldots, A_n \in \pi_0(\sweyl{\open{I}_1})$ and, for
$\mathbf{t}_1,\ldots, \mathbf{t}_n$ define $A_i(\mathbf{t}_i) := U_{\mathbf{t}_i} A_i U_{\mathbf{t}_i}^*$, where $U_{\mathbf{t}} = U_{(r,s)} = e^{-\imath(rP^\otimes -s H^\otimes)}$ implements $\{\gamma_{\mathbf{t}}\}_{\mathbf{t} \in \bR^2}$ leaving $\Psi_0$ fixed
as established in theorem \ref{th:vacuum}. By construction, taking into account the propagation properties of solutions of the Klein-Gordon equation and spacetime simmetries, we have
$U_{\mathbf{t}_i} A_i U_{\mathbf{t}_i}^*  \in \pi_0(\sweyl{\open{I}_1+ \mathbf{t}_i})  \subset  \pi_0(\sweyl{\open{I}})$; so that $A_i(\mathbf{t}_i) \in \pi_0(\sweyl{\open{I}})$
for $i=1,\ldots,n$
whenever $\mathbf{t}_i \in \open{O}$. As a consequence of the hypothesis on $\Upsilon$ we get:
\begin{equation}\label{eq:central}
\ip{\Upsilon}{A_1(\mathbf{t}_1)\cdots A_n(\mathbf{t}_n) \Psi_0} = 0\quad \mbox{if}\: \mathbf{t}_1,\ldots,\mathbf{t}_n \in \open{O}. 
\end{equation}
Referring to the joint spectral measure of $P^\otimes$ and $H^\otimes$ (in the given order) we may write:
\begin{equation}\label{eq:central2} 
 U_{\mathbf{t}} = \int_{\Gamma_m} e^{-\imath (\mathbf{t}|\mathbf{q})} dP(\mathbf{q})
\end{equation}
where $((r,s)|(r',s')) = rr'-ss'$ is the Minkowskian product in $\bR^2$ and $\mathbf{q}$ ranges in the joint spectrum given by the
discrete hyperboloid:
$$\Gamma_m := \{ (r',s') \in \bR^2 \:\:|\:\: r'\in \bZ\setminus \{0\}\:,\: s'= k_m(r')\}.$$
We can extend $\mathbf{t}$ in the right-hand side of (\ref{eq:central2}) to complex values: $\mathbf{t} \to \mathbf{z} \in \bC^2$ provided
$(\Im m \: \mathbf{z} | \mathbf{q}) \leq 0$ for all $\mathbf{q}\in \Gamma_m$. Therefore $U_{\mathbf{t}}$ is extensible
for complex values of $\mathbf{z}$ such that $\Im m \: \mathbf{z}$ belongs to the closure of future light-cone $${V^+}
 := \{(r,s) \in \bR^2\:\:|\:\: ((r,s)|(r,s)) < 0\:,\, s > 0\}\:.$$ Employing standard tools of spectral theory
 in Hilbert spaces, one obtains that the map $\mathbf{z} \mapsto U_{\mathbf{z}}$
  is continuous for $\Im m \:\mathbf{z} \in \overline{V^+}$
 in the strong-operatorial topology and is holomorphic over $V^+$
 when viewed as a function valued in the Banach space of bounded operators on $\varhs{H}_0$. Finally $\norm{U_{\mathbf{z}}} \leq 1$
if $\Im m \:\mathbf{z} \in \overline{V^+}$.
Remembering that analyticity implies weak analyticity,
one can prove straightforwardly that, for $(\mathbf{z}_1, \ldots, \mathbf{z}_n)\in \bC^n$
 the function (notice that $U_{\mathbf{z}}\Psi_0 = \Psi_0$)
\begin{equation} \label{eq:inter}
(\mathbf{z}_1, \ldots, \mathbf{z}_n) \mapsto \ip{\Upsilon}{U_{\mathbf{z}_1} A_1 U_{\mathbf{z}_2}A_2 \cdots U_{\mathbf{z}_n} A_n \Psi_0},
\end{equation}
is well-defined if $(\Im m \:\mathbf{z}_1, \ldots, \Im m \:\mathbf{z}_n) \in \overline{V^+}\times \cdots \times \overline{V^+}$, it is
jointly continuous in this domain, analytic in each variable separately for
 $(\Im m \:\mathbf{z}_1, \ldots, \Im m\: \mathbf{z}_n) \in V^+\times \cdots \times V^+$ (and thus analytic as a function of several variables)
and it reduces to the left-hand side of (\ref{eq:central})
 when $\mathbf{z}_1 = \mathbf{t}_1 $, $\mathbf{z}_2 = \mathbf{t}_2-\mathbf{t}_1$, ... $\mathbf{z}_n = \mathbf{t}_n -\mathbf{t}_{n-1}$ for $\mathbf{t}_k \in \bR^2$, $k=1,\ldots,n$.
 In this case it must vanish when $\mathbf{t}_k \in \open{O}$, $k=1,\ldots,n$.
 Using the ``edge of the wedge theorem'' we are committed to conclude that the  function in (\ref{eq:inter})
 vanishes everywhere on $(\open{O} \times \ldots \times \open{O}) + \imath (V^+ \times \ldots \times V^+)$. However since the domain
 of analyticity of the considered function is connected, it must vanish everywhere therein. Finally, by continuity it must vanish
 on the boundary of the domain of analyticity too, in particular:
\begin{equation}\label{eq:central4}
\ip{\Upsilon}{A_1((r_1,0))\cdots A_n((r_n,0)) \Psi_0} = 0 \quad\mbox{if}\: r_1,\ldots,r_n \in \bR.  
\end{equation}
This implies that $\Upsilon$ is orthogonal to $\bigvee_{r\in \bR} \pi_0\left( \sweyl{\open{I}_1+(r,0) }\right) \Psi_0$. By spatial weak
additivity that space is dense in $\varhs{H}_0$ and thus $\Upsilon = 0$.

We conclude the proof by showing that $\Psi_0$ is separating for every $\snet{R}{\open{I}}$.  We have to prove that, for any fixed $\open{I}\in \poset{R}$, if $A,B \in \snet{R}{\open{I}}$ and $A\Psi_0 = B\Psi_0$ then $A=B$.
Take $\open{I}_1\in \poset{R}$ with $\open{I}_1 \cap \open{I} = \emptyset$ and thus $\snet{R}{\open{I}_1}$ commute with $\snet{R}{\open{I}}$ by locality.
If $C\in \snet{R}{\open{I}_1}$ we have $CA\Psi_0 = CB\Psi_0$ and so $AC\Psi_0 = BC \Psi_0$. $\snet{R}{\open{I}_1} \Psi_0$ is dense in $\varhs{H}_0$ due to (a); since $C\in \snet{R}{I_1}$ is arbitrary,  $AC\Psi_0 = BC \Psi_0$ entails $A=B$ on a
dense domain and thus everywhere in $\varhs{H}_0$ by continuity.   $\Box$\\

\noindent {\bf  Proof of lemma \ref{lem:lemmaE}}.
In the following  $\lambda \in \bR$. We complete the unit-norm vector $\psi\in \hs{H}$ to a Hilbert basis of
$\hs{H}$, pass to the associated Hilbert basis in $\fock{\hs{H}}$
 and denote by  $F$ the dense subspace of $\fock{\hs{H}}$ containing all the finite linear combinations of the vectors of that basis.
Assuming $E_{\lambda}^{(\psi)}= e^{\imath\lambda A}$, taking the derivative at $\lambda=0$ of the identity
$$E_{\lambda}^{(\psi)}  W[\psi] E_{\lambda}^{(\psi)*}  =
W\left[e^\lambda \psi\right]\:, 
$$
(without paying much attention to domain issues) and, finally,  making use of (\ref{eq:Weylgen}), one gets that
\begin{equation}\label{eq:commF}
\left[\imath A, a(\psi) -a^*(\psi)\right] \Phi = (a(\psi)-a^*(\psi)) \Phi,
\end{equation}
if $\Phi$ belongs to some suitable domain we  shall determine presently.
Taking the commutation relation
$[a(\psi), a^*(\psi)] = 1$  into account (recall that $\norm{\psi}=1$),  we see that a candidate for $A$
is some self-adjoint extension of $A:= (\imath/2) (a(\psi)a(\psi) - a^*(\psi)a^*(\psi))$.
$A$ turns out to be symmetric if defined on $F$. If $\Phi \in F$ contains exactly $k$ particles in the state $\psi$
one finds $\norm{A^n \Phi}\leq \sqrt{(2n + k)!}$. From that it arises
 $\sum_{n=0}^{+\infty}\lambda^n \norm{A^n \Phi}/n! <+\infty$ if $|\lambda| <1/2$. Therefore
the vectors in $F$ are analytic for $A$ and thus
 $A$ is essentially
self-adjoint on $F$, $\overline{A}$ being its unique self-adjoint extension.
In particular the commutation relation (\ref{eq:commF}) are, in fact, valid for $\Phi \in F$ and lead to the further commutation relations
\begin{equation}\label{eq:commF2}
\left[(\imath A)^n, a(\psi) -a^*(\psi)\right] \Phi = \sum_{k=0}^{n-1} \binom{n}{k}  (a(\psi)-a^*(\psi)) (\imath A)^k \Phi
\end{equation}
\noindent for all $\Phi\in F$.
Using (\ref{eq:commF2}) one easily proves the validity of this identity  for $|\lambda|<1/4$ and $\Phi \in F$:
\begin{equation}\label{eq:InterE} 
\sum_{n=0}^{+\infty} \frac{(\imath\lambda A)^n}{n!} (a(\psi)-a^*(\psi)) \Phi  =
\sum_{n=0}^{+\infty} (a(\psi)-a^*(\psi)) \frac{(\imath\lambda A + \lambda 1)^n}{n!}  \Phi.
\end{equation}
The series
$\sum_{n=0}^{+\infty} \frac{(\imath\lambda A + \lambda 1)^n}{n!}  \Phi$
 converges for every $\Phi \in F$ and $|\lambda |< 1/4$ as one can establish making use of the bounds
$\norm{(A+\imath 1)^n \Phi} \leq 2^n \sqrt{(2n+k)!}$ 
when $\Phi \in F$ contains exactly $k$ particles in the state $\psi$. Therefore closedness of $\overline{a(\psi)-a^*(\psi)}$
 imply, via (\ref{eq:InterE}), that the following two facts hold:\\

\noindent \textbf{(i)} $\sum_{n=0}^{+\infty} \frac{(\imath\lambda A + \lambda 1)^n}{n!}  \Phi \in \Dom{\overline{a(\psi)-a^*(\psi)}}$ when
 $\Phi \in F$, $|\lambda|< 1/4$,\\

\noindent\textbf{(ii)} $\overline{a(\psi)-a^*(\psi)} \sum_{n=0}^{+\infty} \frac{(\imath\lambda A + \lambda 1)^n}{n!} 
 \Phi = \sum_{n=0}^{+\infty} (a(\psi)-a^*(\psi)) \frac{(\imath\lambda A + \lambda 1)^n}{n!}  \Phi$\\

\noindent Therefore (\ref{eq:InterE}) can be re-written as
\begin{equation}\label{eq:InterEx}
 e^{\imath\lambda \overline{A}} (a(\psi)-a^*(\psi)) \Phi  =
\overline{a(\psi)-a^*(\psi)} \sum_{n=0}^{+\infty}  \frac{(\imath\lambda A + \lambda 1)^n}{n!}  \Phi,
\end{equation}
 where we have also used the fact that $(a(\psi)-a^*(\psi)) \Phi \in F$ when $\Phi\in F$ and thus the exponential
 $e^{\imath\lambda \overline{A}} (a(\psi)-a^*(\psi)) \Phi$ can be expanded in series. Since $\lambda 1$ and $\imath\lambda A$
 commute, following exactly the same proof as used for numbers,
 one achieves $ \sum_{n=0}^{+\infty}  \frac{(\imath\lambda A + \lambda
 1)^n}{n!}  \Phi =e^\lambda \sum_{n=0}^{+\infty}\frac{(\imath\lambda A)^n}{n!}\Phi$. 
On the other hand, since $\Phi$ is analytic for $A$, the right-hand side is nothing but
 $e^{\lambda} e^{\imath\lambda \overline{A}} \Phi$.   
Summing up, the identity (\ref{eq:InterEx}) can be re-stated as
$$e^{\imath\lambda \overline{A}}\, \overline{a(\psi)-a^*(\psi)} \Phi =  \overline{e^\lambda(a(\psi)-a^*(\psi))}\,e^{\imath\lambda \overline{A}}\Phi\:,
\quad \forall\,\, \Phi \in F,\, |\lambda| <1/4.$$
Using recursively this identity we come immediately to
\begin{equation} \label{eq:InterE2}
e^{\imath\lambda \overline{A}}\, \overline{a(\psi)-a^*(\psi)}^n \Phi = \overline{e^{\lambda}(a(\psi)-a^*(\psi))}^n
e^{\imath\lambda \overline{A}}\Phi\quad \forall\,\, \Phi \in F,\, |\lambda| <1/4.
\end{equation}
Since $e^{\imath\lambda \overline{A}}$ is unitary, (\ref{eq:InterE2}) entails that, for $\Phi \in F$, $|\lambda| <1/4$ and every $u \in \bC$:
$$ \sum_{n=0}^\infty \frac{u^n}{n!} \norm{\overline{e^{\lambda}(a(\psi)-a^*(\psi))}^n e^{\imath\lambda \overline{A}} \Phi}  =
 \sum_{n=0}^\infty \frac{u^n}{n!}\norm{\overline{a(\psi)-a^*(\psi)}^n \Phi} < +\infty\:,
$$
where we have used the fact that every $\Phi \in F$ is analytic (for every value of the parameter $u$) for $ \overline{\imath a(\psi)-\imath a^*(\psi)}$
as is well known (see \cite{Bratteli_Robinson_II}). We have found that $e^{\imath\lambda \overline{A}} \Phi$ is analytic for
$\overline{e^{\lambda}(a(\psi)-a^*(\psi))}$. In this context, the identity arising from (\ref{eq:InterE2}) for
$\Phi \in F$ and $|\lambda| <1/4$,
$$e^{\imath\lambda \overline{A}} \sum_{n=0}^{+\infty}\frac{1}{n!}\overline{a(\psi)-a^*(\psi)}^n \Phi =
\sum_{n=0}^{+\infty}\frac{1}{n!} \overline{e^{\lambda}(a(\psi)-a^*(\psi))}^n
e^{\imath\lambda \overline{A}}\Phi $$
 can be re-written as
$e^{\imath\lambda \overline{A}} e^{\overline{a(\psi)-a^*(\psi)}} \Phi =
e^{\overline{e^\lambda (a(\psi)-a^*(\psi))}} 
e^{\imath\lambda \overline{A}} \Phi$.
That is, taking advantage from the fact that $F$ is dense,
$E_{\lambda}^{(\psi)}  e^{\overline{a(\psi)-a^*(\psi)}} E_{\lambda}^{(\psi)*}  =
e^{ \overline{e^\lambda(a(\psi)-a^*(\psi))}} 
$,
where we have defined $E_{\lambda}^{(\psi)} := e^{\imath\lambda \overline{A}}$.
Finally, employing $\bR$-linearity of $\psi \mapsto a(\psi), a^*(\psi)$, the achieved formula can be
re-stated as
$$E_{\lambda}^{(\psi)}  W[\psi] E_{\lambda}^{(\psi)*}  =
W\left[e^\lambda \psi\right]\:. 
$$
The restriction $|\lambda| <1/4$ can be dropped by employing iteratively the identity above and noticing that
$E_{\lambda}^{(\psi)}$ is additive in $\lambda\in \bR$. Hence the obtained identity holds true for every
$\lambda \in \bR$.
$\Box$\\
 
\noindent {\bf Proof of proposition  \ref{prop:last}}.
 Choose two intervals $\open{J}_1, \open{J}_2 \in \poset{R}$ with $(-\pi, \pi) \supset \overline{\open{J}_2}$,
$\open{J}_2 \supset \overline{\open{J}_1}$ and $\open{J}_1 \supset \overline{\open{J}}$.
As a further ingredient we fix an open neighborhood of $1$, $\open{O} = (e^{-\omega}, e^{\omega})$ with $\omega >0$  so small that
(1) $\lambda \overline{\open{J}} \subset \open{J}_1$, (2) $\lambda \overline{\open{J}_2} \subset (-\pi, \pi)$ for all $\lambda \in \open{O}$. Notice that
$\lambda \in \open{O}$ iff $\lambda^{-1} \in \open{O}$.
With these definitions, let $\chi \in \smoothfuncv{\bS^1}{\bR}$ be 
 such that $0 \leq \chi(\theta) \leq 1$ for $\theta \in \bS^1$ and
 $\chi(\theta) =1$ for $\theta \in \open{J}_1$ but $\chi(\theta) =0$ in $\bS^1 \setminus \open{J}_2$.
Now  consider the class of operators $U_\lambda : \sqint{\bS^1}{\theta} \to \sqint{\bS^1}{\theta}$, $\lambda \in \open{O}$, defined by:
  $$(U_\lambda f)(\theta):= \frac{\chi(\theta)}{\sqrt{\lambda}} f(\theta/\lambda)\:, \quad \forall\,\theta \in (-\pi,\pi]\:.$$
   Exploiting the presence of the smoothing function $\chi$ and making a trivial change of variables where appropriate, one proves the following features of $U_\lambda$:
  \begin{eqnarray}
    U_\lambda (\smoothfuncv{\bS^1}{\bR})& \subset & \smoothfuncv{\bS^1}{\bR}\:, \quad \forall\,  \lambda \in \open{O},\label{eq:smoothU}\\
  \norm{U_\lambda} &\leq& 1\label{eq:normU}\:, \quad \forall\,  \lambda \in \open{O},\\
  U_1\rest_{\sqint{\open{J}}{\theta}} &=& 1 \label{eq:IU},\\
  \lim_{\lambda\to 1} U_\lambda f &=& f\quad\mbox{if}\: f \in \smoothfuncv{\open{J}}{\bC}. \label{eq:sadded}
  \end{eqnarray}
  By direct inspection one also finds that:
 \begin{equation}
  (U^*_\lambda f)(\theta)  =   \sqrt{\lambda}\,\chi(\lambda \theta) f(\lambda \theta), \quad \forall\, f\in
\sqint{\bS^1}{\theta},\label{eq:U*U}
  \end{equation}
\noindent for all $\theta \in (-\pi,\pi]$ and $\lambda \in \open{O}$.
  From (\ref{eq:U*U}) properties analogous to those found for $U_\lambda$ can be straightforwardly established:
   \begin{eqnarray}
    U^*_\lambda (\smoothfuncv{\bS^1}{\bR})&\subset& \smoothfuncv{\bS^1}{\bR}\:, \quad \forall\,  \lambda \in \open{O}\:,\label{eq:smoothU*}\\
  \norm{U^*_\lambda} &\leq& 1\label{eq:normU*}\:, \quad \forall\,  \lambda \in \open{O}\:,\\
  U^*_1\rest_{\sqint{\open{J}}{\theta}} &=& 1 \label{eq:IU*}\:,\\
   U^*_{1/\lambda} \rest_{\sqint{\open{J}}{\theta}} &=&  U_{\lambda} \rest_{\sqint{\open{J}}{\theta}} \label{eq:UU*}\:, \quad \forall\,  \lambda \in \open{O}\:,\\
   \lim_{\lambda\to 1} U^{*}_\lambda f &=& f\quad\mbox{if}\: f \in \smoothfuncv{\open{J}}{\bC}. \label{eq:ssadded}
  \end{eqnarray}
  \noindent\textbf{Remark}. 
In view of the definition of $U_\lambda$ and (\ref{eq:UU*}), if $f \in \testfuncv{\open{J}}{\bR}$ then $$\supp{U_\lambda f}= \supp{U^*_{1/\lambda}f} = \lambda \:\supp{f}\:.$$

\noindent  Taking the definition of the quantization map (\ref{eq:1-p struct}) into account, one realizes that a candidate for $D_\lambda$ is the following operator, initially defined on $\smoothfuncv{\bS^1}{\bC}$:
\begin{equation}\label{eq:candidate} 
D^{(0)}_\lambda \psi := A^{1/4} U_\lambda A^{-1/4} \:\Re e (\psi) \:+\:  \imath  A^{-1/4} U^*_{1/\lambda} A^{1/4}  \:\Im m (\psi)
\end{equation}
for all $\lambda \in \open{O}$ and $\psi \in M_{\open{J}_{0}}$.
The right hand side is in fact  well-defined if $\psi \in K(\Cdata{\open{L}})$ with $\poset{R} \ni L \subsetneq \open{J}$;
$A^{-1/4} \: \Re e (\psi)$ and $A^{1/4}  \:\Im m (\psi)$ belong to $\testfuncv{J_0}{\bR}$ so both act on vectors belonging to their domains due to (\ref{eq:smoothU}) and (\ref{eq:smoothU*}).
Moreover it fulfills (a) in the thesis since $D^{(0)}_\lambda \psi \in K(\Cdata{\lambda\open{L}})$ in view of the previous remark.
However both operators $A^{1/4} U_\lambda A^{-1/4}$ and $A^{-1/4} U^*_{1/\lambda} A^{1/4}$ are well defined on $\smoothfuncv{\bS^1}{\bC}$.
To extend the validity of (a)  to every space
 $M_{\open{L}} := \overline{K(\Cdata{\open{L}})}$
with $\open{L} \subsetneq \open{J}$ as requested in the thesis,
it is sufficient  to prove that the  operators $A^{1/4} U_\lambda A^{-1/4}$ and $A^{-1/4} U^*_{1/\lambda} A^{1/4}$ are bounded on $\smoothfuncv{\bS^1}{\bC}$ and to extend them and $\mathcal D^{(0)}_\lambda$ by continuity on the whole space $\sqint{\bS^1}{\theta}$. The restriction
$\mathcal D_\lambda$ to $M_{\open{L}}$ of this extension will satisfy (a) by construction.

We'll use an argument based on an interpolation theorem.
 Consider $f \in \smoothfuncv{\bS^1}{\bC}$ and define $\chi_\lambda(\theta) := \chi(\lambda \theta)$.
     By direct inspection one finds that
     $\norm{A U_\lambda f}_{L^2}^2 \leq \lambda^{-4}\norm{A_{\lambda m} (\chi_\lambda f)}_{L^2}^2$,
  where $A_{\lambda m}$ is $A$ with mass $m$ replaced by $\lambda m$. One also finds that for $\lambda <1$ $\norm{A_{\lambda m} g}_{L^2}^2$ is bounded by
  $\norm{A g}^2$, otherwise by $\lambda^4 \norm{A g}^2$. Summarizing
  $$\norm{A U_\lambda f}_{L^2} \leq \: \sup_{\lambda \in \open{O}}\{1,\lambda^{-4}\}\: \norm{A (\chi_\lambda f)}^2_{L^2} \:.$$
  We can improve this upper bound expanding $A (\chi_\lambda f)$ as follows:
 \begin{equation}\label{eq:long} 
\norm{A (\chi_\lambda f)}_{L^2} \leq \norm{\chi_\lambda A f}_{L^2} + \norm{\frac{d^2 \chi_\lambda}{d\theta^2} f}_{L^2}
+ 2\norm{\frac{d \chi_\lambda}{d\theta} \frac{d f}{d\theta}}_{L^2}\:.
\end{equation}
    Now, integrating per parts we get:
\begin{equation*}
\begin{split}
 \norm{\chi_\lambda A f}_{L^2} &\leq  \norm{\chi_\lambda}_\infty \norm{A f}_{L^2} = \norm{A f}_{L^2}\:,\\
    \norm{\frac{d^2 \chi_\lambda}{d\theta^2} f}_{L^2} &\leq \norm{\frac{d^2 \chi_\lambda}{d\theta^2}}_{\infty}\norm{f}_{L^2},\\
    \norm{\frac{d \chi_\lambda}{d\theta} \frac{d f}{d\theta}}_{L^2} &\leq \norm{\frac{d \chi_\lambda}{d\theta}}_{\infty}
    \norm{\frac{d f}{d\theta}}_{L^2} \leq  \norm{\frac{d \chi_\lambda}{d\theta}}_{\infty}
   \sqrt{ \ip{\overline{f}}{\frac{d^2 f}{d\theta^2}}}\leq\\ 
&\leq\norm{\frac{d \chi_\lambda}{d\theta}}_{\infty} \sqrt{ \norm{f}  _{L^2} \norm{\frac{d^2 f}{d\theta^2}}_{L^2}}
\end{split}
\end{equation*}
Now notice that  $A \geq \lambda_0 1$ where $\lambda_0 >0$ is the least eigenvalue of $A$ and thus $\norm{Af}_{L^2} \geq \lambda_0 \norm{f}_{L^2}$. Similarly $A \geq  -\frac{d^2}{d\theta^2}$ and thus $\norm{Af}_{L^2}\geq \norm{d^2 f/ d\theta^2}_{L^2}$, therefore:
   \begin{eqnarray}
    \norm{\frac{d \chi_\lambda}{d\theta} \frac{d f}{d\theta} }_{L^2} &\leq &  \lambda_0^{-1/2}\norm{\frac{d \chi_\lambda}{d\theta}}_{\infty} \norm{Af}_{L^2} \nonumber
    \end{eqnarray}
  Using these estimates in (\ref{eq:long}) we finally obtains:
\begin{equation}\label{eq:stimaM}
 \norm{A U_\lambda f}_{L^2}  \leq C \norm{Af}_{L^2}
\end{equation}
for all $\lambda \in \open{O}$ and $f\in \smoothfuncv{\bS^1}{\bR}$, where
 $$C =  \sup_{\lambda \in \open{O}}\{1,\lambda^{-4}\} \sup_{\lambda \in \open{O}} \left\{1 + \norm{\frac{d^2 \chi_\lambda}{d\theta^2}}_{\infty} +
\lambda_0^{-1/2}\norm{\frac{d \chi_\lambda}{d\theta}}_{\infty}\right\} \:.$$
 $C$ is finite: it can be proved by shrinking $\open{O}$ and noticing that the function $(\lambda, \theta) \mapsto\chi_\lambda(\theta)$
and its derivatives are bounded in the compact $\overline{\open{O}} \times \bS^1$ since they are continuous. Since $\testfunc{\bS^1}$ is a core for the self-adjoint (and thus closed) operator $A$, as a byproduct (\ref{eq:stimaM}) implies:
\begin{eqnarray}
U_\lambda (\Dom{A}) &\subset& \Dom{A}\quad \forall\, \lambda \in \open{O},  \nonumber\\
\norm{A U_\lambda f}_{L^2}  &\leq& C \norm{Af}_{L^2}\:, \forall\, \lambda \in \open{O},\quad f\in \Dom{A}. \nonumber
\end{eqnarray}
The proof is immediate noticing that if $f \in \Dom{A}$ there is a sequence $\testfunc{\bS^1} \ni f_n \to f$ with
$Af_n \to Af$ and , in view of continuity of $U_{\lambda}$,  $\{U_{\lambda}f_n\}_{n\in \bN}$ is Cauchy and, in view of (\ref{eq:stimaM}), $\{AU_{\lambda}f_n\}_{n\in \bN}$
is Cauchy too. Closedness of $A$ implies that $U_{\lambda}f_n \to U_{\lambda}f \in \Dom{A}$ and $A(U_{\lambda}f_n) \to A(U_{\lambda}f)$. This also proves that (\ref{eq:stimaM})
is still valid in $\operatorname{Dom}(A)$ by continuity. As $A\geq 0$ and (\ref{eq:normU}) holds, \cite[prop. 9, chap. IX.5]{Reed_Simon_II} used twice implies that 
 \begin{eqnarray}
U_\lambda (\operatorname{Dom}(A^{1/4})) &\subset& \operatorname{Dom}(A^{1/4})\quad \forall\, \lambda \in \open{O}, \nonumber\\
\norm{A^{1/4} U_\lambda f}_{L^2}  &\leq& C^{1/4} \norm{A^{1/4}f}_{L^2}\:, \quad\forall\, \lambda \in \open{O},\quad f\in \Dom{A^{1/4}},\nonumber
\end{eqnarray}
so that, since $\operatorname{Ran}(A^{-1/4}) = \Dom{A^{1/4}}$ and
$\Dom{A^{-1/4}}$ is the whole Hilbert space,
\begin{equation}\label{eq:stime}
A^{1/4} U_\lambda A^{-1/4}= B_\lambda: \csqint{\bS^1}{\theta}\to \csqint{\bS^1}{\theta} \quad  
\end{equation}
with $\norm{B_\lambda}\leq C^{1/4}$ for all $\lambda \in \open{O}$. 

\noindent This concludes the proof of the continuity of the
former operator in the right-hand side of (\ref{eq:candidate}). 

Let us pass to the latter operator. By construction, taking the adjoint of $B_\lambda$ and replacing\footnote{Remind that $\lambda \in \open{O}$ iff $1/\lambda \in \open{O}$.} $\lambda$ with $1/\lambda$, we get $A^{-1/4} U^*_{1/\lambda} A^{1/4} \subset  B_{1/\lambda}^*$ on the dense domain $\Dom{A^{1/4}}$.
Since $B^*_{1/\lambda}$ is defined on the whole Hilbert space and $\norm{B^*_{1/\lambda}} = \norm{B_{1/\lambda}} \leq C^{1/4}$, we conclude that
$A^{-1/4} U^*_{1/\lambda} A^{1/4}$ continuosly extends to $B^*_{1/\lambda}: \csqint{\bS^1}{\theta}\to \csqint{\bS^1}{\theta}$.
 This concludes the proof of (a).

 Concerning the property (b): $\lim_{\lambda\to 1} D_\lambda\psi = 1$ for $\psi \in M_\open{L}$ with $\poset{R} \ni \open{L} \subsetneq \open{J}$, it is equivalent to prove that $B_\lambda \Re e (\psi) \to \Re e (\psi)$ and $B^*_{1/\lambda} \Im m (\psi) \to \Im m (\psi)$
  as $\lambda \to 1$ for $\psi \in M_\open{L}$.\\
Notice that $A^{-1/4}$ is continuous and so, when $\psi \in K(\Cdata{\open{L}})$, one has
\begin{equation}
\begin{split}
A^{-1/4} U^*_{1/\lambda} A^{1/4} (\Im m (\psi)) &\to A^{-1/4} U^*_1 A^{1/4} (\Im m (\psi)) =\\  
&=A^{-1/4} A^{1/4} (\Im m (\psi)) = \Im m (\psi),\\
\end{split}
\end{equation}
where we have used  (\ref{eq:ssadded}) and
(\ref{eq:IU*}) noticing that $A^{1/4} (\Im m (\psi)) \in \smoothfunc{\open{L}} \subset\sqint{\open{J}}{\theta}$ when $\psi \in K(\Cdata{\open{L}})$.
Due to the uniform bound (\ref{eq:stime}) the result can be extended to $M_{\open{L}}:= \overline{K(\Cdata{\open{L}})}$. If $\psi\in M_{\open{L}}$, let
   $K(\Cdata{\open{L}}) \ni \psi_n \to \psi$ and denote $\Im m (\psi_n)$ and $\Im m (\psi)$ respectively by $f_n$ and $f$; obviously $f_n \to f$. One has (for $\lambda \in \open{O}$ so that (\ref{eq:stime}) holds)
\begin{equation}
\begin{split}
\norm{B^*_{1/\lambda} f- f} &\leq\norm{B^*_{1/\lambda} (f-f_n)} + \norm{B^*_{1/\lambda} f_n -f_n} + \norm{f_n-f} \leq\\
&\leq (C^{1/4}+1)\norm{f-f_n} + \norm{B^*_{1/\lambda} f_n -f_n}.
\end{split}
\end{equation}
For any fixed $\epsilon>0$, taking $n= n_\epsilon$  such that $(C^{1/4}+1)\norm{f-f_{n_\epsilon}} < \epsilon/2$, we can found $\delta>0$
  such that $\lambda \in (1-\delta, 1+\delta)$ entails $\norm{B^*_{1/\lambda}f_{n_\epsilon} -f_{n_\epsilon}}< \epsilon/2$. Hence
  for that $\epsilon>0$, $\norm{B^*_{1/\lambda} f- f} <
  \epsilon$  provided  that $\lambda \in (1-\delta, 1+\delta)$. That is $B^*_{1/\lambda} \Im m (\psi_n) \to \Im m (\psi)$ as $\lambda \to 1^-$ for all $\psi \in M_{\open{L}}$. 

 To conclude we prove that $\lim_{\lambda\to 1}B_{\lambda} \Re e (\psi) = \Re e (\psi)$ for $\psi \in M_\open{L}$ and $\open{L} \subsetneq \open{J}$. Let us indicate $\Re e(\psi)$ by $f$. As before, first consider the case
$\psi \in K(\Cdata{\open{L}})$. This means in particular that  $f = A^{1/4}h$ for some $h\in \testfuncv{\open{J}}{\bR}$. Now notice that:
\begin{equation}\label{eq.ustep}
\norm{B_{\lambda} f - f}^2 = \norm{B_{\lambda} f}^2 + \norm{f}^2 - 2 \Re e\ip{f}{B_{\lambda} f}.
\end{equation}
In our case, as $\lambda \to 1^-$, due to (\ref{eq:IU}) and (\ref{eq:sadded}):
\begin{equation}
\begin{split}
\ip{f}{B_{\lambda} f} &= \ip{A^{1/4} h}{A^{1/4} U_\lambda h} = \ip{ A^{1/2} h}{U_\lambda h}\to \ip{A^{1/2} h}{h} =\\
&= \ip{A^{1/4} h}{A^{1/4} h} = \ip{f}{f}.\\
\end{split}
\end{equation}
Similarly $\norm{B_{\lambda} f}^2 \to \ip{f}{f}$ as $\lambda \to 1^-$ because:
\begin{equation}
\begin{split}
\norm{B_{\lambda} f}^2 &= \ip{f}{A^{1/4} U^*_\lambda A^{-1/4} A^{1/4} U_\lambda A^{-1/4} f} =\\ 
&=\ip{f}{A^{1/4} U^*_\lambda U_\lambda A^{-1/4} A^{1/4} h} = \ip{f}{A^{1/4} U^*_\lambda U_\lambda h}\\
\end{split}
\end{equation}
and from the definition of $U_\lambda$ and $U_\lambda^{*}$ one finds that 
$U^*_\lambda U_\lambda h = h$, for each $h\in \smoothfuncv{\open{J}}{\bR}$. Putting all together one concludes that $\lim_{\lambda\to 1}B_{\lambda} \Re e (\psi) = \Re e (\psi)$ when $\psi \in K(\Cdata{\open{L}})$. Extension to the case $\psi \in M_\open{L} := \overline{K(\Cdata{\open{L}})}$ is the same as before.  $\Box$\\

\noindent {\bf Proof of lemma \ref{lem:Tomita}}. As a general fact it  holds $\Cdata{\open{I}} \subset (\Cdata{\open{I}'})'$ and
$\Cdata{\open{J}} \subset (\Cdata{\open{J}'})'$, thus applying the quantization map, taking closures and intersections, we get $M_\open{I} \cap M_\open{J}  \subset (M_{\open{I}'})' \cap (M_{\open{J}'})'$.
 In other terms, if
$\psi \in M_\open{I} \cap M_\open{J}$ then
$\Im m \ip{\psi}{ \phi} =0$ when either $\phi\in M_{\open{I}'}$ or $\phi \in M_{\open{J}'}$. In particular,
$\Im m \ip{\psi}{K(\Phi,\Pi)} =0$ when both the smooth real functions $\Phi, \Pi$ are supported in
$\open{I}'$ or in $\open{J}'$. Therefore the distributions\footnote{See the proof of lemma \ref{lem:central_lemma} where it was showed that those functionals are in fact distributions in $\mathcal{D}'(\bS^1)$.} $\smoothfuncv{\bS^1}{\bR} \ni
f \mapsto \ip{\Im m (\psi)}{A^{1/4} f}$
and $\smoothfuncv{\bS^1}{\bR} \ni
f \mapsto \ip{\Re e (\psi)}{A^{-1/4} f}$ have support included in $\bS^1 \setminus (\open{I}'\cup \open{J}') = (\bS^1 \setminus \open{I}') \cap (\bS^1 \setminus \open{J}')
= \overline{\open{I}} \cap \overline{\open{J}}$. Since $\open{I}$ and $\open{J}$ are disjoint proper open segments one has
$\overline{\open{I}} \cap \overline{\open{J}}=\partial \open{I} \cap \partial \open{J}$.
Therefore, if $\partial \open{I} \cap \partial \open{J}= \emptyset$ both distributions $\ip{\Re e (\psi)}{A^{-1/4} \cdot}$
and $\ip{\Im m (\psi)}{A^{1/4}  \cdot}$ vanish and this implies that $\psi =0$ since
$\overline{A^{\pm 1/4}(\smoothfunc{\bS^{1}})} = \sqint{\bS^1}{\theta}$ as proved in proposition \ref{prop:hyperbolic_generators}. Otherwise
$\partial I \cap \partial J$ contains two points at most, say $p$ and $q$.
We can assume, without loss of generality, that $\theta_p=0$ and $\theta_q \in (0,2\pi)$ (this situation can always be achieved by
redefining the origin of
coordinate $\theta$ on $\bS^1$). It is  a well-known result
of distributions theory that distributions with support given by a single point are
polynomials of derivatives of Dirac deltas supported on that point (the case of a finite number of points is a
trivial extension). Consider first of all $\ip{\Im m (\psi)}{A^{1/4}  f\cdot}$. In our case there must be a finite number
of coefficients $a_j,b_j \in \bR$ such that, for every $f \in \smoothfuncv{\bS^1}{\bR}$  it must hold
$$\ip{\Im m (\psi)}{A^{1/4} f} = \sum_{j=0}^{N_p} a_j \frac{d^j}{d\theta^j} f|_{p} +
\sum_{j=0}^{N_q} b_j \frac{d^j}{d\theta^j} f|_{q}\:.$$
Passing to Fourier transformation, the identity above can be re-written if $\psi_k$ and $f_k$ are the Fourier
coefficients of $\Im m (\psi)$ and $f$ respectively
$$\sum_{k\in \bZ}  \overline{\psi_k} (k^2+m^2)^{1/4} f_k = \sum_{k\in \bZ}  \left(
\sum_{j=0}^{N_p} a_j (\imath k)^j + \sum_{j=0}^{N_q} b_j (\imath k)^j e^{\imath k \theta_q}\right) f_k\:.$$
(notice that $f_k \to 0$ faster than every power $|k|^{-M}$ so that the right hand side is well defined).
Since the functions $f$ are dense in $\csqint{\bS^{1}}{\theta}$, this is equivalent to say that:
\begin{equation}\label{eq:AGG}
 \overline{\psi}_{k}:= (k^{2}+m^{2})^{-\frac{1}{4}}\left(\sum_{j=0}^{N} (a_{j} + e^{\imath k\theta_{q}} b_{j})(\imath k)^{j}\right)
\end{equation}
where we have defined $N:=\max(N_{p},N_{q})$ (assuming $a_j=0$ and $b_j=0$ for the added coefficients).
Let us prove that the right-hand side defines a $\ell^2(\bZ)$ sequence - as it is required by $\psi\in\csqint{\bS^1}{\theta}$ -  only if $a_j=0$ and $b_j=0$ for every $j$.
Assume that $\{{\psi}_{k}\}_k \in \ell^2(\bZ)$ so that the right-hand side of (\ref{eq:AGG}) defines a $\ell^2(\bZ)$ sequence.
If $c_{j,k}:= \Re e\left[ (a_{j} + e^{\imath k\theta_{q}} b_{j})\imath^{j}\right]$, then
\begin{equation}\label{eq:ReRaw}
(\Re e(\overline{\psi}_{k}))^{2}= \left( k^{2}+m^{2}\right)^{-\frac{1}{2}}\sum_{l,j=0}^{N} c_{j,k}\,c_{l,k}\, k^{l+j}\:.
\end{equation}
The sequence $\{k\theta_{q}\}_{k\in\mathbb{Z}}$ in $[0, 2\pi]$ may be either
 periodic -- and this happens when $\frac{\theta_{q}}{2\pi}$ is rational --
 or it is dense in $[0, 2\pi]$ -- and this arises for $\frac{\theta_{q}}{2\pi}$ irrational. In both
cases, fixing $k_{0}\in\mathbb{Z}\setminus \{0\}$ and $\epsilon > 0$, there is a sequence of integers
$\{k^{(\epsilon)}_{n}\}_{n\in\mathbb{Z}}$ such that:
\begin{equation*}             
|c_{N,k^{(\epsilon)}_{n}}-c_{N,k_{0}}|<\epsilon\:,\quad\forall\,\; n\in\mathbb{Z}\:.
\end{equation*}
Moreover, defining $M := \max_{j=0,\ldots, N} |a_j| +|b_j|$ one has $c_{j,k} \geq -M> -\infty$,
therefore a lower bound for the right-hand side of (\ref{eq:ReRaw}) is
\begin{equation*}\label{eq:ReIneq}
(\Re e \overline{\psi}_{k^{(\epsilon)}_{n}})^{2}\geq ((k^{(\epsilon)}_{n})^{2}+m^{2})^{-\frac{1}{2}}\left((c_{N,k_{0}}-
\mbox{sign}(c_{N,k_{0}})  \epsilon)^{2}\,
(k_{n}^{(\epsilon)})^{2N} -
 \sum_{l+j<2N} M^2\,|k^{(\epsilon)}_{n}|^{l+j}
\right).
\end{equation*}
If $c_{N,k_{0}} \neq 0$ the leading term in the latter expression is
\[((k_{n}^{(\epsilon)})^{2}+m^{2})^{-\frac{1}{2}}(c_{N,k_{0}}-\mbox{sign}(c_{N,k_{0}}) \epsilon)^{2}(k_{n}^{(\epsilon)})^{2N},\] 
so that the right-hand side of (\ref{eq:ReIneq})
diverges to $+\infty$ -- and
 $\{{\psi}_{k}\}_{k}\notin\ell^{2}(\mathbb{Z})$ - unless $c_{N,k_{0}}-\mbox{sign}(c_{N,k_{0}}) \epsilon=0$.
Arbitrariness of  $\epsilon$ implies $c_{N,k_{0}}=0$ that is
$\Re e \left[ (a_{N} + e^{\imath k_{0}\theta_{q}} b_{N})\imath^{j}\right]= 0$.
Analogously one sees that
$\Im m \left[ (a_{N} + e^{\imath k_{0}\theta_{q}} b_{N})\imath^{j}\right] = 0
$,
and thus
$a_{N} + e^{\imath k_{0}\theta_{q}} b_{N}=0$. However,
since  $k_{0}$ was arbitary one also has  $a_{N}=b_{N}=0$.
Iterating the procedure one achieves
$
a_{j}=b_{j}=0,\quad \forall\,\; j=N,N-1,\ldots,1\:.
$
It remains to consider the case $j=0$, that is the case of $\{\psi_k\}_{k\in \bZ} \in \ell^2(\bZ)$ with
\begin{equation*}
{\overline{\psi}}_{k}:= (k^{2}+m^{2})^{-\frac{1}{4}}\left(a_{0} + e^{\imath k\theta_{q}} b_{0}\right),
\end{equation*}
\noindent where $a_0,b_0\in \bR$ are constant. Now
\begin{equation*}
|{\psi}_{k}|^{2} =
|\overline{\psi}_{k}|^{2}=(k^{2}+m^{2})^{-\frac{1}{2}}|a_{0} +
 e^{\imath k\theta_{q}} b_{0}|^{2}\geq (k^{2}+m^{2})^{-\frac{1}{2}}||a_{0}|-|b_{0}||^{2},                                                           
\end{equation*}
and thus $\{\psi\}_{k}\notin\ell^{2}(\mathbb{Z})$ unless $b_{0}=\pm a_{0}$.
With that choice we have in turn:
\begin{equation*}
|{\psi}_{k}|^{2}=2a_{0}^{2}\,(k^{2}+m^{2})^{-\frac{1}{2}}(1\pm\cos(k\theta_{q})).                             
\end{equation*}
As the series $\sum_{k=0}^{\infty}\frac{\cos(k\theta_{q})}{k}$ converges ($\theta_{q}\neq 0 \mod{2\pi}$ by hypotheses),
and $(k^{2}+m^{2})^{-\frac{1}{2}} \sim \frac{1}{k}$ for $k\rightarrow\infty$, it arises that
 $\sum_{k=0}^{\infty}|{\psi}_{k}|^{2}$ diverges barring the case $a_{0}= b_0 = 0$. 
This concludes the proof of the fact that  $a_j = b_j=0$ for all $j$ if $\psi \in \sqint{\bS^1}{\theta}$.
We have found that the distribution  $\ip{\Im m (\psi)}{A^{1/4} \cdot}$ must vanish.
The proof for $\ip{\Re e (\psi)}{A^{-1/4}  \cdot}$ is strictly analogous. Since both distributions vanish
and $\overline{A^{\pm 1/4}(\smoothfunc{\bS^1}}) = \csqint{\bS^1}{\theta}$, we are commited to admit that $\psi =0$, so that
$M_{\open{I}} \cap M_{\open{J}} = \{0\}$. \\
Concerning the last statement, from (\ref{eq:LRT2}) one has
$\snet{R}{\open{I}} \cap \snet{R}{\open{J}} =  \alg{R}[M_\open{I}] \cap \alg{R}[M_\open{J}] = \alg{R}[M_\open{I}\cap M_\open{J}] = \alg{R}[\{0\}] = \bC\cdot 1$.
$\Box$

\chapter{States and characters.}
\section{Characters over $\Cdata{\sdc{C}}$.}
Starting from  the ground state $\omega_0$ it's possible to define other states making use of {\em characters} of $\Cdata{\sdc{C}}$.
If $U(1)$ is the Abelian multiplicative group of unitary complex numbers equipped with the topology induced by $\bC$,
an (algebraic) {\bf character} $\chi : \Cdata{\sdc{C}}\to U(1)$ is a group homomorphism, where $\Cdata{\sdc{C}}$ is viewed as an additive group.
Therefore if $\chi$ is a character, for all $(\Phi,\Pi), (\Phi',\Pi') \in \Cdata{\sdc{C}}$,
\begin{gather}\label{eq:char} 
\chi((\Phi,\Pi) + (\Phi',\Pi')) = \chi((\Phi,\Pi)) \chi((\Phi',\Pi')),\notag\\ 
\chi((0,0)) = 1 \:,\quad \chi(-(\Phi,\Pi))= \overline{\chi((\Phi,\Pi))}.
\end{gather}
The set of characters of
$\Cdata{\sdc{C}}$ will be denoted by $\Ch(\Cdata{\sdc{C}})$. $\Ch(\Cdata{\sdc{C}})$ is a commutative group with product given by pointwise product of functions.
Fixing the ground state $\omega_0$, $\Cdata{\sdc{C}}$ and every $\Cdata{\open{I}}$ become a real pre-Hilbert space if equipped with the scalar product:
\begin{equation}\label{eq:mu}
 \mu((\Phi,\Pi), (\Phi',\Pi')) := \Re e \ip{K(\Phi,\Pi)}{K(\Phi',\Pi') }.
\end{equation}
The $\mu$-Hilbert completion of $(\Cdata{\sdc{C}}, \mu)$ coincides with
$\overline{K(\Cdata{\sdc{C}})} \subset \hs{H}_0$. Similarly the $\mu$-Hilbert completion of $(\Cdata{\open{I}}, \mu)$ is
$\overline{K(\Cdata{\open{I}})} = M_\open{I} \subset \hs{H}_0$. \\

\begin{proposition}\label{prop:DMP}
Referring to $\Ch(\Cdata{\sdc{C}})$ the following facts hold.\\
{\bf (a)} $\chi \in \Ch(\Cdata{\sdc{C}})$ is continuous on $\Cdata{\sdc{C}}$ (resp. $\Cdata{\open{I}}$ for some $\open{I} \in \poset{R}$) w.r.t. $\mu$ if and only if
there is a linear continuous functional
$f_\chi : \Cdata{\sdc{C}} \to \bR$ (resp. $f_\chi : \Cdata{\open{I}} \to \bR$)
with:
$$\chi((\Phi,\Pi)) = e^{\imath f_\chi((\Phi,\Pi))}\:,\quad \forall\, (\Phi,\Pi) \in \Cdata{\sdc{C}}
\quad \mbox{(resp. $\Cdata{\open{I}}$).}$$
Such $f_\chi$ is uniquely individuated by $\chi$.\\
{\bf (b)}  The characters
\begin{eqnarray}
\chi_f((\Phi,\Pi)) &:=&  e^{\imath\int_{\bS^1} f(\theta) \Phi(\theta)\: d\theta}\:, \quad\forall\,(\Phi,\Pi) \in \Cdata{\sdc{C}} \mbox{(resp. $\Cdata{\open{I}}$),} \label{eq:chif}\\ 
\chi^{(g)}((\Phi,\Pi)) &:=&  e^{-\imath\int_{\bS^1} g(\theta) \Pi(\theta) \: d\theta}\:, \quad\forall\,(\Phi,\Pi) \in \Cdata{\sdc{C}} \mbox{(resp. $\Cdata{\open{I}}$),}  \label{eq:chig}
\end{eqnarray}
are continuous w.r.t. $\mu$ when $f\in \csqint{\bS^1}{\theta}$ and $g \in \Dom{A^{1/4}}$ are real-valued functions. 
\end{proposition}
\begin{proof}
 (a) The first statement has been proved in the appendix of \cite{DMP} in the general case where $\Cdata{\sdc{C}}$ is a topological vector space. Concerning the uniqueness of $f_\chi$, we notice that if there was
another such function $f_\chi'$, the difference $g:= f_\chi - f_\chi'$ would be a continuous linear function
$g: \Cdata{\sdc{C}} \to \bR$
(resp. $g: \Cdata{\open{I}} \to \bR$)
 whose range is a subset of $W:= \{2\pi k\:|\: k\in \bZ\}$. Since the range must be connected -
because of  continuity and connectedness of $\Cdata{\sdc{C}}$ (resp. $\Cdata{\open{I}}$) -
$\Ran{g}$ must be one of the components of $W$. Since $g(0)=0$ by linearity, the only possible component is $\{0\}$. Hence
$g \equiv 0$ and $f_\chi=f_\chi'$.\\
(b) Concerning (\ref{eq:chif}) one has, using $\Dom{A^{-1/4}} = \sqint{\bS^1}{\theta}$, self-adjointness of $A$, $f= \bar{f}$
 and the fact that $A^{-1/4}$ commutes with complex conjugation:
$$\int_{\bS^1} f \Phi\: d\theta = \int_{\bS^1} f A^{-1/4} A^{1/4}\Phi\: d\theta =
 \int_{\bS^1} \overline{A^{-1/4} f}  A^{1/4}\Phi\: d\theta = \ip{A^{-1/4} f}{ A^{1/4}\Phi}.$$
As a consequence:
\begin{equation}
\begin{split}
\left|\int_{\bS^1} f \Phi\: d\theta\right|^2  &\leq \norm{A^{-1/4} f}^2 \norm{A^{1/4}\Phi}^2 \leq\\
&\leq\norm{A^{-1/4} f}^2\left(\norm{A^{1/4} \Phi}^2  +  \norm{A^{-1/4} \Pi}^2\right) =\\
&= 2\norm{A^{-1/4} f}^2 \: \mu((\Phi,\Pi),(\Phi,\Pi))\\
\end{split}
\end{equation}
and so $\chi_f$ is continuous with respect to $\mu$.
The proof of (\ref{eq:chig}) is analogous.
\end{proof}

\subsection{States induced by characters.}  Every character $\chi$ individuates an associated pure state $\omega_\chi$ over $\alg{W}_{KG}$. Indeed, consider the unique Weyl-algebra isomorphism $\alpha_\chi: \alg{W}_{KG} \to \alg{W}_{KG}$ induced by linearity and continuity by:
$$\alpha_\chi : W(\Phi,\Pi) \mapsto \chi((\Phi,\Pi))\: W(\Phi,\Pi)\:,$$
(notice that Weyl relations are preserved in view of (\ref{eq:char})). The state $\omega_\chi$ is defined by
\begin{equation}\label{eq:omegachi}
\omega_\chi(a) := \omega_{0}(\alpha_\chi(a))\:, \quad \forall\, a \in \alg{W}_{KG}.
\end{equation}
 By construction one sees that a possible representation for the GNS triple of $\omega_\chi$ is
 $(\varhs{H}_0, \pi_\chi, \Psi_0)$, that is the same GNS triple as that for $\omega_0$ but with $\pi_\chi$ induced by
 \begin{equation}
 \pi_\chi(W(\Phi,\Pi)) := \chi((\Phi,\Pi)) \pi_0(W(\Phi,\Pi)) \label{eq:pichi}
 \end{equation}
Making use of the  GNS triple $(\varhs{H}_0, \pi_\chi, \Psi_0)$ for $\omega_\chi$
one immediately finds that  $\omega_\chi$ is pure since $\pi_0$ and thus $\pi_\chi$ is irreducible. Since
each local von Neumann algebra $\snet{R}{\open{I}}_\chi  := \pi_\chi\left(\sweyl{\open{I}}\right)''$ associated with $\omega_\chi$
coincides with $\snet{R}{\open{I}}$, associated  with the vacuum $\omega_0$, we have that  the isotonous, spatial local and irreducible class of von Neumann algebras 
$\{\snet{R}{\open{I}}_\chi \}_{\open{I} \in \poset{R}}$ satisfies {\em spatial weak additivity}, {\em spatial   additivity}, {\em Reeh-Schlieder property},
{\em spatial local
definiteness}, {\em spatial Haag duality}, {\em spatial punctured Haag duality} and {\em factoriality}. Further properties are those
established in the following theorem. 
 
\begin{theorem} \label{th:chiequiv}
 If $\chi_1, \chi_2 \in \Ch(\Cdata{\sdc{C}})$ and $\omega_{\chi_i} : \alg{W}_{KG} \to \bC$ denotes the pure states
  individuated by the two characters via (\ref{eq:omegachi}), with GNS triple $(\varhs{H}_0, \pi_{\chi_i}, \Psi_0)$, the following facts hold.\\
{\bf (a)} For $\open{I} \in \poset{R}$ there is an unitary operator $U_{\open{I}} : \varhs{H}_0 \to \varhs{H}_0$ with
   $$U_{\open{I}}\pi_{\chi_1}(a)  =  \pi_{\chi_2}(a) U_{\open{I}} \:,\quad \forall\, a \in \sweyl{\open{I}}$$
   if and only if the character $\chi_1\overline{\chi_2}$ is continuous on $\Cdata{\open{I}}$ w.r.t $\mu$.\\
{\bf (b)}   $\omega_{\chi_1}$ is unitarily equivalent to $\omega_{\chi_2}$
if and only if the character $\chi_1\overline{\chi_2}$
is continuous on $\Cdata{\sdc{C}}$ w.r.t. $\mu$. \\
{\bf (c)} The characters $\chi_f$ and $\chi^{(g)}$ defined in (\ref{eq:chif}) and (\ref{eq:chig})
for real-valued functions $f\in \sqint{\bS^1}{\theta}$ and $g \in \Dom{A^{1/4}}$ induce states $\omega_{\chi_f}$ and $\omega_{\chi^{(g)}}$
unitarily equivalent to $\omega_0$ such that:
\begin{eqnarray}
\pi_{\chi_f} (a) &=&  W\left[\imath\,2^{-1/2} A^{-1/4} f\right] \pi_{0}(a) W\left[\imath\,2^{-1/2} A^{-1/4} f\right]^*,\\ \label{eq:intert1}
 \pi_{\chi^{(g)}} (a)  &=& W\left[-2^{-1/2} A^{1/4} g\right] \pi_{0}(a)  W\left[-2^{-1/2} A^{1/4} g\right]^*, \label{eq:intert2}
\end{eqnarray}
\noindent for all $a\in \alg{W}_{KG}$.\\
{\bf (d)} $\omega_{\chi_1}$ is unitarily equivalent to $\omega_{\chi_2}$ if and only if it is locally unitarily equivalent to
 $\omega_{\chi_2}$.
\end{theorem}
\begin{proof}
 First of all we notice that in statements (a), (b) and (d) it is always possible to reduce to the simpler
 case where $\chi_2=1$ costantly. This is because
   $U_{\open{I}}\pi_{\chi_1}(a)  =  \pi_{\chi_2}(a) U_{\open{I}}$ for all $a \in \sweyl{\open{I}}$ is equivalent to
 $$U_{\open{I}}\pi_{\chi_3}(a)  =  \pi_{0}(a) U_{\open{I}}\:,\quad \forall\, a \in \sweyl{\open{I}} \quad \mbox{and where $\chi_3:=\chi_1\overline{\chi_2}$}\:.$$
 Therefore in the following we deal with the case $\chi_2=1$ and $\chi_1$ will be denoted by $\chi$.\\
\textbf{(a,b)} We only prove (a) since the proof of (b) is strictly analogous.
  Suppose that there is $U_{\open{I}}$ as in the hypotheses. Referring to the representation $\pi_\chi$
  defined in (\ref{eq:pichi}), there must be an unitary operator $U_{\open{I}}:\varhs{H}_0 \to \varhs{H}_0$ with
  $U_{\open{I}} \pi_\chi(a) U_{\open{I}}^*= \pi_0(a)$ for all $a \in \sweyl{\open{I}}$. In particular one has
  $\chi((\Phi,\Pi)) U_{\open{I}}\pi_0\left(W(\Phi,\Pi) \right)U_{\open{I}}^* = \pi_0\left(W(\Phi,\Pi) \right)$
 that is
 $$\chi((\Phi,\Pi)) = \ip{\Psi_0}{\pi_0\left(W(\Phi,\Pi) \right)U_{\open{I}}\pi_0\left(W(-(\Phi,\Pi)) \right)U_{\open{I}}^*\: \Psi_0}.$$
 By theorem \ref{th:vacuum},
 $\pi_0\left(W(-(\Phi,\Pi)) \right) =  e^{\overline{a(\psi) -  a^*(\psi)}}$
 where $\psi :=K(\Phi,\Pi)$. It is known \cite{Bratteli_Robinson_II} that $\hs{H}_0 \ni \psi \mapsto e^{\overline{a(\psi) -  a^*(\psi)}}$ is strongly continuous. If $\Cdata{\open{I}} \ni (\Phi_n,\Pi_n) \to 0$ in the topology induced by the scalar product $\mu$ (\ref{eq:mu}), the sequence
 $\psi_n := K(\Phi_n,\Pi_n)$ tends to $0$ in $\hs{H}_0$. Since
 $$\chi((\Phi,\Pi)) := \ip{\pi_0\left(W(-(\Phi,\Pi)) \right)\Psi_0}{U_{\open{I}}\pi_0\left(W(-(\Phi,\Pi)) \right)U_{\open{I}}^*
 \Psi_0},$$
\noindent we conclude that $\chi((\Phi_n,\Pi_n)) \to 1$ if $(\Phi_n,\Pi_n) \to 0$. From (\ref{eq:char})
 this is enough to establish  the continuity of $\chi$ in $\Cdata{\open{I}}$, viewed as an additive subgroup of $\Cdata{\sdc{C}}$.

 Now suppose that $\chi \in \Ch(\Cdata{\sdc{C}})$ is continuous on $\Cdata{\open{I}}$ and consider the state $\omega_\chi$. By proposition \ref{prop:DMP},
 there is $f_{\chi}: \Cdata{\open{I}} \to \bR$ linear and continuous such that
 $\chi((\Phi, \Pi)) = e^{\imath f_{\chi}((\Phi,\Pi))}$ for every $(\Phi,\Pi) \in \Cdata{\open{I}}$. Riesz theorem entails that there is $\psi_\chi$ in the
 completion of $(\mu, \Cdata{\open{I}})$, that is in $\overline{K(\Cdata{\open{I}})} = M_{\open{I}}$, with
$f_{\chi}((\Phi,\Pi)) = \Re e \ip{ \psi_\chi}{K(\Phi,\Pi)}$.
 Therefore, for $(\Phi,\Pi) \in \Cdata{\open{I}}$, one finds that it must hold
 $$\pi_\chi(W(\Phi, \Pi)) = e^{\imath  \Re e \ip{\psi_{\chi}}{K(\Phi,\Pi)}} \pi_0\left(W(\Phi,\Pi) \right) =
 e^{-\imath \Im m\ip{\imath \psi_\chi}{K(\Phi,\Pi)}} \pi_0\left(W(\Phi,\Pi) \right).$$
 \noindent Following (\ref{eq:Weylgen}) we define
$W[\imath\psi_\chi/2] := e^{\overline{a(\imath\psi_\chi/2) -  a^*(\imath\psi_\chi/2)}}$.
 Finally, making use of
 (\ref{eq:Ksigmaext}) and exploiting Weyl relations we find that, for every $(\Phi,\Pi) \in \Cdata{\sdc{C}}$:
 \begin{equation}\label{eq:riesz3} 
W[\imath\psi_\chi/2]\pi_0\left(W(\Phi,\Pi) \right) W[\imath\psi_\chi/2]^* =
 e^{-\imath\Im m \ip{\imath\psi_\chi}{K(\Phi,\Pi) }} \pi_0\left(W(\Phi,\Pi) \right).
\end{equation}
 The right-hand side is nothing but $\pi_\chi(W(\Phi, \Pi))$ provided that $(\Phi,\Pi) \in \Cdata{\open{I}}$.
 Hence restricting ourselves to $\Cdata{\open{I}}$, this identity proves that the unitary operator $U_{\open{I}}:= W[-\imath\psi_\chi/2]$ fulfills
 $U_{\open{I}} \pi_\chi(W(\Phi,\Pi)) = \pi_0 (W(\Phi,\Pi)) U_{\open{I}}$, for all $(\Phi,\Pi) \in \Cdata{\open{I}}$.
 Extending this equation by linearity and continuity we have the thesis.  \\
\textbf{(c)} The characters $\chi_f$ and $\chi^{(g)}$ are continuous as established in Proposition \ref{prop:DMP}, so
 the associated states are unitarily equivalent to $\omega_0$.
The explicit espression for the unitary intertwiners  appearing in (\ref{eq:intert1})
 and (\ref{eq:intert2}) is obtained making use of (\ref{eq:riesz3})
taking into account that it must hold, for $\chi_f$ and $\chi^{(g)}$ respectively,
$i\int  f \Phi\: d\theta = -\imath  \Im m\ip{\imath\psi_{\chi_f}}{K(\Phi,\Pi)}$ and
$-\imath\int  g \Pi \: d\theta = -\imath\Im m \ip{\imath\psi_{\chi^{(g)}}}{K(\Phi,\Pi)}$
for all $(\Phi,\Pi) \in \Cdata{\sdc{C}}$.\\
 \textbf{(d)} We have only to prove that local unitary equivalence implies unitary equivalence, the other implication being
trivial. Suppose that $\omega_\chi$ is  locally unitarily equivalent to $\omega_0$ and consider $\open{I},\open{J} \in \poset{R}$ with
$\open{I} \cup \open{J} = \bS^1$. Let $f_\open{I} : \Cdata{\open{I}} \to \bR$ and $f_\open{J} : \Cdata{\open{J}} \to \bR$ be the continuous
linear functions associated with $\chi\rest_{\Cdata{\open{I}}}$ and $\chi\rest_{\Cdata{\open{J}}}$ as established in proposition \ref{prop:DMP}.
Finally consider two functions $g_{\open{I}}, g_{\open{J}} \in \smoothfuncv{\bS^1}{\bR}$ with $g_{\open{I}}(\theta) + g_{\open{J}}(\theta) = 1$ for all $\theta \in  \bS^1$
and $\supp{g_{\open{I}}} \subset \open{I}$, $\supp{g_{\open{J}}} \subset\open{J}$. If $(\Phi, \Pi) \in \Cdata{\sdc{C}}$ one has:
$$\chi((\Phi, \Pi)) = \chi((g_{\open{I}}\Phi, g_{\open{I}}\Pi))\chi((g_{\open{J}}\Phi, g_{\open{J}}\Pi)) =
e^{\imath \left[ f_{\open{I}}((g_{\open{I}}\Phi, g_{\open{I}}\Pi) + f_{\open{J}}((g_{\open{J}}\Phi, g_{\open{J}}\Pi))\right]}\:.$$
To conclude the proof it is sufficient to establish that the maps $\Cdata{\sdc{C}} \ni (\Phi, \Pi) \mapsto f_{\open{I}}((g_{\open{I}}\Phi, g_{\open{I}}\Pi))$ and
$\Cdata{\sdc{C}} \ni (\Phi, \Pi) \mapsto f_{\open{J}}((g_{\open{J}}\Phi, g_{\open{J}}\Pi))$
are continuous. In fact, in this case
the linear function  $\Cdata{\sdc{C}} \ni (\Phi, \Pi) \mapsto f_{\open{I}}((g_{\open{I}}\Phi, g_{\open{I}}\Pi)) +  f_{\open{J}}((g_{\open{J}}\Phi, g_{\open{J}}\Pi))$ would be continuous and
thus the character $\chi$ would be such, therefore
  $\omega_\chi$ has to be unitarily equivalent to $\omega_0$ due to (a).
Since $f_{\open{I}}$ and $f_{\open{J}}$ are continuous, continuity of $\Cdata{\sdc{C}} \ni (\Phi, \Pi) \mapsto f_{\open{I}}((g_{\open{I}}\Phi, g_{\open{I}}\Pi))$ and
$\Cdata{\sdc{C}} \ni (\Phi, \Pi) \mapsto f_{\open{J}}((g_{\open{J}}\Phi, g_{\open{J}}\Pi))$  holds  if the multiplicative operator $\Cdata{\sdc{C}} \ni (\Phi, \Pi) \mapsto (h\Phi, h\Pi) \in \Cdata{\sdc{C}}$, for any fixed
 $h \in \smoothfuncv{\bS^1}{\bR}$ (in particular $h= g_{\open{I}},g_{\open{J}}$),
 is continuous with respect to the norm induced by the real scalar product $\mu$ on $\Cdata{\sdc{C}}$.
In turn, since $$\mu((\Phi,\Pi),(\Phi,\Pi)) = \norm{K((\Phi,\Pi))}^2 = \frac{1}{2}\norm{A^{1/4} \Phi}^2 +  \frac{1}{2}
\norm{A^{-1/4} \Pi}^2$$
that is equivalent to say that, if 
$h \in \smoothfuncv{\bS^1}{\bR}$
is used as a multiplicative operator in $\csqint{\bS^1}{\theta}$, there are constants $C_h, C_h'\geq 0$
with
 $\norm{A^{1/4} h f} \leq C_h\norm{A^{1/4} f}$
and $\norm{A^{-1/4} h f} \leq C'_h\norm{A^{-1/4} f}$ for all $f \in \smoothfuncv{\bS^1}{\bR}$.
Following the same proof as that of proposition \ref{prop:last}, but restricting to the case $\lambda=1$, one finds that
$A^{1/4} h A^{-1/4} $ and $ A^{-1/4} h A^{1/4}$ are well defined
on the dense spaces $\Dom{A^{-1/4}}$ and $\Dom{A^{1/4}}$
and continuously extend  to bounded operators over
$\sqint{\bS^1}{\theta}$ so that there are constants $C_h, C_h'\geq 0$ with
$\norm{A^{1/4} h A^{-1/4} \phi} \leq C_h\norm{\phi}$ and $\norm{A^{-1/4} h A^{1/4} \phi} \leq C_h'\norm{\phi}$. This happens in particular when
$\phi \in A^{1/4}\,\smoothfuncv{\bS^1}{\bR}$ and $\phi\in A^{-1/4}\,\smoothfuncv{\bS^1}{\bR}$ respectively since
$A^{\pm 1/4}\,\smoothfuncv{\bS^1}{\bR} \subset \smoothfuncv{\bS^1}{\bR} \subset \Dom{A^{\pm 1/4}}$ due to proposition \ref{prop:hyperbolic_generators}, hence
$\norm{A^{1/4} h f} \leq C_h\norm{A^{1/4} f}$ and $\norm{A^{-1/4} h  f} \leq C_h'\norm{A^{-1/4} f}$ for all $f \in \smoothfuncv{\bS^1}{\bR}$
as requested.
\end{proof}

\noindent\textbf{Remark}\\ 
\noindent{\textbf{(1)}} Validity of (d) for free QFT in the cylindric flat spacetime implies that there is {\em no straightforward generalization of  DHR superselection sectors} in our theory
referring to the states induced by characters $\chi \in \Ch(\Cdata{\sdc{C}})$.
This is because any state in a DHR sector different from the vacuum sector must be however locally unitarily equivalent
to the vacuum state. Using characters, sectors different from that containing the vacuum  may arise  in the occurrence of  breakdown of local unitarily equivalence only.
 In Minkowski spacetime the poset $\poset{R}$ is replaced by a suitable directed set of relatively compact regions of a fixed Cauchy
surface and there is no finite number of such regions whose union covers the Cauchy surface; thus the proof given for (c) does not apply.\\
{\bf (2)} The unitary intertwiners $W\left[\imath\,2^{-1/2} A^{-1/4} f\right]$ and $W\left[-2^{-1/2} A^{1/4} g\right]$  do not belong to $\pi_0(\alg{W}_{KG})$ in general but they
belong to $\pi_0(\alg{W}_{KG})''= \bop{\varhs{H}_0}$ in any case.

\chapter{Universal algebras}\label{app:universal_algebras}
Through this appendix we consider a class of $C^*$-algebras with unit $1$ in common,
 $\{\snet{A}{I}\}_{I\in \poset{I}}$, where $\poset{I}$
is a poset; we denote by $\subset$ the ordering relation in
$\poset{I}$. Assume that the class $\{\snet{A}{I}\}_{I\in \poset{I}}$ is isotonous, i.e.
$$\snet{A}{I} \subset \snet{A}{J} \quad \mbox{when $I\subset J$ }$$
for $I,J\in \poset{I}$, where $\snet{A}{I} \subset \snet{A}{J}$ means that the former is a sub $C^*$-algebra
of the latter. It is {\em not} assumed that  $\poset{I}$ is directed with respect to $\subset$ and thus one cannot define the inductive limit of the
class $\alg{A}$. However, as pointed out by Fredenhagen in \cite{Fredenhagen_90}, it is possible to give a sort of generalized  inductive limit
of the isotonous class of $C^*$-algebras $\{\snet{A}{I}\}_{I\in \poset{I}}$ which corresponds, in physical
applications, to the $C^*$-algebra of quasi local observables also in those
contexts where the set $\poset{I}$ is not directed. This is the case treated here, where $\poset{I}\equiv\poset{R}$ and $\snet{A}{I}\equiv\sweyl{\open{I}}$.

\begin{definition}\label{def:UA}
An unital $C^*$-algebra $\alg{A}$ is called an {\bf universal algebra} associated with $\{\snet{A}{I}\}_{I\in \poset{I}}$
if it fulfills the following properties.\\
\noindent {\bf (1)} $\alg{A}$ contains every $\snet{A}{I}$ as a $C^*$-subalgebra for all $I \in \poset{I}$ and coincides with the $C^*$-algebra generated by all of the subalgebras together\footnote{This requirement was not assumed in \cite{Fredenhagen_90} but it has been added in the subsequent \cite{Fredenhagen_Rehren_Schroer_92}; it's essential for the uniqueness of $\alg{A}$.},\\
\noindent {\bf (2)} if $\{\pi_I\}_{I\in \poset{I}}$ is a class of representations on $\bop{\hs{H}}$, for some fixed Hilbert space $\hs{H}$:
$$\pi_I : \snet{A}{I} \to \bop{\hs{H}},$$
satisfying  compatibility conditions
\begin{equation}\label{eq:comp} 
\pi_I\rest_{\snet{A}{J}} = \pi_J\quad \mbox{when $J\subset I$ },
\end{equation}
\noindent for $I,J\in \poset{I}$, then there is an {\em unique} representation $\pi: \alg{A} \to \bop{\hs{H}}$ such that:
\begin{equation}\label{eq:comp2} 
\pi\rest_{\snet{A}{I}} = \pi_I\quad \mbox{for every $I\in \poset{I}$}. 
\end{equation} 
\end{definition}
\noindent The next proposition establishes existence and uniqueness of the universal algebra, and shows that it truly extends the notion of inductive limit of a $C^{*}$-algebras net.
\begin{proposition}
With the given hypotheses on $\{\snet{A}{I}\}_{I\in \poset{I}}$, the following facts hold.\\
{\bf (a)}  $\{\snet{A}{I}\}_{I\in \poset{I}}$ admits an universal algebra $\alg{A}$.\\
{\bf (b)} The universal algebra is uniquely determined 
 up to $C^*$-algebras isomorphisms.\\
 {\bf (c)} If $(\poset{I}, \subset )$ is directed, $\alg{A}$ is isomorphic to the inductive limit of the net $\{\snet{A}{I}\}_{I\in \poset{I}}$. 
\end{proposition}
\begin{proof}
(a) The existence of an universal algebra $\alg{A}$ has been proved in \cite{Fredenhagen_90}. \\
(b) Consider two universal algebras $\alg{A}_1$ and $\alg{A}_2$ and (faithfully and isometrically) represent these $C^*$-algebras
in terms of subalgebras of $\bop{\hs{H}_1}$ and $\bop{\hs{H}_2}$ respectively, for suitable Hilbert spaces $\hs{H}_1$ and $\hs{H}_2$.
For $i=1,2$ the classes of embeddings $\{(\iota_I)_i\}_{I\in\poset{I}}$
 $(\iota_I)_i : \snet{A}{I} \to \alg{A}_i$ can be viewed as classes of representations $\{(\pi_I)_i\}_{I\in\poset{I}}$
valued on $\bop{\hs{H}_i}$. By construction both $\{(\pi_I)_1\}_{I\in\poset{I}}$ and $\{(\pi_I)_2\}_{I\in\poset{I}}$
fulfill separately the compatibility conditions (\ref{eq:comp}). Considering $\alg{A}_1$ as the universal algebra, property (2)
of the definition implies that there is  representation $\pi_{12}: \alg{A}_1 \to \bop{\hs{H}_2}$ such that
$$\pi_{12} \circ (\pi_I)_{1} = (\pi_I)_{2}\quad \forall\, I \in \poset{I}\:.$$
Interchanging the role of $\alg{A}_1$ and $\alg{A}_2$, one finds another representation $\pi_{21}:\alg{A}_2 \to \bop{\hs{H}_1}$ with
 $$\pi_{21} \circ (\pi_I)_{2} = (\pi_I)_{1}\quad \forall\, I \in \poset{I}\:.$$
 These two classes of identities together implies:
\begin{align}
(\pi_{21} \circ \pi_{12})\rest_{(\pi_I)_{1}(\snet{A}{I})} &= id_{(\pi_I)_1(\snet{A}{I})}\\
(\pi_{12} \circ \pi_{21})\rest_{(\pi_I)_2(\snet{A}{I})} &= id_{(\pi_I)_2(\snet{A}{I})} .
\end{align}
for all $I \in \poset{I}$. Then, using continuity of representations $\pi_{21}$
and  $\pi_{12}$ and closedness of their domains,
the identities above entail
that (i) $\pi_{21}$ includes $\pi_{12}(\alg{A}_{g1})$ in its  domain
and  $\pi_{12}$ includes $\pi_{21}(\alg{A}_{g2})$ in its  domain,
where $\alg{A}_{g1}$ and $\alg{A}_{g2}$ are the sub $C^*$-algebras of $\alg{A}_1$ and $\alg{A}_2$ respectively generated by all of
$\alg{A}_1(I)$ and all of $\alg{A}_2(I)$, and (ii)
${\pi_{21}} \circ \pi_{12}\rest_{\alg{A}_{g1}}  = id_{\alg{A}_{g1}}\:, \quad {\pi_{12}} \circ \pi_{21}
\rest_{\alg{A}_{g2}}  =
 id_{\alg{A}_{g2}}$.
 Since $\alg{A}_{gi} = \alg{A}_i$ we have actually obtained that:
 $${\pi_{21}} \circ \pi_{12} = id_{\alg{A}_{1}}\:,
 \quad {\pi_{12}} \circ \pi_{21} =
 id_{\alg{A}_{2}} $$
so that $\pi_{12}$ and $\pi_{21}$ are in fact $C^*$-algebra isomorphisms, and, in particular $\alg{A}_2 = \pi_{12}(\alg{A}_1)$.\\
(c) The inductive limit $\alg{A}$ is the completion of the $*$-algebra $\bigcup_{I\in \poset{I}} \snet{A}{I}$.
If $a\in \alg{A}$,
there must be a sequence $\{I_n\}_{n\in \bN} \subset \poset{I}$, with $I_i \subset I_k$  for $i \leq  k$,
such that $a_n \to a$ as $n\to +\infty$ and $a_n \in \alg{A}(I_n)$. if $\{\pi_I\}_{I\in \poset{I}}$ is a class
of representations on $\bop{\hs{H}}$, for some Hilbert space $\hs{H}$:
$$\pi_I : \alg{A}_I \to \bop{\hs{H}},$$
satisfying  compatibility conditions (\ref{eq:comp})
and $\pi$ is a representation (on $\bop{\hs{H}}$) of $\alg{A}$ which reduces to $\pi_I$ on every $\snet{A}{I}$, it holds, {\em remembering that
representations are norm decreasing and thus continuous}:
$$\pi(a) = \pi\left(\lim_{n \to +\infty} a_n\right) = \lim_{n \to +\infty} \pi(a_n) = \lim_{n \to +\infty} \pi_{I_n}(a_n)$$
 so that $\pi$ is completely individuated by the class of $\pi_I$. On the other hand, such a class of representations
individuates a representation $\pi$ of $\alg{A}$ by means of the same rule
(notice that, if $m\geq n$,
 $\norm{\pi_{I_n}(a_n) - \pi_{I_m}(a_m)} = \norm{\pi_{I_m}(a_n)- \pi_{I_m}(a_m)} \leq \norm{a_n-a_m}$
 so that $\{\pi_n(a_n)\}$ is Cauchy when $\{a_n\}$ is such). We have proved that the inductive limit is an universal algebra. 
\end{proof}
 
\noindent\textbf{Remark.} If $\alg{B}$ is a  unital sub $C^*$-algebra of an unital $C^*$-algebra $\alg{A}$ and every representation $\pi$ of $\alg{B}$ on some space of bounded operators on a Hilbert space admits an unique extension to $\alg{A}$, it is anyway possible that $\alg{B} \subsetneq \alg{A}$: it is sufficient that $\alg{B}$ includes a closed two-sided ideal of $\alg{A}$ (see \cite{Dixmier_69}). Therefore the requirement that the sub algebras $\snet{A}{I}$ generate $\alg{A}$ is essential in proving the uniqueness of the universal algebra $\alg{A}$.\\
  
\noindent As an example consider the theory on $\bS^1$ studied earlier and focus on the class of unital
$C^*$-algebras (Weyl algebras) $\{\sweyl{\open{I}}\}_{\open{I} \in \poset{R}}$: it's simply proved that $\alg{W}_{KG}$ is the associated universal algebra. \\
\begin{proposition}
 $\alg{W}_{KG}$ is the  universal algebra for $\{\sweyl{\open{I}}\}_{I \in \poset{R}}$.
\end{proposition}

\begin{proof}
 Condition (1) in definition \ref{def:UA} is trivially fulfilled.
Then consider a class of representations $\{\pi_\open{I}\}_{\open{I}\in \poset{R}}$ on $\bop{\hs{H}}$, for some Hilbert space $\hs{H}$
satisfying  compatibility conditions (\ref{eq:comp}). Suppose that there is $\pi: \alg{W}_{KG} \to \bop{\hs{H}}$ satisfying (\ref{eq:comp2}).
Fix $\open{I},\open{J} \in \poset{R}$ with $\open{I} \cup \open{J} = \bS^1$ and $f,g \in \smoothfuncv{\bS^1}{\bR}$
with $f+g =1$ and $\supp{f} \subset \open{I}$, $\supp{g} \subset \open{J}$. For $(\Pi,\Phi) \in \Cdata{\sdc{C}}$ one has, if  $h(\Phi,\Pi)$ denotes the couple $(h\cdot \Phi, h \cdot \Pi)$:
\begin{equation}
\begin{split}
\pi\left(W(\Phi,\Pi)\right) &= \pi\left(W(f(\Phi,\Pi) + g(\Phi,\Pi))\right)=\\ 
&= \pi\left(W(f(\Phi,\Pi)\right)\pi\left(W(g(\Phi,\Pi))\right) e^{-\imath\sigma(f(\Phi,\Pi), g(\Phi,\Pi))/2}\\
\end{split}
\end{equation}
 We have found that:
$$\pi\left(W(\Phi,\Pi)\right) =e^{-\imath\sigma(f(\Phi,\Pi), g(\Phi,\Pi))/2} \pi_\open{I}\left(W(f(\Phi,\Pi)\right)
 \pi_\open{J}\left(g(\Phi,\Pi))\right)
$$
Incidentally, by direct inspection, one finds that $\sigma(f(\Phi,\Pi), g(\Phi,\Pi))=0$  also if $f\cdot g \neq 0$. Therefore
\begin{equation} \pi\left(W(\Phi,\Pi)\right) = \pi_\open{I}\left(W(f(\Phi,\Pi)\right)
 \pi_\open{J}\left(W(g(\Phi,\Pi))\right) \:. \label{eq:coinc}
\end{equation}
The right-hand side does not depend on $\pi$. Since every element of $\alg{W}_{KG}$ is obtained by linearity and continuity
from generators $W(\Phi,\Pi)$ and representations are continuous, we conclude that every representation of $\alg{W}_{KG}$ satisfying
(\ref{eq:comp2}) must coincide with $\pi$ due to (\ref{eq:coinc}). Now we prove that
 $\{\pi_\open{I}\}_{\open{I}\in \poset{R}}$
satisfying  compatibility conditions (\ref{eq:comp}) individuates a representation $\pi$ fulfilling (\ref{eq:comp2}).
First of all suppose that there is $(\Phi,\Pi)$ supported in some $\open{I}\in \poset{R}$ with
$\pi_\open{I}\left(W(\Phi,\Pi)\right)=0$.
Using Weyl relations, for every $\open{J} \in \poset{R}$ such that there is $\open{K}\in \poset{R}$ with $\open{K} \supset \open{I},\open{J}$:
\begin{equation}
\begin{split}
\pi_\open{J}(W(\Phi',\Pi')) &=
 \pi_\open{K}\left(W(\Phi',\Pi')\right) =\\
&= c\,\pi_\open{K}\left(W(\Phi'- \Phi,\Pi'-\Pi)\right)\pi_\open{K}\left(W(\Phi,\Pi)\right) =0
\end{split}
\end{equation}
whenever $(\Phi',\Pi') \in \Cdata{\open{J}}$, $c \in \bC$ being the appropriate exponential arising by Weyl relations. Taking two such $\open{J}$ one easily concludes that
$\pi_\open{L}\left(W(\Phi',\Pi')\right)=0$ for all $\open{L}\in \poset{R}$ and $(\Phi',\Pi') \in\Cdata{\open{L}}$.
Therefore, by continuity all representations $\pi_\open{I}$ are degenerate.
A representation $\pi$ fulfilling (\ref{eq:comp2}) in this case is the degenerate one $\pi(a)=0$ for all $a \in \alg{W}_{KG}$.

Now consider the case where $\pi_\open{L}\left(W(\Phi,\Pi)\right) \neq 0$ for all $(\Phi,\Pi)\in\Cdata{\open{L}}$ and $\open{L}\in \poset{R}$.
Fix $\open{I},\open{J}\in \poset{R}$ with $\open{I}\cup \open{J} = \bS^1$ and $f,g \in \smoothfuncv{\bS^{1}}{\bR}$
with $f+g =1$ and $\supp{f}  \subset \open{I}$, $\supp{g}  \subset \open{J}$.
  For $(\Pi,\Phi) \in\Cdata{\sdc{C}}$ define
  $$\pi\left(W(\Phi,\Pi)\right) :=e^{-\imath\sigma(f(\Phi,\Pi), g(\Phi,\Pi))/2} \pi_\open{I}\left(W(f(\Phi,\Pi)\right)
 \pi_\open{J}\left(W(g(\Phi,\Pi))\right)
$$
The right-hand side cannot vanish because all the factors appearing  therein  are invertible by construction.
Making use of (\ref{eq:comp}), it is simply proved that, for every fixed $\open{K} \in \poset{R}$
\begin{equation} \label{eq:QQQ}
\pi\left(W(\Phi,\Pi)\right) = \pi_\open{K}\left(W(\Phi,\Pi)\right) \quad \mbox{for all $(\Phi,\Pi) \in \Cdata{\open{K}}$.} 
\end{equation}
By direct inspection, using Weyl relations one verifies that the nonvanishing operators
$\pi\left(W(\Phi,\Pi)\right)$ fulfills Weyl relations for every $W(\Phi,\Pi) \in \alg{W}_{KG}$.
Finally consider the sub $C^*$-algebra $\hat{\alg{W}}$ generated in $\bop{\hs{H}}$ by the generators
$\pi\left(W(\Phi,\Pi)\right)$.
As is well-known (\cite{Bratteli_Robinson_II}) there is a faithful representation $\pi$ of $\alg{W}_{KG}$ onto $\hat{\alg{W}}$
(notice that the unit of $\alg{W}_{KG}$ is in general represented by an orthogonal projector in $\bop{\hs{H}}$) which uniquely
extends the map $W(\Phi,\Pi) \mapsto \pi\left(W(\Phi,\Pi)\right)$ by linearity and continuity.
By construction (\ref{eq:comp2}) is fulfilled by $\pi$ due to (\ref{eq:QQQ}).
\end{proof}

\chapter{Miscellaneous results}\label{app:misc}

In this appendix we collect some definitions and results cited en passant in the course of our exposition. First of all we recall some basic definitions from the theory of von Neumann algebras. 

\begin{definition}
Let $\alg{L}$ be a von Neumann algebra on a Hilbert space $\hs{H}$; the center of $\alg{L}$ is the subalgebra $\alg{L}\cap\alg{L}^{'}$.\\
\noindent \textit{i)} $\alg{L}$ is said to be a factor if it has a trivial center, i.e. if $\alg{L}\cap\alg{L}^{'} = \bC\cdot 1$;\\
\textit{ii)} A factor is said to be of \emph{type I} if it contains a minimal (non-zero) projection;\\
\textit{iii)} Two projections $E,F\in\alg{L}$ are said to be equivalent (in $\alg{L}$) if there exists $W\in\alg{L}$ such that $E = W^{*}W$ and $F = WW^{*}$; we write $E\sim F$.\\
\end{definition}

Type I factors are in fact isomorphic to $\bop{\hs{H}}$, for some Hilbert space $\hs{H}$ (see e.g. \cite[Prop. 2.7.19]{Bratteli_Robinson_I}).

Now comes some material about inclusions of von Neumann algebras \cite{Roberts_04}.

\begin{definition}\label{def:vNa_inclusion}
Let $\alg{L}$, $\alg{M}$, $\alg{N}$ be von Neumann algebras on a fixed Hilbert space $\hs{H}$; then:

\noindent \textit{i}) An inclusion $\alg{L}\subset\alg{N}$ is said to have \emph{property B} if every non-zero projection $E$ in $\alg{L}$ is equivalent to $1$ in $\alg{N}$;\\
\textit{ii)} An inclusion $\alg{L}\subset\alg{N}$ is said to be \emph{split} if there is an intermediate type $I$ factor $\alg{M}$, i.e. if $\alg{L}\subset\alg{M}\subset\alg{N}$;\\
\textit{iii)} A triple $\Lambda\equiv\left\lbrace\alg{L},\alg{N},\Omega\right\rbrace$ is a \emph{standard split inclusion} if  the inclusion $\alg{L}\subset\alg{N}$ is split and $\Omega\in\hs{H}$ is a joint cyclic and separating vector for $\alg{L}$, $\alg{N}$ and $\alg{L}^{'}\cap\alg{N}$;\\
\textit{iv)} A net of von Neumann algebras $\net{A}{K}$ has the \emph{Borchers property} if given $O\in\poset{K}$ there is $O_{1}\subset O$ such that the inclusion $\snet{A}{O_{1}}\subset\snet{A}{O}$ has the property $B$.
\end{definition}

We need a basic property of standard split inclusions.
\begin{proposition}\label{prop:standard_split}
A standard split inclusion has property $B$.
\end{proposition}
\begin{proof}
See e.g. \cite{Roberts_04}.
\end{proof}

We conclude with the definition of a $C^*$-category.

\begin{definition}[$C^{*}$- category]
A category $\category{C}$ is said to be a $C^*$-category if the following conditions are satisfied:
\begin{itemize}
 \item The set of arrows $(z,z_1)$ between two elements $z, z_1\in\category{C}$ has  a complex vector space structure $\left\lbrace (z,z_1), +, \cdot\right\rbrace$; composition between arrows is bilinear with respect to it;
\item The space $(z,z_1)$ admits a complete norm $\norm{\cdot}$ satisfying the $C^*$-property (namely $\norm{r^*r} = \norm{r}^2$ for each $r\in (z,z_1)$), making it into a complex Banach space;
\item There should be an adjoint $*$, that is an involutive controvariant functor acting as the identity on the objects.
\end{itemize}
As a consequence, if $\category{C}$ is a $C^*$-category then $(z,z)$ is a $C^*$-algebra for every $z\in\category{C}$.
\end{definition}

% \nocite{*}
% \bibliographystyle{plain} 
% \bibliography{QFT}

\addcontentsline{toc}{chapter}{Bibliography}

\begin{thebibliography}{References}

\bibitem{Aharonov_Bohm_59} Y. Aharonov D. Bohm \emph{``Significance of electromagnetic potentials in quantum theory''}, Phys. Rev. \textbf{115}, (1959), 485-491.

\bibitem{Araki_99} H.~Araki, {\em  ``Mathematical Theory of Quantum Fields,''} Oxford University press,
Oxford, (1999)
 
\bibitem{Araki2} H.~Araki, {\em ``Von Neumann algebras of local observables for free scalar field,''} 
  J. Math. Phys.  {\bf 5}, 1-13, (1964)\\
  H.~Araki, {\em ``A lattice of von Neumann algebras associated with the quantum theory of a free Bose field,''}  
  J. Math. Phys.  {\bf 4}, 1343-1362, (1963)

\bibitem{Bar_Ginoux_Pfaffle_06} C. B\"{a}r , N. Ginoux, F. Pf\"{a}ffle	\emph{``Wave equations on Lorentzian manifolds and quantization''}. 
 
\bibitem{Bernal_Sánchez_06} A.N. Bernal, M. Sánchez \emph{``Further results on the smoothability of Cauchy hypersurfaces and Cauchy time functions''}, Lett. Math. Phys. 77, no.2, (2006) 183–197.

\bibitem{Borchers} H. J. Borchers, {\em ``A remark on a Theorem of B. Misra,"} Commun. Math. Phys. {\bf 4}, 315-323 (1967)
 
\bibitem{Bratteli_Robinson_I} O.~Bratteli, D.~W.~Robinson,
{\em ``Operator algebras and quantum statistical mechanics,''} Vol.1,
Springer Berlin, Germany,  (1996)
 
 
\bibitem{Bratteli_Robinson_II} O.~Bratteli, D.~W.~Robinson,
{\em ``Operator algebras and quantum statistical mechanics,''} Vol. 2,
Springer Berlin, Germany,  (1996)

\bibitem{BF} R. Brunetti, K. Fredenhagen, 
{\em ``Microlocal analysis and interacting quantum field theories: renormalization on physical backgrounds,''}
 Commun. Math. Phys. {\bf 208}, 623--661, (2000)

\bibitem{BFV} R. Brunetti, K. Fredenhagen and R. Verch, {\em ``The generally covariant locality principle---a new paradigm for local quantum field theory,''}
 Commun. Math. Phys. {\bf 237}, 31-68,  (2003)

\bibitem{BGL} R. Brunetti, D. Guido and R. Longo, 
{\em ``Modular structure and duality in conformal quantum field theory,"} Commun. Math. Phys. {\bf 156}, 201-219, (1993)

\bibitem{Brunetti_Ruzzi_07} R. Brunetti and G. Ruzzi, {\em ``Superselection sectors and general covariance. I,''} 
Commun. Math. Phys. {\bf 270}, 69-108, (2007)

\bibitem{Brunetti_Ruzzi_08} R. Brunetti and G. Ruzzi, {\em ``Quantum charges and topology: The emergence of new superselection sectors,''}
 Commun. Math. Phys. (2008) in print. (Published online: DOI 10.1007/s00220-008-0671-6)

\bibitem{BS} D. Buchholz and J. Schlemmer, {\em ``Local Temperature in Curved Spacetime,"}  Classical Quantum Gravity 
 {\bf 24},  F25-F31 (2007)

\bibitem{BSu} D. Buchholz, O. Dreyer, M. Florig and S. J. Summers,
 {\em ``Geometric modular action and spacetime symmetry groups,"} Rev. Math. Phys. {\bf 12 }, 475-560 (2000)

\bibitem{BW} D. Buchholz and E. H. Wichmann, {\em ``Causal independence and the energy-level density of states in local quantum field theory,"}
Commun. Math. Phys. {\bf 106}, 321-344, (1986)

\bibitem{Chernoff_73} P.R. Chernoff, {\em ``Essential Self-Adjointness of Powers of Generators of Hyperbolic Equations,''}
 J. Functional Analysis  {\bf 12},   401-414 (1973)
 
\bibitem{Ciolli} F. Ciolli, {\em ``Massless scalar free Field in 1+1 dimensions I: Weyl algebras Products
and Superselection Sectors''}, arXiv:math-ph/0511064v3;
{\em ``Massless scalar free field in 1+1 dimensions II: Net cohomology and completeness
of superselection sectors''}, arXiv:0811.4673v1 [math-ph].


\bibitem{Cordes_87} H. O. Cordes, {\em ``Spectral Theory of Linear Differential Operators and Comparison Algebras,''}
Lecture Notes Series 76, Cambridge University Press,  London, (1987)

\bibitem{D'Antoni_Longo_Radulescu_98} C. D'Antoni, R. Longo and F. Radulescu, 
{\em ``Conformal nets, maximal temperature and models from free probability,"}  J. Operator Theory  {\bf 45}, 195-208 (2001)

\bibitem{DMP} C. Dappiaggi, V. Moretti and N. Pinamonti,
{\em ``Rigorous steps towards holography in asymptotically flat spacetimes,''}  
Rev. Math. Phys.  {\bf 18}, 349--415, (2006)

\bibitem{DMP2} C. Dappiaggi, V. Moretti and N. Pinamonti,
{\em ``Cosmological Horizons and Reconstruction of Quantum Field Theories,''} 
Commun. Math. Phys.  (2008) in print (Published online: DOI 10.1007/s00220-008-0653-8).

\bibitem{Dimock_80} J. Dimock \emph{``Algebras of local observables on a manifold''}, Commun. Math. Phys. 77 (1980).

\bibitem{Dixmier_69} J. Dixmier, \emph{``Les C*-algebras et leurs représentations''}, Chaier Scientifiques, Fascicule XXIX, Gautthier-Villars, Paris 1969. 

\bibitem{DL} S. Doplicher and R. Longo, {\em ``Standard and split inclusions of von Neumann algebras,"} Inv. Math. {\bf  
75}, 493-536 (1984)

\bibitem{DHR_69_I} S. Doplicher, R. Haag, J.E. Roberts,  \emph{``Fields, observables and gauge transformations I''}, Commun. Math. Phys. \textbf{13}, (1969)

\bibitem{DHR_69_II} S. Doplicher, R. Haag, J.E. Roberts,  \emph{``Fields, observables and gauge transformations II''}, Commun. Math. Phys. \textbf{15}, (1969)

\bibitem{DHR_71} S. Doplicher, R. Haag, J.E. Roberts,  \emph{``Local observables and particle statistics I''}, Commun. Math. Phys. \textbf{23}, (1971), 199-230.

\bibitem{DHR_74} S. Doplicher, R. Haag, J.E. Roberts,  \emph{``Local observables and particle statistics II''}, Commun. Math. Phys. \textbf{35}, (1974)



\bibitem{Driessler} W. Driessler, {\em ``Duality and absence of locally generated superselection sectors for CCR-type algebras,''} Commun. Math. Phys. {\bf 70}, 213-220 (1979)

\bibitem{F} C. Fewster,
{\em ``Quantum energy inequalities and stability conditions in quantum field theory.  Rigorous quantum field theory,''}  
 Progr. Math., {\bf 251}, 95--111, Birkhäuser, Basel, (2007)\\
  C. Fewster,
 {\em ``Quantum energy inequalities and local covariance. II. Categorical formulation,''}  Gen. Relativity Gravitation 
  {\bf 39},
  1855--1890, (2007)

\bibitem{Fredenhagen_90}
K. Fredenhagen,
{\em ``Generalizations of the theory of superselection sectors,''}
published in {\em The algebraic theory of superselection sectors},
 (Palermo, 1989), 379-387, World Sci. Publ., River Edge, NJ, 1990

\bibitem{Fredenhagen_Rehren_Schroer_92}
K. Fredenhagen, K.-H. Rehren and B. Schroer,
{\em ``Superselection sectors with braid group statistics and exchange algebras. II. Geometric aspects and conformal covariance,''}
Rev. Math. Phys.  Special Issue, 113-157 (1992)

\bibitem{Guido_Longo_92} Guido D., Longo R. \emph{``Relativistic invariance and charge conjugation in quantum field theory''}, Commun. Math. Phys. 148 (1992), 521-551. 

\bibitem{GLRV} D. Guido, R. Longo, J.E. Roberts and R. Verch,
{\em ``Charged sectors, spin and statistics in quantum field theory on curved spacetimes,''}
 Rev. Math. Phys.  {\bf 13},  125--198,  (2001)


\bibitem{Haag_96} R. Haag, {\em ``Local Quantum Physics,''} second revised edition, 
Springer, Berlin (1996)

\bibitem{Haag_Kastler_64} R. Haag, D. Kastler \emph{``An Algebraic Approach to Quantum Field Theory''}, J. Math. Phys. 5, n. 7 (1964), 848-871.
 
\bibitem{Hannay_85} J.H. Hannay \emph{``Angle variable anholonomy in adiabatic excursions of an integrable Hamiltonian''}, J. Phys. A: Math. Gen. \textbf{18}, (1985), 221-230.

\bibitem{Hol}  S. Hollands {\em ``The Operator product expansion for perturbative quantum 
field theory in curved spacetime,''} Commun. Math. Phys. {\bf 273}, 1-36, (2007)

 \bibitem{HW1} S. Hollands and R.M. Wald, {\em ``Local Wick polynomials and time ordered products of quantum fields in curved spacetime,''}
 Comm. Math. Phys.  {\bf 223}, 289--326, (2001)
 
 \bibitem{HW2} S. Hollands and R.M. Wald, {\em ``Existence of local covariant time ordered products of quantum field in curved spacetime,''}  
 Comm. Math. Phys.  {\bf 231},  309--345, (2002)
 
 \bibitem{HW3} S. Hollands and R.M. Wald, {\em ``On the renormalization group in curved spacetime,''}  
 Comm. Math. Phys.  {\bf 237},  123--160, (2003)
  
 \bibitem{HW4} S. Hollands and R.M. Wald, {\em ``Conservation of the stress tensor in perturbative interacting quantum field theory in curved spacetimes,''}  
  Rev. Math. Phys.  {\bf 17}, 227--311, (2005)
 
 \bibitem{HW_OPE}   S. Hollands and R.M. Wald,  {\em ``Axiomatic quantum field theory in curved spacetime,''}
 available online at: arXiv:0803.2003 [gr-qc]
 
\bibitem{Kay} B.S. Kay, {\em ``Linear Spin-Zero Quantum Fields in External Gravitational and Scalar Fields,''}
 Commun. \ Math. \ Phys {\bf 62}, 55-70 (1978)
 
\bibitem{Kay_Wald_91}
B.~S.~Kay and R.~M.~Wald,
{\em ``Theorem on the uniqueness and thermal properties of stationary, nonsingular, quasifree states on
space-times with a bifurcate Killing
horizon,''}
Phys.\ Rept.\  {\bf 207}, 49 (1991)

\bibitem{Kobayashi_Nomizu_96} S. Kobayashi; K. Nomizu, {\em ``Foundations of differential geometry,''} Vol. I and II.
 Wiley Classics Library. A Wiley-Interscience Publication. John Wiley \& Sons, Inc., New York, (1996)

\bibitem{Leyland_Roberts_Testard_78}P.~Leyland, J.~Roberts and D.~Testard, {\em ``Duality for Quantum Free Fields,''} CPT-78/P-1016, Jul. 1978\\
available on-line at http://ccdb4fs.kek.jp/cgi-bin/img$\underline{\:\:}$index?7901157

\bibitem{M03} V. Moretti, {\em ``Comments on the Stress-Energy Tensor Operator in Curved Spacetime,''}
 Commun. \ Math. \ Phys {\bf 232}, 189-201, (2003)

\bibitem{M08} V. Moretti,  {\em ``Quantum out-states holographically induced by asymptotic flatness: 
invariance under spacetime symmetries, energy positivity and Hadamard property,''}
 Comm. Math. Phys.  {\bf 279}, 31--75, (2008)

\bibitem{MM} M. M\"uger, {\em ``Superselection structure of massive quantum field theories in $1+1$ dimensions,''}  
Rev. Math. Phys.  {\bf 10}, 1147--1170,  (1998)

\bibitem{O'Neill_83} B. O'Neill, {\em ``Semi-Riemannian Geometry With Applications to Relativity''}, Academic Press (1983) 

\bibitem{Osterwalder} K. Osterwalder, {\em ``Duality for free Bose fields,''}  Commun. Math. Phys.  {\bf 29},  1-14 (1973)

\bibitem{Roberts1} J. E. Roberts, {\em ``Local cohomology and superselection structure,"} Commun. Math. Phys. {\bf 51}, 107-119 (1976)
 
\bibitem{Roberts_04}  J. E. Roberts, {\em ``More lectures on algebraic quantum field theory,''} 263--342, in 
{\em ``Noncommutative geometry''} Edited by S. Doplicher and R. Longo., Lecture Notes in Math., 1831, Springer, Berlin, 2004

\bibitem{Reed_Simon_II} M. Reed, B. Simon, {\em ``Methods of Modern Mathematical Phisics''} II Fourier Analysis, Self-Adjointness,
Academic Press, Boston, USA, (1975)

\bibitem{ReehSchlieder}
H. Reeh and S. Schlieder,  {\em `` "{U}ber den Zerfall der Feldoperatoralgebra im Falle einer Vakuumentartung,''} 
 Nuovo Cimento (10),  {\bf 26},   32--42, (1962)\\
 H. Reeh and S. Schlieder,  {\em ``Bemerkungen zur Unit\"{a}r\"{a}quivalenz von Lorentzinvarianten Felden,''} 
 Nuovo Cimento (10),  {\bf 22},   1051--1068, (1961)

\bibitem{RuzziRMP} G. Ruzzi,
{\em ``Homotopy of posets, net-cohomology and superselection sectors in globally hyperbolic space-times,''}  
Rev. Math. Phys.  {\bf 17}, 1021--1070,  (2005).

\bibitem{Ruzzi_05} G.Ruzzi \emph{``Punctured Haag Duality in Locally Covariant Quantum Field Theory''}, Commun. Math. Phys. 256 (2005), 621-634.

\bibitem{SV} J. Schlemmer and R. Verch, {\em ``Local Thermal Equilibrium States and Quantum Energy Inequalities,"} 
available online at: arXiv:0802.2151v1 [gr-qc]

\bibitem{Strocchi_05} F. Strocchi, \emph{An introduction to the Mathematical Structure of Quantum Mechanics},  Advanced Series in Mathematical Physics Vol. 27, World Scientific, Singapore,  2005 

\bibitem{Strohmaier_00r} A. Strohmaier, {\em ``The Reeh-Schlieder property for quantum fields on stationary spacetimes,''}
  Comm. Math. Phys.  {\bf 215}, 105--118,  (2000)

\bibitem{T} M. Takesaki, {\em ``Theory of operator algebras,''}
vol. I-III  Encyclopaedia of Mathematical Sciences, 124. 
Operator Algebras and Non-commutative Geometry, 5. Springer-Verlag, Berlin, (2002)

\bibitem{Verch_93} R. Verch,  {\em``Local definiteness,
 primarity and quasiequivalence of quasifree Hadamard quantum states in curved spacetime,''}  
 Comm. Math. Phys.  {\bf 160}, 507--536, (1994)

\bibitem{Verch2} R. Verch, {\em``Continuity of symplectically adjoint maps and the algebraic structure
 of Hadamard vacuum representations for quantum fields on curved spacetime,''}  
 Rev. Math. Phys.  {\bf 9}, 635--674,  (1997)
  
\bibitem{Wald_79} R. M. Wald, {\em ``On the Euclidean approach to quantum field theory in curved spacetime,''} Commun. Math. Phys.  {\bf 70}, 221-242. (1979)

\bibitem{Wald_94} R.M. Wald \emph{``Quantum Field Theory in Curved Spacetime and Black Hole Thermodynamics''}, The University of Chicago Press, 1994.

\end{thebibliography}

\end{document}